\documentclass[a4paper,11pt]{article}

\usepackage{amsthm}
\usepackage{fullpage}%
\usepackage{graphicx,tikz,enumerate}
\usepackage[colorlinks=true,citecolor=blue,linkcolor=red]{hyperref}
\usepackage[title]{appendix}
\usepackage{amssymb, amsfonts,mathtools,stmaryrd}
\usepackage{xcolor}
\usepackage{enumitem}
\usepackage{etoolbox}

\usepackage[margin=0.97in]{geometry}

\usepackage{tikz}
\usetikzlibrary{arrows.meta,positioning,shapes.geometric}

\newenvironment{example}{
	\par\medskip
	\noindent\textbf{Example.}\ \normalfont
}{
	\par\medskip
}

\usepackage{mathtools}

\usepackage{algorithm}
\usepackage[noend]{algpseudocode}

\usepackage{dsfont}
\usepackage{bbm}
\newtheorem{definition}{\textbf{Definition}}
\newtheorem{theorem}[definition]{\textbf{Theorem}}
\newtheorem{lemma}[definition]{\textbf{Lemma}}
\newtheorem{corollary}[definition]{\textbf{Corollary}}
\newtheorem{proposition}[definition]{\textbf{Proposition}}
\newtheorem*{proof*}{\textbf{Proof}}
\newcommand{\Dist}{\ensuremath{\mathcal{D}}}
\newcommand{\Supp}{\mathsf{Supp} }
\newcommand{\Prb}{\ensuremath{\mathbb{P}}}
\newcommand{\Exp}{\ensuremath{\mathbb{E}}}
\newcommand{\MDP}{\ensuremath{\mathsf{G}}}
\newcommand{\MEMDP}{\ensuremath{\Gamma}}
\newcommand{\POMDP}{\ensuremath{\Lambda}}
\newcommand{\cyl}{\ensuremath{\msf{cyl}}}

\newcommand{\N}{\mathbb{N}}
\newcommand{\R}{\mathbb{R}}

\newcommand{\msf}[1]{\mathsf{#1}}

\newcommand{\head}[1]{\ensuremath{\msf{last}(#1)}}

\usepackage{bookmark}

\usepackage{pifont}
\usepackage{authblk}
\usepackage{graphicx}
\usepackage{multirow}

\author{Benjamin Bordais, Jean-François Raskin
	\and Université Libre de Bruxelles}
\date{}
\usepackage{array,booktabs}
\newcolumntype{M}[1]{>{\centering\arraybackslash}m{#1}}

\begin{document}
	\title{Multi-Environment MDPs with Prior and Universal Semantics}
	\maketitle	
	
	\begin{abstract}
		Multiple-environment Markov decision processes (MEMDPs) equip an MDP with several probabilistic transition functions (one per possible environment) so that the state is observable but the environment is not. Previous work studies two semantics: (i) the universal semantics, where an adversary picks the environment; and (ii) the prior semantics, where the environment is drawn once before execution from a fixed distribution. We clarify the relation between these semantics. For parity objectives, we show that the qualitative questions, i.e. value one, coincide, and we develop a new algorithm for the general value of MEMDP with prior semantics. In particular, we show that the prior value of an MEMDP with a parity objective can be approximated to any precision with a space efficient algorithm; equivalently, the associated gap problem is decidable in PSPACE when probabilities are given in unary (and in EXPSPACE otherwise). We then prove that the universal value equals the infimum of prior values over all beliefs. This yields a new algorithm for the universal gap problem with the same complexity (PSPACE for unary probabilities, EXPSPACE in general), improving on earlier doubly-exponential-space procedures. Finally, we observe that MEMDPs under the prior semantics form an important tractable subclass of POMDPs: our algorithms exploit the fact that belief entropy never increases, and we establish that any POMDP with this property reduces effectively to a prior-MEMDP, showing that prior-MEMDPs capture a broad and practically relevant subclass of POMDPs. 
	\end{abstract}

	\section{Introduction}

	Multiple-environment Markov decision processes (MEMDPs), first introduced in~\cite{FSTTCS14-MEMDP}, generalize classical MDPs. They model decision-making problems where the system state is perfectly observable, while the stochastic \emph{environment} is unknown but fixed and chosen or drawn before execution from a finite set of candidates.
	Formally, an MEMDP consists of a common state and action space together with a probabilistic transition function per possible environment.
	A single controller acts without being explicitely revealed which environment is operating. Intuitively, this model lies between fully observable MDPs and POMDPs: it captures a natural class of partial-information problems while avoiding many undecidability barriers of general POMDPs.
	
	Initial work introduced MEMDPs and showed that key synthesis questions that are undecidable for POMDPs become decidable (and sometimes efficiently so) in MEMDPs, thereby justifying the model as a tractable yet expressive {\em variant} of POMDPs~\cite{FSTTCS14-MEMDP}.
	Subsequent work further developed the theory for $\omega$-regular objectives and clarified the algorithmic landscape~\cite{ICALP25-MEMDP-Parity,CONCUR24-MEMDP-Rabin}.
	Beyond foundational interest, MEMDPs have been advocated as a practical modeling tool: their structure enables cheaper belief updates than in general POMDPs and supports applications such as contextual recommendation and parameter uncertainty in probabilistic models~\cite{ICAPS20-MEMDP}.
	
	\paragraph{\textbf{Universal and prior semantics.}}
	In previous works, MEMDPs have been studied under two different semantics: ~\cite{FSTTCS14-MEMDP,CONCUR24-MEMDP-Rabin,ICALP25-MEMDP-Parity} study the model under the {\em universal} semantics and~\cite{ICAPS20-MEMDP} studies the model under the {\em prior} semantics.
	In the \emph{universal} (adversarial) semantics, an adversary picks the environment that the fixed controller strategy has to face; the value of a strategy is its probability to satisfy the objective in the worst environment.
	In the \emph{prior} semantics, the environment is drawn initially at random from a fixed distribution (the prior); the value is the corresponding expectation.
	In both cases, the chosen environment is not revealed to the controller.
	We focus on parity objectives, that are a canonical way of expressing $\omega$-regular properties, and establish structural and algorithmic links between these semantics.
	
	Before presenting our main contributions, we note a \emph{qualitative} equivalence between the two semantics. The value~1 notions coincide in both the \emph{almost-sure} and \emph{limit-sure} senses: there exist strategies that achieve value~1 (resp. arbitrarily close to 1) in the universal semantics if and only if there exist strategies whose prior value is~1 (resp. arbitrarily close to 1). %This connects the two viewpoints at the qualitative level and lets us derive value~1 results for the prior semantics directly from the universal-semantics analysis (see Corollary~\ref{coro:limit-sure-almost-sure}).
	This correspondence lets us transfer value~1 results between the two semantics (see Corollary~\ref{coro:limit-sure-almost-sure}).
	
	\textbf{Contributions.}	Our main contributions can be summarized as follows:
	\begin{enumerate}
		%\item \textbf{Qualitative alignment.}
		%We prove that the qualitative (value~1) questions, both {\em almost sure} and {\em limit sure} winning, coincide under both semantics% for parity objectives
		%. 
		%For almost, we prove that there exists a strategy that achieves value~1 in the universal semantics if and only if there exists a strategy whose prior value is~1. The same property holds for limit sure winning: if there exists a family of strategies to get arbitrarily close to one in the prior semantics, there there is one such family in the universal semantics, and conversely.
		%We prove that there exist strategies that achieve value~1 (resp. arbitrarily close to 1) in the universal semantics if and only if there exist strategies whose prior value is~1 (resp. arbitrarily close to 1). This connects the two viewpoints at the qualitative level and lets us derive value~1 results for the prior semantics directly from the universal-semantics analysis (see Corollary~\ref{coro:limit-sure-almost-sure}).
		%
		\item \textbf{Approximating the prior value.}
		We give an algorithm that approximates the \emph{prior} value of an MEMDP %with a parity objective 
		to arbitrary precision. This is the most technically demanding part of our contribution. More precisely, we show how to solve algorithmically the $\varepsilon$-gap problem for the prior semantics. This problem asks, given an MEMDP $M$, a prior belief $b$ about the operating environment, a parity objective $W$, a threshold $0< \alpha <1$, and a precision $\varepsilon>0$, to answer \textsf{YES} if the {\em prior value} is at least $\alpha$, \textsf{NO} if this value is at most $\alpha-\epsilon$, with no requirement otherwise. Our algorithm runs in \textsc{PSPACE} when probabilities are given in unary, and in \textsc{EXPSPACE} otherwise (see Theorem~\ref{thm:complexity_deciding_gap_problem}). 
		
		\item \textbf{From prior to universal% via infima over beliefs
			.}
		We also show that the \emph{universal} value equals the infimum, over all prior beliefs, of the corresponding \emph{prior} values (see Theorem~\ref{thm:approx_value}).
		Using this result together with the 1-Lipschitz continuity of the prior value with respect to the prior belief, we obtain a new algorithm for the universal $\varepsilon$-gap problem with the same complexity as for the prior value (\textsc{PSPACE} for unary probabilities, \textsc{EXPSPACE} in general), see Theorem~\ref{thm:complexity_deciding_gap_problem_uni_val}, improving on earlier doubly-exponential-space procedures proposed in~\cite{ICALP25-MEMDP-Parity} .
		
		\item \textbf{MEMDPs with prior semantics as a tractable POMDP subclass.}
		A key reason for MEMDP tractability is that the hidden environment is fixed, which constrains belief dynamics: the (information theory) entropy of the belief about the operating environment is non-increasing.
		%Perhaps surprisingly, w
		We show that in fact any POMDP whose belief entropy is non-increasing can be %effectively 
		reduced to a prior-MEMDP (at an exponential cost). Thus, our algorithms apply beyond “pure” MEMDPs and capture a significant subclass of POMDPs, while avoiding classical undecidability phenomena for $\omega$-regular objectives (see Theorem~\ref{thm:Dirac-preserving-pomdp-are-memdps}).    
	\end{enumerate}
	
	\paragraph{\textbf{Related work.}}
	As recalled above, MEMDPs were introduced under the universal semantics by Raskin and Sankur in~\cite{FSTTCS14-MEMDP} as a formal model for decision making in scenarios where the environment is fixed yet unknown, and is %assumed to be 
	chosen from a finite set of candidate models. Within this framework, several qualitative $\omega$-regular synthesis problems become decidable, in contrast to the corresponding situation for partially observable Markov decision processes (POMDPs).  Subsequent work has refined the qualitative complexity landscape. For almost-sure satisfaction of Rabin objectives, Suilen et al.\ establish a PSPACE-complete decision procedure~\cite{CONCUR24-MEMDP-Rabin}. For parity objectives, Chatterjee et al.\ prove that the value-1 problem is PSPACE-complete in general (and solvable in {\sc PTime} when the number of environments is fixed), demonstrate that pure strategies suffice to achieve value~1, and present a double-exponential-space approximation scheme for computing the value~\cite{ICALP25-MEMDP-Parity}. Under the prior semantics, MEMDPs have further been proposed as a computationally tractable subclass of POMDPs with discounted payoffs, leveraging more efficient belief-state updates and a non-increasing belief-entropy property~\cite{ICAPS20-MEMDP}.
	
	For context, classical (fully observable) MDP qualitative problems are solvable in polynomial time~\cite{DBLP:books/daglib/0020348}, while analogous POMDP questions are much harder as most of them are undecidable~\cite{PT87-MDP-Complexity,AIJ03-Undec-Probabilistic-Planning}, and only bounded horizon questions are known to be decidable~\cite{PT87-MDP-Complexity}.
	So, in this context, MEMDPs occupy a particular place: they retain an interesting part of the partial-information expressiveness but avoid key undecidability barriers that arise in POMDPs. The positive results obtained in this paper add and refine those obtained in~\cite{FSTTCS14-MEMDP,ICAPS20-MEMDP,CONCUR24-MEMDP-Rabin,ICALP25-MEMDP-Parity}.

	\section{Definitions}
	\label{sec:Def}
	Consider a non-empty set $Q$. The \emph{support} $\Supp(d)$ of a function $d: Q \to [0,1]$ is the set $\Supp(d) := \{ q \in Q \mid d(q) > 0 \}$. A function $d: Q \to [0,1]$ is a \emph{distribution} over $Q$ if it has {\em countable} support and $\sum_{q \in Q} d(q) = 1$. We let $\Dist(Q)$ denote the set of all probability distributions over the set $Q$. A probability distribution $d \in \Dist(Q)$ is \emph{Dirac} if $|\Supp(b)| = 1$.
	For all sets $A$ and $\rho = q_0 \cdot (a_1,q_1) \cdots (a_n,q_n) \in Q \cdot (A \cdot Q)^*$, we let $\head{\rho} := q_n \in Q$% and $\hd{A}{\rho} := a_n \in Q$
	. 
	
	\paragraph{MDPs and strategies.}
	A \emph{Markov Decision Process} (MDP for short) $\MDP$ is a tuple $\MDP = (Q,A,\delta)$ where $Q$ is a non-empty finite set of \emph{states}, $A$ is a non-empty finite set of \emph{actions}, and $\delta\colon Q \times A \to \Dist(Q)$ is the transition function mapping each state-action pair to a probability distribution over successor states. An element in $Q \cdot (A \cdot Q)^*$ (resp. $(Q \cdot A)^\omega$) is called a finite (resp. infinite) \emph{run}, an element in $Q^*$ (resp. $Q^\omega$) is called a finite (resp. infinite) \emph{path}. A \emph{strategy} on $(Q,A)$ is a function $\sigma\colon Q \cdot (A \cdot Q)^* \to \Dist(A)$ mapping finite runs to probability distributions over actions. We let $\msf{Strat}(Q,A)$ denote the set of all strategies on $(Q,A)$. 
	
	\paragraph{Objectives.}
	Consider a non-empty finite set $Q$. An \emph{objective} $W \subseteq Q^\omega$ is a Borel set: $W \in \msf{Borel}(Q)$. We focus on parity objectives, i.e. objectives $W \in \msf{Borel}(Q)$ such that there is a labeling function $f: Q \to \N$ such that $W$ is the set of infinite paths whose highest label seen infinitely often is even: $W = \{ \rho \in Q^\omega \mid \max \{ n \in \N \mid \forall i \in \N,\; \exists j \geq i: f(\rho_j) = n \} \text{ is even}\}$. %Parity objectives are an example of \emph{prefix-independent} objectives, i.e. objectives $W \in \msf{Borel}(Q)$ such that, for all $\rho \in Q^\omega$ and $\pi \in Q^*$, we have $\rho \in W$ iff $\pi \cdot \rho \in W$. 
	
	Consider an MDP $\MDP = (Q,A,\delta)$. A strategy $\sigma\colon Q \cdot (A \cdot Q)^* \to \Dist(A)$ induces a probability measure on the set of finite runs, which can be canonically extended to the associated Borel $\sigma$-algebra over infinite runs. We first define the probability measure induced by a strategy after a fixed finite run $\rho \in Q \cdot (A \cdot Q)^*$. For $\pi \in Q \cdot (A \cdot Q)^*$, we define the value $\Prb_\rho^\sigma(\MDP,\pi) \in [0,1]$ by induction on the length of $\pi$ as follows. If $|\pi| \leq |\rho|$, then we put $\Prb_\rho^\sigma(\MDP,\pi) := 1$ if $\pi$ is a prefix of $\rho$, and $\Prb_\rho^\sigma(\MDP,\pi) := 0$ otherwise. For all $\pi \in Q \cdot (A \cdot Q)^*$ such that $\rho$ is a prefix of $\pi$ and for all $(a,t) \in A \times Q$, we define $\Prb_\rho^\sigma(\MDP,\pi \cdot (a,t)) := \Prb_{\rho}^\sigma(\MDP,\pi) \cdot \sigma(\pi)(a) \cdot \delta(\head{\pi},a)(t).$
	
	Then, for all $\pi \in Q \cdot (A \cdot Q)^*$, we define the probability $\Prb_\rho^\sigma[\MDP,\cyl(\pi)]$ of the cylinder set $\cyl(\pi) := \{\pi \cdot \rho' \mid \rho' \in (A \cdot Q)^\omega \} \in \msf{Borel}(Q \cdot A)$ by $\Prb_\rho^\sigma[\MDP,\cyl(\pi)] := \Prb_\rho^\sigma(\MDP,\pi) \in [0,1]$. We naturally extend this into a probability measure on Borel sets: $\Prb_\rho^\sigma[\MDP,\cdot]\colon \msf{Borel}(Q \cdot A) \to [0,1]$. We obtain the probability measure of Borel sets in $\msf{Borel}(Q)$ via projection.
	
	\paragraph{MEMDPs.} A \emph{Multi-Environment Markov Decision Process} (MEMDP for short) $\MEMDP$ is a tuple $\MEMDP = (Q,A,E,(\delta_e)_{e \in E})$ where $E$ is a non-empty finite set of environments, and for all $e \in E$, $(Q,A,\delta_e)$ is an MDP, which we denote $\MEMDP[e]$. %We assume that, for all $e \neq e' \in E$, we have $\delta_e \neq \delta_{e'}$. 
	Unless otherwise stated, an MEMDP $\MEMDP$ refers to the tuple $\MEMDP = (Q,A,E,(\delta_e)_{e \in E})$. For all $e \neq e' \in E$, a state-action pair $(q,a) \in Q \times A$ is $(e,e')$-\emph{distinguishing} if $\delta_e(q,a) \neq \delta_{e'}(q,a)$. We let $\msf{Dstg}(\MEMDP,e,e')$ denote the set of all $(e,e')$-distinguishing state-action pairs. We also let $\msf{Dstg}(\MEMDP) := \cup_{e \neq e' \in E} \msf{Dstg}(\MEMDP,e,e')$. 
	
	\paragraph{Values in MEMDPs.} Consider an MEMDP and a parity objective whose probability we seek to maximize. When synthesizing a strategy in an MEMDP, we do not know in which environment it executes. Thus, when defining the value of a synthesized strategy, we may require either that it performs well in \emph{all} environments or, given a prior belief on the operating environment, that it performs well on \emph{average}. The goal of this paper is to study the latter \textquotedblleft{}prior\textquotedblright{} value, and to link it to the former \textquotedblleft{}universal\textquotedblright{} value%, which has already been studied in the literature \cite{FSTTCS14-MEMDP,DBLP:conf/tacas/VegtJJ23,DBLP:conf/icalp/Chatterjee0RS25}
	. 
	\iffalse
	
	There are two natural settings: a \textquotedblleft{}universal\textquotedblright{} setting, where a strategy no prior on the environment in which we play; and a \textquotedblleft{}prior\textquotedblright{} setting, where we have a belief about the likelihood of the environment in which we play. In the \textquotedblleft{}universal\textquotedblright{} setting, the strategy value is the least probability of the objective over all possible environments. In the \textquotedblleft{}prior\textquotedblright{} setting, given a belief on the environments---formally described as a probability distribution over the environments---a strategy value is the expected value of the objective given that belief.

	We may have no prior knowledge about which is the current environment, in which we look for a strategy guaranteeing a probability of the objective against all environments. This is the \textquotedblleft{}adversarial\textquotedblright{} setting, studied in \cite{DBLP:conf/icalp/Chatterjee0RS25}. On the other hand, we may have some information about the current environment, which translates into a prior belied on (i.e. probability distribution over) the current environment. This is the \textquotedblleft{}prior\textquotedblright{} setting, in which we are looking for a strategy maximizing the expected value (w.r.t. the prior belief) of the probability of the objective. 
	
	We are thus interested in different values. 
	\fi 
	Specifically, given an MEMDP $\MEMDP% = (Q,A,E,(\delta_e)_{e \in E})
	$, a state $q \in Q$, and a parity objective $W \in \msf{Borel}(Q)$, we define $\msf{uni}$-values
	\begin{equation*}
		\msf{val}^{\msf{uni}}_q(\MEMDP,W) := \sup_{\sigma \in \msf{Strat}(Q,A)} \msf{val}^{\msf{uni}}_q(\MEMDP,\sigma,W) \text{ with } \msf{val}^{\msf{uni}}_q(\MEMDP,\sigma,W) := \min_{e \in E} \Prb_q^\sigma[\MEMDP[e],W] 
	\end{equation*}
	and $\msf{pr}$-values, given a prior belief $b \in \Dist(E)$
	\begin{equation*}
		\msf{val}^{\msf{pr}}_q(\MEMDP,b,W) := \sup_{\sigma \in \msf{Strat}(Q,A)} \msf{val}^{\msf{pr}}_q(\MEMDP,b,\sigma,W) \text{ with }\msf{val}^{\msf{pr}}_q(\MEMDP,b,\sigma,W) := \sum_{e \in E} b(e) \cdot \Prb_q^\sigma[\MEMDP[e],W]
	\end{equation*}
	
	\begin{example}
		\label{ex:cardgame}
		Let us consider Fig.~\ref{fig:duplicateddeck}, which represents an MEMDP. States, actions, and probabilistic successors are defined as in a classical MDP, with the key difference that here we consider multiple valuations of the parameters that label the probabilistic transitions (i.e., the parameters on the outgoing edges of the triangular nodes).
		
		As a first example, assume that we want to model a deck of cards in which card $1$ has two copies and card $2$ has one copy. In this case, the probability of drawing card $1$ is $\alpha_1=\frac{2}{3}$, while the probability of drawing card $2$ is $\alpha_2=\frac{1}{3}$. Suppose moreover that the player is always allowed to request another draw before making a guess; this is captured by setting $\alpha_0=0$.
		
		To model the fact that the duplicated card is card $1$, we assign the parameters $\beta_1=1$ and $\beta_2=0$, meaning that guessing card $1$ leads with certainty to the winning state $W$, while guessing card $2$ leads to the losing state $L$. Symmetrically, we set $\gamma_1=0$ and $\gamma_2=1$. We refer to this valuation of the parameters as environment $E_1$. We can define another environment, denoted $E_2$, that models the situation in which card $2$ is duplicated instead. In this case, the parameters are set to $\alpha_1=\frac{1}{3}$ and $\alpha_2=\frac{2}{3}$, together with $\beta_1=0$, $\beta_2=1$, $\gamma_1=1$, and $\gamma_2=0$.
		
		Finally, to encode the reachability objective (i.e., reaching state $W$), we use the following parity labeling: all states have label $1$, except the winning state $W$, which has label $2$. Under this labelling, the parity objective is satisfied if and only if the play eventually reaches $W$, i.e., the player correctly guesses the duplicated card.
		
		%We will use Fig.~\ref{fig:duplicateddeck} as a running example throughout the paper, by instantiating it with different valuations of its parameters, corresponding to different environments. We will consider two semantics for selecting the active environment: (i) the universal semantics, where an adversary chooses the environment, and (ii) the prior semantics, where the environment is drawn at random according to a fixed distribution.
	\end{example}
	
	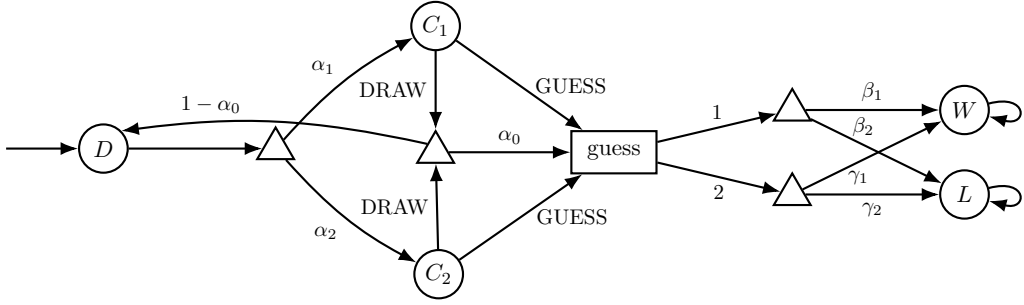
\begin{figure}[t!]
		\centering
		\begin{tikzpicture}[
			>=Latex,
			node distance=16mm and 22mm,
			state/.style={circle, draw, thick, minimum size=8mm, inner sep=1pt},
			prob/.style={regular polygon, regular polygon sides=3, draw, thick, inner sep=1.5pt, minimum size=7mm},
			ctrl/.style={rectangle, draw, thick, minimum width=14mm, minimum height=7mm, inner sep=2pt},
			lab/.style={font=\small},
			scale=0.8, transform shape]
			
			% --- Nodes (left part)
			\node[state] (D) {$D$};
			\node[prob, right=of D] (p0) {};
			\node[state, above right=of p0] (C1) {$C_1$};
			\node[state, below right=of p0] (C2) {$C_2$};
			
			% DRAW probabilistic node (right of C1/C2)
			\node[prob, below=13mm of C1] (pdraw) {};
			
			% Guess control node
			\node[ctrl, right=20mm of pdraw] (G) {guess};
			
			% --- Nodes (right part)
			\node[prob, right=20mm of G, yshift=7mm] (p1) {};
			\node[prob, right=20mm of G, yshift=-7mm] (p2) {};
			\node[state, right=22mm of p1] (W) {$W$};
			\node[state, right=22mm of p2] (L) {$L$};
			
			% --- Initial arrow
			\draw[->, thick] (-16mm,0) -- (D);
			
			% --- From D to initial probabilistic choice
			\draw[->, thick] (D) -- (p0);
			
			% --- Probabilistic split to C1/C2
			\draw[->, thick,bend left=10] (p0) to node[lab, above left] {$\alpha_1$} (C1);
			\draw[->, thick,bend right=10] (p0) to node[lab, below left] {$\alpha_2$} (C2);
			
			% --- DRAW transitions to probabilistic node
			\draw[->, thick] (C1) -- node[lab, left] {DRAW} (pdraw);
			\draw[->, thick] (C2) -- node[lab, left] {DRAW} (pdraw);
			
			% --- From C2: GUESS action directly to guess node
			\draw[->, thick] (C2) -- node[lab, midway, right] {~GUESS} (G);
			\draw[->, thick] (C1) -- node[lab, midway, right] {~GUESS} (G);
			
			% --- From DRAW probabilistic node: back to D or to guess node
			\draw[->, thick]
			(pdraw.north west) to[out=180, in=30,bend right=10]
			node[lab, above,xshift=-30pt] {$1-\alpha_0$}
			(D.north east);
			\draw[->, thick, bend right=0] (pdraw) to node[lab, above] {$\alpha_0$} (G);
			
			% --- From guess: two actions (1 and 2) to probabilistic nodes
			\draw[->, thick] (G) -- node[lab, above] {$1$} (p1);
			\draw[->, thick] (G) -- node[lab, below] {$2$} (p2);
			
			% --- Probabilistic outcomes to W/L
			\draw[->, thick] (p1) -- node[lab, above] {$\beta_1$} (W);
			\draw[->, thick] (p1) -- node[lab, pos=0.4,above] {$\beta_2$} (L);
			
			\draw[->, thick] (p2) -- node[lab, pos=0.4, below] {$\gamma_1$} (W);
			\draw[->, thick] (p2) -- node[lab, below] {$\gamma_2$} (L);
			
			% --- Absorbing self-loops
			\draw[->, thick, loop right] (W) to (W);
			\draw[->, thick, loop right] (L) to (L);
			
			% Optional: small blue annotations from the figure (remove if not needed)
			% \node[lab, blue, above=1mm of D] {1};
			% \node[lab, blue, above=1mm of C1] {1};
			% \node[lab, blue, above=1mm of C2] {1};
			% \node[lab, blue, below=1mm of G] {1};
			% \node[lab, blue, right=1mm of W] {2};
			% \node[lab, blue, right=1mm of L] {1};
			
		\end{tikzpicture}
		\caption{The figure depicts an MEMDP (inspired from~\cite{DBLP:conf/icalp/Chatterjee0RS25}) modelling a simple card game played with a deck containing two card types, $1$ and $2$. The deck composition is parameterized by $\alpha_1$ and $\alpha_2$, with $\alpha_1+\alpha_2=1$; for instance, if card $1$ appears twice as often as card $2$, then $\alpha_1=2\alpha_2$. At each turn, the player may either \emph{guess} which card type is in the majority or request an additional \emph{draw} from the deck; after a draw request, the game continues with probability $1-\alpha_0$, while with probability $\alpha_0$ the player is forced to guess immediately. The player’s objective is to reach the winning state $W$, which corresponds to correctly guessing the card type that has the larger number of copies in the deck; this is captured %in the MEMDP 
			by suitable choices of the %outcome 
			parameters (e.g., $\beta$ and $\gamma$) governing the transition from the guess action to $W$ or $L$.}
		\label{fig:duplicateddeck}
		
	\end{figure}
	
	\section{Qualitative and quantitative problem}
	\label{sec:solve_gap_problem}
	In this section, we focus on deciding the existence of strategies with good enough prior-values% with parity objectives
	. We first study the qualitative problems (value 1) and then turn to the  quantitative problem.
	
	\subsection{Qualitative problem}
	For value (arbitrarily close to) 1, the universal and prior settings coincide. 
	
	\begin{proposition}[Proof~\ref{proof:prop_value_one}]
		\label{prop:value_one}
		Consider an MEMDP $\MEMDP$, a state $q \in Q$, and a parity objective $W$. Let $b \in \Dist(E)$ such that $\Supp(b) = E$. Then $\msf{val}^{\msf{uni}}_q(\MEMDP,W) = 1$ if and only if $\msf{val}^{\msf{pr}}_q(\MEMDP,b,W) = 1$ and, 
		for all $\sigma \in \msf{Strat}(Q,A)$: $\msf{val}^{\msf{uni}}_q(\MEMDP,\sigma,W) = 1$ if and only if $\msf{val}^{\msf{pr}}_q(\MEMDP,b,\sigma,W) = 1$.
	\end{proposition}
	\begin{proof}[Proof sketch]
		If $\msf{val}^{\msf{uni}}_q(\MEMDP,W) = 1$, then there are strategies whose $\msf{uni}$-values are arbitrarily close to 1, thus their $\msf{pr}$-values are arbitrarily close to 1. Hence, $\msf{val}^{\msf{pr}}_q(\MEMDP,W) = 1$. Furthermore, if a strategy has a $\msf{pr}$-value, w.r.t. $b$, of at least $1 - d \cdot \varepsilon$ for some $\varepsilon > 0$, with $d := \min_{e \in E} b(e) > 0$, then its smallest value in any environment (since $\Supp(b) = E$) is at least $1 - \varepsilon$. Hence, if $\msf{val}^{\msf{pr}}_q(\MEMDP,W) = 1$ then %there are strategies whose average values are arbitrarily close to 1, and thus there are strategies whose values, in all environments, are arbitrarily close to 1. Hence, 
		$\msf{val}^{\msf{uni}}_q(\MEMDP,W) = 1$. The other equivalence is straightforward.
	\end{proof}
	
	The complexity of deciding the existence of almost-surely winning strategies (value 1) and of limit-surely winning strategies (values arbitrarily close to 1) is thus the same in MEMDPs with parity objectives in the universal and prior settings. 
	%The corollary below is thus a direct consequence of \cite[Theorems 4, 13]{DBLP:conf/icalp/Chatterjee0RS25}.
	%Since deciding the existence of a strategy of value 1 (\textquotedblleft{}almost-surely winning\textquotedblright{}) and of strategies of values arbitrarily close to 1 (\textquotedblleft{}limit-surely winning\textquotedblright{}) in the universal setting can be done in polynomial time if the set of environment is fixed, and in polynomial space otherwise \cite[Theorem 4, 13]{DBLP:conf/icalp/Chatterjee0RS25}, we obtain the corollary below.
	We deduce the corollary below.% for parity objectives.
	\begin{corollary}[of {\cite[Theorems 4, 13]{DBLP:conf/icalp/Chatterjee0RS25}}]
		\label{coro:limit-sure-almost-sure}
		In an MEMDP, given a prior belief and a parity objective, the problem of deciding the existence of strategies of $\msf{pr}$-value (arbitrarily close to) 1 %1 and of strategies of $\msf{pr}$-value  1 are both 
		is $\msf{PSPACE}$-complete. When the number of environments is fixed, it can be solved in polynomial time.
	\end{corollary}
	
	\begin{example}
		%Let us illustrate Corrolary~\ref{coro:limit-sure-almost-sure} by considering 
		Consider again the MEMDP from Example~\ref{ex:cardgame}. It is easy to see that no strategy can guarantee winning with probability $1$. Indeed, after any finite number of draws, the empirical frequencies may still be atypical and hence suggest the wrong duplicated card.
		
		Nevertheless, for every $\varepsilon>0$, the player can request sufficiently many draws so that the probability of observing such a misleading sample drops below $\varepsilon$. By the law of large numbers, this implies that although almost-sure winning is impossible, the player can make the winning probability arbitrarily close to $1$ by choosing an appropriate family of strategies (indexed by $\varepsilon$).
		
		This reasoning applies both under the universal semantics and under the prior semantics whenever the prior distribution assigns positive probability to both environments. In either case there is no surely winning strategy, but there exists a family of strategies witnessing limit-sure winning; in particular, the construction works for any fixed prior over the two environments.
	\end{example}
	
	\subsection{Quantitative problem}
	Let us now turn to the quantitative problem.
	%, i.e. the problem of deciding, in an MEMDP $\MEMDP$, with a prior belief $b \in \Dist(E)$ and a parity objective $W$, given a threshold $\alpha \in [0,1]$, if $\msf{val}^{\msf{pr}}_q(\MEMDP,b,W) \geq \alpha$. In fact, as is done in \cite{DBLP:conf/icalp/Chatterjee0RS25}, we do not tackle exactly this problem, instead 
	We design an algorithm that, for any prescribed precision level $\varepsilon > 0$, computes an $\varepsilon$-approximation of the $\msf{pr}$-value $\msf{val}^{\msf{pr}}_q(\MEMDP,b,W)$. This approximation procedure is subsequently employed to solve the associated \emph{gap problem}, defined as follows: given a parameter $\varepsilon > 0$, one must output YES if $\msf{val}^{\msf{pr}}_q(\MEMDP,b,W) \geq \alpha$, NO if $\msf{val}^{\msf{pr}}_q(\MEMDP,b,W) \leq \alpha - \varepsilon$, and an arbitrary answer otherwise. Since there are inputs for which no specific output is required, the gap problem %This gap problem does not constitute 
	is not a classical decision problem; it is instead a promise problem %, since there exist inputs for which no specific output is required 
	\cite{DBLP:conf/birthday/Goldreich06a}.  
	
	%Similarly to what is done in \cite{DBLP:conf/icalp/Chatterjee0RS25}, we develop an approximation algorithm solving the \emph{gap problem} which, given some $\varepsilon > 0$, is supposed to answer: Yes if $\msf{val}^{\msf{pr}}_q(\MEMDP,b,W) \geq \alpha$; No if $\msf{val}^{\msf{pr}}_q(\MEMDP,b,W) \leq \alpha - \varepsilon$; and arbitrarily otherwise. An algorithm solving the gap problem 
	%\begin{itemize}
	%	\item Yes if $\msf{val}^{\msf{pr}}_q(\MEMDP,b,W) \geq \alpha$;
	%	\item No if $\msf{val}^{\msf{pr}}_q(\MEMDP,b,W) \leq \alpha - \varepsilon$;
	%	\item arbitrarily otherwise. 
	%\end{itemize}
	%This procedure can in turn be used as an approximation algorithm for solving the quantitative problem. An exact solution for the quantitative problem is left as future work.
	
	%The remainder of this section is dedicated to the design of an algorithm solving this gap problem. 
	Our goal is to compute an approximation of the $\msf{pr}$-value % of strategies, thus we start with 
	given a prior environment belief $b \in \Dist(E)$. When playing in an MEMDP, the environment belief is updated each time distinguishing state-action pairs are encountered: when we visit a transition $(q,a,q') \in Q \times A \times Q$ such that $\delta_e(q,a)(q') > \delta_{e'}(q,a)(q')$ for some $e \neq e' \in E$, the belief is updated to account for an increased likelihood of $e$ and a decreased likelihood in $e'$. 
	%
	%the fact that the belief that the current environment is environment $e$ is more likely than before compared to environment $e'$. 
	%
	%environment $e$ is more likely than compared %  before becomes more likely than before compared 
	%to the environment $e'$. 
	We formally define below this belief update. %Our definition actually corresponds to the belief update used in partially observable MDPs (POMDPs). 
	\begin{definition}[Belief update]
		\label{def:belief_env_update}
		Consider an MEMDP $\MEMDP = (Q,A,E,(\delta_e)_{e \in E})$. Let $b \in \Dist(E)$. For all $(q,a) \in Q \times A$, we let $p[b,q,a] \in \Dist(Q)$ be defined by, for all $q' \in Q$: $p[b,q,a](q') := \sum_{e \in E} b(e) \cdot \delta_e(q,a)(q')$. Then, for all $q' \in Q$, we let $\lambda[b,q,a,q'] \in \Dist(E)$ be an arbitrary distribution if $p[b,q,a](q') = 0$, and otherwise, for all $e \in E$: $\lambda[b,q,a,q'](e) := \frac{b(e) \cdot \delta_e(q,a)(q')}{p[b,q,a](q')}$.
		
		Given any initial state $q \in Q$ and prior belief $b \in \Dist(E)$, we define a (\textquotedblleft{}likelihood\textquotedblright{}) function $\msf{lk}_b: Q \cdot (A \cdot Q)^* \to \Dist(E)$ which iteratively updates the environment belief: for all $q \in Q$, $\msf{lk}_b(q) := b \in E$ and, for all $\rho \cdot (a,q') \in Q \cdot (A \cdot Q)^+$, we let $\msf{lk}_b(\rho \cdot (a,q')) := \lambda[\msf{lk}_b(\rho),\head{\rho},a,q']$.
	\end{definition}
	
	\begin{example}
		Let us illustrate how belief updates work on our running example. Assume that the prior over the two environments is given by the distribution $b = (\frac{1}{4},\frac{3}{4})$, so the initial belief is skewed towards environment $E_2$ (i.e., towards card $2$ being duplicated). Now suppose that after the first draw we observe card $2$, meaning that we reach state $2$. This observation should increase our confidence that the operating environment is indeed $E_2$. Applying the update rule above, the posterior belief becomes $\lambda[b,D,a,C_2](E_2) = \frac{3}{4} \cdot \frac{\frac{2}{3}}{\frac{1}{3} \cdot \frac{1}{4} + \frac{2}{3} \cdot \frac{3}{4}} = \frac{3}{4} \cdot \frac{24}{21} = \frac{6}{7}$, and $\lambda[b,D,a,C_2](E_1) = \frac{1}{7}$. Conversely, if the first draw yields card $1$, then the posterior belief becomes $\lambda[b,D,a,C_1](E_2) = \frac{3}{4} \cdot \frac{\frac{1}{3}}{\frac{2}{3} \cdot \frac{1}{4} + \frac{1}{3} \cdot \frac{3}{4}} = \frac{3}{4} \cdot \frac{12}{15} = \frac{3}{5}$, and $\lambda[b,D,a,C_1](E_1) = \frac{2}{5}$.
	\end{example}
	
	Our algorithm computing an approximation of the $\msf{pr}$-value %solving the gap problem 
	works recursively on the number of environments in the support of the belief. When there is a single environment in the support, we are actually in an MDP, and we can compute in polynomial time the values of all states \cite{DBLP:books/daglib/0020348}. Assume now that the support of the prior belief is of size at least 2. If %, along a run of the MEMDP, 
	the belief in an environment $e$ ever drops to 0, because we have visited a transition $(q,a,q')$ such that $\delta_e(q,a)(q') = 0$% that is not compatible with it
	, we can call the algorithm recursively with a belief of smaller support. 
	%As long as no distinguishing state-action pair is visited, the belief does not change and the MEMDP behaves as a regular MDP. 
	%When distinguishing state-action pairs are visited, the belief in an environment may drop to 0 if we have seen a transition that is not compatible with it, in which case we can call the algorithm recursively with a starting belief of smaller support. 
	Alternatively, the belief in an environment may be close to 0. In that case, it is almost harmless (in terms of value change) to truncate the belief into one with a smaller support, by 1-Lipschitz continuity of the $\msf{pr}$-value. 
	\begin{lemma}[Proof~\ref{proof:lem_small_belief_change_ok}]
		\label{lem:small_belief_change_ok}
		Consider an MEMDP $\MEMDP$ and a parity objective $W$. For all beliefs $b,b' \in \Dist(E)$, we have $|\msf{val}^{\msf{pr}}_q(\MEMDP,b,W) - \msf{val}^{\msf{pr}}_q(\MEMDP,b',W)| \leq \msf{Diff}(b,b')$, with $\msf{Diff}(b,b') := \frac{1}{2} \cdot \sum_{e \in E} |b(e) - b'(e)|$. 
		\iffalse
		\begin{equation*}
			|\msf{val}^{\msf{pr}}_q(\MEMDP,b,W) - \msf{val}^{\msf{pr}}_q(\MEMDP,b',W)| \leq \frac{1}{2} \cdot \sum_{e \in E} |b(e) - b'(e)| := \msf{Diff}(b,b')
		\end{equation*}
		\fi
		Furthermore, for all beliefs $b \in \Dist(E)$ such that $|\Supp(b)| \geq 2$, there is $\msf{Truncate}(b) \in \Dist(E)$ such that $\Supp(\msf{Truncate}(b)) \subsetneq \Supp(b)$ and $\msf{Diff}(b,\msf{Truncate}(b)) \leq \min_{e \in \Supp(b)} b(e)$. 
		%
		%. Let $\tilde{e} \in E$ be such that $b(\tilde{e}) = \min_{e \in E} b(e)$ and $b' \in \Dist(E)$ be such that $b'(\tilde{e}) = 0$, there is some $e' \in E$ such that $b'(e') = b(\tilde{e}) + b(e')$, and $b'(e) = b(e)$ for all $e \in E \setminus \{ \tilde{e},e' \}$. Then, for all objectives $W \in \msf{Borel}(Q)$, we have:
		%\begin{equation*}
		%	|\msf{val}^{\msf{pr}}_q(\MEMDP,b,W) - \msf{val}^{\msf{pr}}_q(\MEMDP,b',W)| \leq b(e)
		%\end{equation*}
	\end{lemma}
	
	%As long as no distinguishing state-action pair is visited, the belief does not change and the MEMDP behaves as a regular MDP. 
	%On the other hand, t
	The environment belief is updated each time a distinguishing state-action pairs is visited. However, even if many distinguishing state-action pairs are visited, it could be that there is no environment whose belief drops to 0. %In fact, this surely never happens if the transitions functions of all environments have the same support.
	%Overall, we obtain that if we play in the MEMDP and come across a belief whose minimal value in the support is small enough, we can truncate it into a belief with smaller support. 
	%Along runs visiting many distinguishing state-action pairs, it could be that no environment belief drops to O. In fact, this surely happens if the transitions functions of all environments have the same support. 
	Nonetheless, if enough distinguishing state-action pairs are visited, with high probability, there is an environment whose belief drops close to 0. %This constitutes the main technical result of this section, it is formally stated in the theorem below. 
	%This is formally stated in the theorem below%, which constitutes a crucial result of this paper
	%.
	\begin{theorem}[Proof~\ref{proof:thm_main_enough_actions_small_belief}]
		\label{thm:main_enough_actions_small_belief}
		Consider an MEMDP $\MEMDP = (Q,A,E,(\delta_e)_{e \in E})$ and a prior environment belief $b \in \Dist(E)$. Let $k := |E|$, and $l$ denote the maximum number of bits to write any probability occurring in $\MEMDP$. For all $\varepsilon > 0$, we let: $\msf{NeverSmallBelief}(b,\varepsilon) := \{ \rho \in Q \cdot (A \cdot Q)^* \mid \forall \pi %\in Q \cdot (A \cdot Q)^*,\; \pi 
		\sqsubseteq \rho,\; \forall  e \in E:\; \msf{lk}_b(\pi)(e) > \varepsilon \}$, with $\pi \sqsubseteq \rho$ meaning that $\pi$ is a prefix of $\rho$. For all $\rho \in %Q \cdot (A \cdot Q)^* \cup 
		(Q \cdot A)^\omega$, we let $\msf{nb}_{\msf{dstg}}(\rho) \in \N \cup \{\infty\}$ denote the number of times $\rho$ visits a distinguishing state-action pair.
		
		Then, for all $\varepsilon > 0$, letting $x := \lceil \log(\frac{1}{\varepsilon})\rceil$, there is $m(\MEMDP,\varepsilon) \in \N$ such that $m(\MEMDP,\varepsilon) = O(poly(k,2^l,x))$ ensuring that for all strategies $\sigma \in \msf{Strat}(Q,A)$ and environments $e \in E$:
		\begin{equation*}
			\Prb_q^\sigma(\MEMDP[e],\{ %\rho \in \msf{Hist}^\omega(Q,A) \mid \msf{nb}_{\msf{dstg}}(\rho) \geq m(\MEMDP,\varepsilon) 
			\msf{nb}_{\msf{dstg}} \geq m(\MEMDP,\varepsilon) \} \cap \msf{NeverSmallBelief}(b,\varepsilon)) \leq \varepsilon
		\end{equation*}
	\end{theorem}
	This theorem constitutes a crucial result of this paper, its proof is rather intricate. %the main technical result of this section with a rather intricate proof. 
	Below, we only describe the two main ingredients of that proof; all the details can be found in Appendix~\ref{appen:actual_proof}, with the definition of the quantity $m(\MEMDP,\varepsilon)$ provided in Appendix~\ref{appen:def_bound}. 
	
	\begin{proof}[Proof sketch]
		Assume that some $e \in \Supp(b)$ is the operating environment. Consider another environment $e' \neq e \in \Supp(b)$ and assume, for simplicity, that $\Supp(\delta_e(q,a)) = \Supp(\delta_{e'}(q,a))$ for all $(q,a) \in Q \times A$. % all transition functions have the same support for $e$ and $e'$. 
		Our goal is to show that, regardless of the strategy $\sigma$, with high probability, runs $\rho$ that visit enough $(e,e')$-distinguishing state-action pairs are such that $\msf{lk}_b(\rho)(e')$ is arbitrarily small. We focus on the ratio of the beliefs in the environments $e$ and $e'$, which initially is equal to $\frac{b(e')}{b(e)}$. After some run $\rho \in Q \cdot (A \cdot Q)^*$, it is equal to $\frac{\msf{lk}_b(\rho)(e')}{\msf{lk}_b(\rho)(e)} = \frac{b(e')}{b(e)} \cdot \prod_{\pi \cdot (a,q) \sqsubseteq \rho} \frac{\delta_{e'}(\head{\pi},a)(q)}{\delta_{e}(\head{\pi},a)(q)}$. Since the sum of random variables is easier to analyze than their product, we aim at showing ($\ast$): \textquotedblleft{}Regardless of the strategy $\sigma$, with high probability, if $\rho$ visits enough $(e,e')$-distinguishing state-action pairs, the sum $\sum_{\pi \cdot (a,q) \sqsubseteq \rho} \log(\frac{\delta_{e'}(\head{\pi},a)(q)}{\delta_{e}(\head{\pi},a)(q)})$ is low enough\textquotedblright{}. An instrumental result in the proof of this fact is that there is some $\eta < 0$ such that, for all $(q,a) \in \msf{Dstg}(\MEMDP,e,e')$:% $(e,e')$-distinguishing state-action pairs $(q,a)$:%, the one step expected value, in the environment $e$, of $\log(\frac{\delta_{e'}(q,a)(\cdot)}{\delta_{e}(q,a)(\cdot)})$ is at most $\eta$::
		$\sum_{q' \in Q} \delta_{e}(q,a)(q') \cdot \log(\frac{\delta_{e'}(q,a)(q')}{\delta_{e}(q,a)(q')}) \leq \eta$. This is a consequence of the difference between the arithmetic and geometric means.
		
		In order to show ($\ast$), we want to use Hoeffding's inequality \cite{hoeffding1963probability} (which was also a central tool used in \cite{DBLP:conf/icalp/Chatterjee0RS25}). Informally, this inequality entails that if we consider $n \in \N$ independent real-valued random variables, the probability that their sum is higher than their expected value decreases exponentially with $n$. In our case, we could consider the random variables $X_k$, for $k \in \N$, that map the $k$-th transition $(q,a,q') \in Q \times A \times Q$ visited for which $(q,a) \in \msf{Dstg}(\MEMDP,e,e')$ % is a $(e,e')$-distinguishing state-action pair 
		to the real value $\log(\frac{\delta_{e'}(q,a)(q')}{\delta_{e}(q,a)(q')}) \in \R$. However, these random variables are clearly not independent and thus Hoeffding's inequality cannot be applied. Nonetheless, since we can show that the expected value of all of these random variables is at most $\eta < 0$, we can adapt the proof of Hoeffding's inequality to show that, regardless of the strategy $\sigma$, with probability exponentially small in $n$, the sum of the $n$ first random variables $X_1,\ldots,X_n$ is not much higher than $n \cdot \eta < 0$.
		
		Overall, if we appropriately choose $n$, we obtain that, in the environment $e$, for all strategies $\sigma$, the probability to visit at least $n$ $(e,e')$-distinguishing state-action pairs while having the belief in the environment $e'$ never dropping below $\varepsilon$ to be at most $\varepsilon > 0$. The quantity $m(\MEMDP,\varepsilon)$ is chosen such that if at least $m(\MEMDP,\varepsilon)$ distinguishing state-action pairs are visited, for all environments $e \in E$, there is an environment $e' \in E$ such that at least $n$ $(e,e')$-distinguishing state-action pairs are visited. The theorem follows.
	\end{proof}

	\begin{example}
		Let us consider our running example with environments $E_1$ and $E_2$ as above, except that we consider a small $\alpha_0 > 0$, so that with high probability, the strategy can gather many observations before having to make a guess. As we have seen in the above example on the update of the belief, visiting state $C_2$ increases the belief in environment $E_2$, and visiting state $C_1$ increases the belief in environment $E_1$. A consequence of the above theorem is that, with high probability, the number of visits to the states revealing a card is massively skewed towards either $C_1$ (and the belief in $E_2$ drops close to 0) or $C_2$ (and the belief in $E_1$ drops close to 0). Once this occurs, we can almost harmlessly truncate the belief (recall Lemma~\ref{lem:small_belief_change_ok}) and obtain a regular MDP, in which we can compute the value in polynomial time.
		%, and we assume the uniform prior $(\frac{1}{2},\frac{1}{2})$ over the two environments. Since $\alpha_0=1$, the player is forced to guess the duplicated card after a single draw, and therefore can no longer achieve a winning probability arbitrarily close to $1$.
		
		%Under the uniform prior, the optimal decision rule is immediate: if the observed card is $1$, the player should guess $1$, and if the observed card is $2$, the player should guess $2$. This strategy is optimal, and its value is equal to $...$.
	\end{example}
	
	%The detailed proof of Theorem~\ref{thm:main_enough_actions_small_belief}. Note that we start with an adaptation of ...
	
	\paragraph{Algorithm solving the gap problem.} %Theorem~\ref{thm:main_enough_actions_small_belief} gives us that if we see enough distinguishing state-action pairs, then with high probability there is an environment whose belief drops near 0; in which case, we can truncate the belief (in $\Dist(E)$) into a new environment belief with smaller support, while only slightly changing the $\msf{pr}$-value (by Lemma~\ref{lem:small_belief_change_ok}). 
	
	%Building on this idea, w
	We have designed Algorithm~\ref{algo:MEMDP_parity} that, given an MEMDP $\MEMDP$, a parity objective $W$, a prior environment belief $b \in \Dist(E)$ and some $\gamma > 0$, computes, for every state $q \in Q$, a value $v_q$ such that $|\msf{val}^{\msf{pr}}_q(\MEMDP,b,W) - v_q| \leq \gamma$. The algorithm takes as additional argument an integer $n \in \N$ which bounds the maximum number of distinguishing state-action pairs that can be visited before the environment belief is truncated. The algorithm should be called with $n := m(\MEMDP,\frac{\gamma}{3|E|})$\footnote{We use $\frac{\gamma}{3|E|}$ instead of $\gamma$ because the algorithm performs several approximations by truncating environment beliefs, the sum of all these approximations should total at most $\gamma$.}. This algorithm works recursively on the number of environments in the support of the belief; the base case is handled in Line 1: in an MDP with parity objectives, we can compute in polynomial time the values of all states. Then, if we have visited enough distinguishing state-action pairs ($n = 0$) or the smallest positive belief is sufficiently close to 0, we recursively call the algorithm on a truncated environment belief with smaller support (Lines 2-4). Otherwise, for each distinguishing state-action pair $(q,a)$, and states $q' \in Q$ such that $p[b,q,a](q') > 0$ (recall Definition~\ref{def:belief_env_update}), we compute the value $v_{q,a,q'} \in [0,1]$ of that state $q'$ with an updated belief $\lambda[b,q,a,q'] \in \Dist(E)$ by recursively calling the algorithm with a bound $n'$ that is either reset to $m(\MEMDP,\frac{\gamma}{3|E|})$ if the support of the belief has shrunk, or equal to $n-1$ otherwise (Lines 5-8). We then transform the MEMDP into an MDP by adding two fresh sink states $q_{\msf{win}}$ and $q_{\msf{lose}}$ redirecting all transitions $(q,a,q') \in Q \times A \times Q$ such that $(q,a)$ is a distinguishing state-action pair to the state $q_{\msf{win}}$ with probability $v_{q,a,q'}$ and to the state $q_{\msf{lose}}$ with probability $1 - v_{q,a,q'}$; we also modify the parity objective $W$ into a parity objective $W'$ for which looping on $q_{\msf{win}}$ is winning and looping on $q_{\msf{lose}}$ is losing (Line 9). (This transformation is formally described in Appendix~\ref{appen:def-MDP-from-MEMDP}.) We can finally compute the value of the states in that MDP with the objective $W'$ (Line 10). They correspond to the values of the MEMDP $\MEMDP$.
	
	\begin{algorithm}[t]
		\caption{$\textsf{MEMDP-Prior-Parity}(\MEMDP,W,b,n,\gamma)$}
		\label{algo:MEMDP_parity}
		\textbf{Input}: MEMDP $\MEMDP = (Q,A,E,(\delta_e)_{e \in E})$, parity objective $W$, belief $b \in \Dist(E)$, $n \in \N$, $\gamma > 0$
		\begin{algorithmic}[1]
			\If{$|\Supp(b)| =1$}
			\Return $\textsf{MDP-Parity}(\MEMDP,W): Q \to [0,1]$
			\EndIf
			\If{$n=0$ \textsf{or }$\min_{e \in \Supp(b)} b(e) \leq \frac{\gamma}{3|E|}$}
			\State $b' \gets \msf{Truncate}(b)$ \Comment{As in Lemma~\ref{lem:small_belief_change_ok}: $\Supp(b') \subsetneq \Supp(b)$}
			\State
			\Return $\textsf{MEMDP-Prior-Parity}(\MEMDP,W,b',m(\MEMDP,\frac{\gamma}{3|E|}),\gamma)$
			\EndIf
			\For{$(q,a) \in \msf{Dstg}(\MEMDP)$ \textsf{and }$q' \in \Supp(p[b,q,a])$}
			\State $b' \gets \lambda[b,q,a,q']$ \Comment{As in Definition~\ref{def:belief_env_update}}
			\If{$\Supp(b') \subsetneq \Supp(b)$} $n' \gets \frac{\gamma}{3|E|}$ \textbf{else} $n' \gets n-1$
			%\State $x_{q,a,q'} \gets \textsf{MEMDP-Prior-Parity}(\MEMDP,W,b',m,m)(q')$
			\EndIf
			\State $v_{q,a,q'} \gets \textsf{MEMDP-Prior-Parity}(\MEMDP,W,b',n',m)(q')$
			\EndFor
			\State $(\MDP,W') \gets \msf{MDP}\text{-}\msf{from}\text{-}\msf{MEMDP}(\MEMDP,W,b,v)$
			\State
			\Return $\textsf{MDP-Parity}(\MDP,W'): Q \to [0,1]$
		\end{algorithmic}
	\end{algorithm}
	
	Consider now the complexity of Algorithm~\ref{algo:MEMDP_parity}. First, note that, whenever it recursively calls itself, either the support of the belief has shrunk, or the bound (which resets to $m(\MEMDP,\frac{\gamma}{3|E|})$) on the number of visited distinguishing state-action pairs has decreased. Thus, the recursion depth is bounded by $|E| \times m(\MEMDP,\frac{\gamma}{3|E|})$. %Furthermore, each recursive step takes polynomial time. 
	In addition, the number of bits used to describe the environment beliefs grows linearly with the number of updates and truncations. The space taken by the algorithm is in fact in $O(poly(|Q|,|A|,|E|,|\log(\gamma)|,x_b,m(\MEMDP,\frac{\gamma}{3|E|})))$, where $x_b$ is the number of bits to represent %the probabilities in 
	the prior belief $b$. All details on %the correction/complexity of 
	Algorithm~\ref{algo:MEMDP_parity} %(with the formal definition of the above MDP) 
	are given in Appendix~\ref{appen:approximation_algorithm}.
	
	The gap problem can be solved by calling Algorithm~\ref{algo:MEMDP_parity} with $\gamma := \varepsilon/3$, and comparing the result with $\alpha$. Given the bound on $m(\MEMDP,\frac{\gamma}{3|E|})$ from Theorem~\ref{thm:main_enough_actions_small_belief}, we obtain the theorem below.
	\begin{theorem}
		\label{thm:complexity_deciding_gap_problem}
		In an MEMDP $\MEMDP% = (Q,A,E,(\delta_e)_{e \in E})
		$, given a prior belief $b \in \Dist(E)$ (given in binary), a parity objective, a threshold $0 < \alpha < 1$ (given in binary), and a precision $\varepsilon > 0$ (given in binary), the gap problem with $\msf{pr}$-values can be decided in $\msf{EXPSPACE}$. If the probabilities involved in the MEMDP are given in unary, the gap problem can be decided in $\msf{PSPACE}$.
	\end{theorem}
	
	\section{Relating the $\msf{pr}$- and $\msf{uni}$-values}
	\label{sec:relating-the-values}
	The procedure described in \cite{DBLP:conf/icalp/Chatterjee0RS25} that solves the gap problem with $\msf{uni}$-values executes in space doubly exponential in $|Q|$ (the number of states)%, even if the probabilities are given in unary or if the number of environments is fixed (the space taken is doubly exponential in the number of states)
	. On the other hand, we have exhibited an algorithm solving the gap problem with $\msf{pr}$-values in polynomial space when the probabilities are given in unary, in exponential space otherwise. Our goal now is to link the $\msf{pr}$- and $\msf{uni}$-values so that we can use our algorithm %of Section~\ref{sec:solve_gap_problem} 
	to solve more efficiently the gap problem with $\msf{uni}$-values. 
	
	It is clear that the $\msf{uni}$-value is lower than or equal to the $\msf{pr}$-value for any prior belief. In fact, the $\msf{uni}$-value is actually equal to the infimum $\msf{pr}$-value over all possible prior beliefs. %This is formally stated in the theorem below.
	\begin{theorem}%[Proof~\ref{appen:mixed_strategies}]
		\label{thm:approx_value}
		Consider an MEMDP $\MEMDP = (Q,A,E,(\delta_e)_{e \in E})$, a state $q \in Q$, and a parity objective $W \subseteq \msf{Borel}(Q)$. We have: $\msf{val}_q^{\msf{uni}}(\MEMDP,W) = \inf_{b \in \Dist(E)} \msf{val}_q^{\msf{pr}}(\MEMDP,b,W)$. 
		%\begin{equation*}
		%	\msf{val}_q^{\msf{uni}}(\MEMDP,W) = \inf_{b \in \Dist(E)} \msf{val}_q^{\msf{pr}}(\MEMDP,b,W)
		%\end{equation*}
	\end{theorem}
	
	\begin{example}
		We reconsider our running example and modify the two environments $E_1$ and $E_2$ as follows. In both environments, we set \(\alpha_0 = 1\), so the player is required to guess the duplicated card immediately after a single draw. If all other aspects of the environments remain unchanged, the \(\msf{uni}\)-value is \(\frac{2}{3}\); it is attained by the strategy that guesses the environment corresponding to the single observed card. Under a uniform prior over the two environments, the \(\msf{pr}\)-value coincides with this \(\frac{2}{3}\) value.
		We now introduce asymmetric environments. In \(E_1\), card 1 is duplicated twice and card 2 is duplicated once (hence \(\alpha_1 = \frac{3}{5}\)), while 
		in \(E_2\), card 2 is duplicated twice (hence \(\alpha_1 = \frac{1}{4}\)). In this setting, the \(\msf{uni}\)-value remains equal to \(\frac{2}{3}\); however, the infimum of the \(\msf{pr}\)-values is now achieved under a prior \(b\) that is slightly biased towards the less favorable environment, namely \(b(E_1) := \frac{5}{9}\), rather than under the uniform prior over environments.
		%
		%winning probabilities: if the duplicated card is $1$ and the player guesses correctly, then the probability of reaching $W$ is $\beta_1=\frac{1}{2}$; if the duplicated card is $2$ and the player guesses correctly, then the probability of reaching $W$ is $\beta_2=\frac{2}{4}$.
		%
		%Thus, unlike the previous setting, the player can no longer gather additional information and must decide after one observation. Under these assumptions, the prior that minimizes the value under the prior semantics is $(\dots,\dots)$.
	\end{example}
	
	To establish Theorem~\ref{thm:approx_value}, we introduce the notion of mixed strategies. Given a set of states $Q$ and set of actions $A$, a mixed strategy on $(Q,A)$ is a probability distribution $\tau \in \Dist(\msf{Strat}(Q,A))$ over the strategies (even though the set $\msf{Strat}(Q,A)$ is uncountable, the probability distributions that we consider have a countable support). We naturally define, in an MDP $\MDP = (Q,A,\delta)$, the probability measure $\Prb_\rho^\tau[\MDP,\cdot]: \msf{Borel}(Q) \to [0,1]$ induced by a mixed strategy $\tau \in \Dist(\msf{Strat}(Q,A))$: $\Prb_\rho^\tau[\MDP,\cdot] := \sum_{\sigma \in \msf{Strat}(Q,A)} \tau(\sigma) \cdot \Prb_\rho^\sigma[\MDP,\cdot]$. For all Borel objectives $W \in \msf{Borel}(Q)$, we naturally extend to mixed strategies $\tau \in \Dist(\msf{Strat}(Q,A))$ the $\msf{uni}$-values $\msf{val}^{\msf{uni}}_q(\MEMDP,\tau,W) \in [0,1]$ and $\msf{pr}$-values $\msf{val}^{\msf{uni}}_q(\MEMDP,b,\tau,W) \in [0,1]$ given a prior belief $b \in \Dist(E)$. Then, %
	%
	%As a first step towards proving Theorem~\ref{thm:approx_value}, 
	we can relate the supremum $\msf{uni}$-values that mixed strategies can achieve with the infimum $\msf{pr}$-values over all prior beliefs. %This is a direct consequence of a standard generalization of von Neuman's minimax theorem. This is formally stated in the lemma below.
	\begin{lemma}[Proof~\ref{proof:lem_vonBeuman}]
		\label{lem:vonNeuman}
		Consider an MEMDP $\MEMDP$, a state $q \in Q$, and a parity objective $W$. We have: $\sup_{\tau \in \Dist(\msf{Strat}(Q,A))} \inf_{b \in \Dist(E)} \msf{val}^{\msf{pr}}_q(\MEMDP,b,\tau,W) = \inf_{b \in \Dist(E)} \sup_{\tau \in \Dist(\msf{Strat}(Q,A))}$ $\msf{val}^{\msf{pr}}_q(\MEMDP,b,\tau,W)$. Therefore: $ \sup_{\tau \in \Dist(\msf{Strat}(Q,A))} \msf{val}^{\msf{uni}}_q(\MEMDP,\tau,W) = \inf_{b \in \Dist(E)} \msf{val}^{\msf{pr}}_q(\MEMDP,b,W)$.
	\end{lemma}
	\begin{proof}[Proof sktech]
		The first equality is a direct consequence of a standard generalization of von Neuman's minimax theorem \cite{von1947theory}, which holds because the set of environments is finite \cite{sion1958general}. Furthermore, for all non-empty sets $X$ and $f\colon X \to [0,1]$, we have $\sup_{d \in \Dist(X)} \sum_{x \in X} d(x) \cdot f(x) = \sup_{x \in X} f(x)$ and $\inf_{d \in \Dist(X)} \sum_{x \in X} d(x) \cdot f(x) = \inf_{x \in X} f(x)$. %(This also holds for infimum instead of supremum.) 
		Therefore, we have that for all $\tau \in \Dist(\msf{Strat}(Q,A))$, $\inf_{b \in \Dist(E)} \msf{val}^{\msf{pr}}_q(\MEMDP,b,\tau,W) = \msf{val}^{\msf{uni}}_q(\MEMDP,\tau,W)$; and for all $b \in \Dist(E)$, $\sup_{\tau \in \Dist(\msf{Strat}(Q,A))}$ $\msf{val}^{\msf{pr}}_q(\MEMDP,b,\tau,W) = \msf{val}^{\msf{pr}}_q(\MEMDP,b,W)$.
	\end{proof}
	
	To establish Theorem~\ref{thm:approx_value}, it is now sufficient to show that mixed strategies in $\Dist(\msf{Strat}(Q,A))$ do not achieve higher $\msf{uni}$-value than strategies in $\msf{Strat}(Q,A)$, as stated in the lemma below.
	\begin{lemma}[Proof~\ref{appen:mixed_strategies}]
		\label{lem:approx_value}
		Consider an MEMDP $\MEMDP$, a state $q \in Q$, and a parity objective $W \in \msf{Borel}(Q)$. We have: $ \sup_{\tau \in \Dist(\msf{Strat}(Q,A))} \msf{val}^{\msf{uni}}_q(\MEMDP,\tau,W) = \msf{val}_q^{\msf{uni}}(\MEMDP,W)$.
	\end{lemma}
	\begin{proof}[Proof sketch]
		Given any $\tau \in \Dist(\msf{Strat}(Q,A))$, we define %---as in the proof of \cite[Lemma 19]{DBLP:conf/icalp/Chatterjee0RS25}--- 
		$\sigma \in \msf{Strat}(Q,A)$ such that, for all $e \in E$ and Borel objectives $W \in \msf{Borel}(Q)$: $\Prb_q^\tau[\MEMDP[e],W] = \Prb_q^\sigma[\MEMDP[e],W]$. Thus, $\msf{val}^{\msf{uni}}_q(\MEMDP,\tau,W) = \msf{val}_q^{\msf{uni}}(\MEMDP,\sigma,W) \leq \msf{val}_q^{\msf{uni}}(\MEMDP,W)$. The lemma follows.
	\end{proof}
	
	\paragraph{How to use Theorem~\ref{thm:approx_value}.} Lemmas~\ref{lem:vonNeuman} and~\ref{lem:approx_value} together imply Theorem~\ref{thm:approx_value}. This theorem gives us that computing the $\msf{uni}$-value amounts to computing the infimum of $\msf{pr}$-values over all prior beliefs. Furthermore, Lemma~\ref{lem:small_belief_change_ok} gives us that $\msf{pr}$-values induced by two close beliefs are not far-off. Therefore, we can obtain an $\varepsilon$-approximation of the $\msf{uni}$-value by computing the $\msf{pr}$-value for enough prior beliefs that tightly cover the set of all beliefs. %given some $\varepsilon > 0$, if we compute $\varepsilon/2$-approximations of $\msf{pr}$-values w.r.t. beliefs in a set $S_{\varepsilon/2}$ %that $\varepsilon/2$-covers the set of all beliefs (i.e. such that for all beliefs $b$, there is a belief $b' \in S_{\varepsilon/2}$ such that $\msf{Diff}(b,b') \leq \varepsilon/2$, we obtain an $\varepsilon$-approximation of the $\msf{uni}$-value. 
	We deduce the theorem below that significantly improves the %(unconditional) 
	doubly-exponential-space complexity established in \cite{DBLP:conf/icalp/Chatterjee0RS25}. 
	\begin{theorem}[Proof~\ref{proof:thm_complexity_deciding_gap_problem_uni_val}]
		\label{thm:complexity_deciding_gap_problem_uni_val}
		In an MEMDP given a parity objective, a threshold $0 < \alpha < 1$ (in binary), and a precision $\varepsilon > 0$ (in binary), the gap problem with $\msf{uni}$-values can be decided in $\msf{EXPSPACE}$. If the probabilities are given in unary, the gap problem can be decided in $\msf{PSPACE}$.
	\end{theorem}
	\begin{proof}[Proof sketch.]
		Let $N := \lceil \log(\frac{1}{\varepsilon}) \rceil$, $k := |E|$, and $S := \{ b \in \Dist(E) \mid \forall e \in E,\; b(e) \in \{ \frac{x}{k \cdot 2^{N+1}} \mid 0 \leq x \leq k \cdot 2^{N+1}\} \}$. We have $|S| = O(2^{k^2 \cdot (N+1)})$ and, for all $b \in \Dist(E)$, there is $b' \in S$ such that $\msf{Diff}(b,b') \leq \frac{\varepsilon}{2}$. Enumerating all beliefs in $S$ can be done in space polynomial in $k$ and $N$. Furthermore, all the beliefs in $S$ are described with a number of bits polynomial in $k$ and $N$. Therefore, executing Algorithm~\ref{algo:MEMDP_parity} on any belief in $S$ with $\gamma := \varepsilon/2$ takes space exponential in the input (resp. polynomial in the input, if the probabilities are given in unary).
	\end{proof}
	%
	%The complexity of Theorem~\ref{thm:complexity_deciding_gap_problem_uni_val} significantly improves the (unconditional) doubly-exponential-space complexity established in \cite{DBLP:conf/icalp/Chatterjee0RS25}. 
	
	Our goal is now to use Theorem~\ref{thm:approx_value} to derive a lower bound on the complexity of approximating the $\msf{pr}$-value with probabilities written in unary, we have already established a polynomial space upper bound. It was shown in \cite[Theorem 26]{DBLP:journals/corr/RaskinS14} that the gap problem with $\msf{uni}$-values and reachability objectives in two-environment MEMDPs with probabilities written in unary is $\msf{NP}$-hard. Thus, we focus below on how to approximate $\msf{uni}$-values in MEMDPs with two-environments. %In fact, %we are able to establish the lemma below.%show that 
	As argued below, this can be done with only polynomially many calls to an oracle approximating the $\msf{pr}$-value, which allows the transfer of the above $\msf{NP}$-hardness result.
	%The proof of Theorem~\ref{thm:complexity_deciding_gap_problem_uni_val} relies on an enumerative search over exponentially many beliefs. In fact, it is possible to use a binary-search-like approach to more efficiently find, among these exponentially many beliefs, one that minimizes the $\msf{pr}$-value. This approach is still exponential in the number of environments, but is polynomial in the precision. This is particularly interesting when the number of environments is fixed.
	%
	%As stated in the proof sketch of Theorem~\ref{thm:complexity_deciding_gap_problem_uni_val}, we can find a $\varepsilon$-approximation of the $\msf{uni}$-value with exponentially many calls Algorithm~\ref{algo:MEMDP_parity}. In fact, when the number of environments is fixed, we can find an approximate more efficiently, with only polynomially calls to Algorithm~\ref{algo:MEMDP_parity}. This is formally stated in the lemma below.
	\begin{proposition}[Proof~\ref{proof:lem_binary_search}]
		\label{prop:binary_search}
		In an MEMDP $\MEMDP$ with two environments, given a parity objective $W$, and a precision $\varepsilon > 0$, letting $N := \lceil \log(\frac{1}{\varepsilon}) \rceil$, we can compute a $\varepsilon$-approximation of the $\msf{uni}$-value in time polynomial in $N$, with $O(N)$ calls to an oracle %Algorithm~\ref{algo:MEMDP_parity} with $\gamma := \varepsilon/2$ and beliefs described with a number of bits polynomial in $N$.
		computing $\gamma$-approximation of the $\msf{pr}$-values on beliefs and $\gamma$ described with a number of bits polynomial in $N$.
	\end{proposition}
	\begin{proof}[Proof sketch]
		Let $E = \{e_1,e_2\}$. We consider $f\colon [0,1] \to [0,1]$ such that, for all $x \in [0,1]$, $f(x) := \msf{val}^{\msf{pr}}_q(\MEMDP,b_x,W)$, with $b_x \in \Dist(E)$ such that $b_x(e_1) := e$. By Theorem~\ref{thm:approx_value}, the infimum $f$-value is the $\msf{uni}$-value. The function $f$ is quasi-convex: for all $x < y < z \in [0,1]$: $f(y) \leq \max(f(x),f(z))$. This implies that, for all $x < y < z < t$, if $f(z) < f(y)$, then $\inf_{u \in [x,t]} f(u) = \inf_{u \in [y,t]} f(u)$, and  if $f(y) < f(z)$, $\inf_{u \in [x,t]} f(u) = \inf_{u \in [x,z]} f(u)$. We design a binary-search-like procedure based on this observation that finds an approximation of the minimal $f$-value by searching in intervals $[x,t]$ of decreasing length until that length becomes small enough (which is sufficient because $f$ is 1-Lipschitz continuous% by Lemma~\ref{lem:small_belief_change_ok}
		). %As we have only access to an approximation of $f$, we have to carefully design this binary-search-like procedure.
	\end{proof}
	
	%Since our algorithm computing approximated $\msf{pr}$-values already executes in the best case in polynomial space, Proposition~\ref{prop:binary_search} cannot yet be used to improve the complexity established in Theorem~\ref{thm:complexity_deciding_gap_problem_uni_val}. However it can straightforwardly be used to translate an $\msf{NP}$-hardness result. %of solving the gap problem with $\msf{uni}$-values to the hardness of solving the gap problem with $\msf{pr}$-values 
	%We obtain the corollary below.
	\begin{corollary}[Proof~\ref{proof:cor_complexity_deciding_gap_problem_uni_val_hardness}]
		\label{coro:complexity_deciding_gap_problem_uni_val_hardness}
		In two-environment MEMDPs with probabilities in unary, deciding the $\msf{pr}$-value gap problem with parity objectives % in MEMDPs with two environments and probability distributions written in unary
		cannot be done in polynomial time unless $\msf{P} = \msf{NP}$.
	\end{corollary}
	%\begin{proof}[Proof sketch]
	%	It was shown in \cite[Theorem 26]{DBLP:journals/corr/RaskinS14} that the gap problem with $\msf{uni}$-values and reachability objectives in two-environment MEMDPs with probabilities written in unary, is $\msf{NP}$-hard. By Proposition~\ref{prop:binary_search}, this gap problem could be solved in polynomial time if we could solve in polynomial time the gap problem with $\msf{pr}$-values and parity objectives in two-environment MEMDPs with probabilities written in unary. The corollary follows.
	%\end{proof}
	
	\section{Characterization of MEMDPs with entropy}
	\label{sec:entropy}
	When playing in an MEMDP, we have only partial information about where we are: we know the current state, but we do not know the operating environment. In fact, MEMDPs are a special kind of Partially Observable MDPs (POMDP for short), i.e. MDPs in which we play on an underlying set of states to which we have only indirect access via an observation that may be identical for different states. Formally, a POMDP is a tuple $\POMDP = (Q,A,\delta,\Omega,O)$ where $(Q,A,\delta)$ is an MDP, $\Omega$ is a non-empty set of observations, and $O\colon Q \to \Omega$. An MEMDP $\MEMDP$ naturally induces the POMDP $\POMDP(\MEMDP) = (Q',A,\delta,\Omega,O)$ such that $Q' := Q \times E$; for all $(q,e,q',a) \in Q \times E \times Q \times A$, $\delta((q,e),a)((q',e)) := \delta_e(q,a)(q')$; $\Omega := Q$; and for all $(q,e) \in Q \times E$, $O(q,e) := q$. Strategies in POMDPs are functions in $\msf{Strat}(\Omega,A)$. Given any strategy $\sigma \in \msf{Strat}(\Omega,A)$, for any state $q \in Q$, we naturally define the probability measure $\Prb^{\sigma}_q[\POMDP,\cdot]\colon \msf{Borel}(Q) \to [0,1]$. Given $b \in \Dist(Q)$ and $W \in \msf{Borel}(Q)$, we let $\msf{val}(\POMDP,b,W) := \sup_{\sigma \in \msf{Strat}(O,A)} \sum_{q \in Q} b(q) \cdot \Prb_q^\sigma[\POMDP,W]$.
	
	The goal of this section is to characterize MEMDPs among POMDPs. To do so, we study the belief in PODMPs. When playing in a POMDP, we start with an initial belief about the current state\footnote{In POMDPs, we always have an initial belief about the current state. That is why, although MEMDPs in the \textquotedblleft{}prior\textquotedblright{} semantics are special kinds POMDPs, it is not the case of MEMDPs in the \textquotedblleft{}universal\textquotedblright{} semantics.% These two models are incomparable.
	}; that belief is updated according to the actions played and observations gathered. %Note that, contrary to the environment belief update in MEMDPs,  (i.e. ), and it is crucial to update that belief along the play as a function of the  (similarly to Definition~\ref{def:belief_env_update}).
	
	\begin{definition}[Belief in POMDP]
		Consider a POMDP $\POMDP = (Q,A,\delta,\Omega,O)$. A \emph{belief} is a probability distribution $b \in \Dist(Q)$. Given $b \in \Dist(Q)$, $a \in A$, and $o \in \Omega$, we let $p(b,a,o) \in [0,1]$ denote the likelihood of $o$ given $(b,a)$: $p(b,a,o) := \sum_{q \in O^{-1}(o)} \sum_{q' \in Q} b(q') \cdot \delta(q',a)(q)$. 
		
		We let $\msf{Comp}(b,a) := \{ o \in \Omega \mid p(b,a,o) > 0 \} \neq \emptyset$ denote of the set of observations \emph{compatible} with $(b,a)$. For all $o \in \msf{Comp}(b,a)$, we let $\lambda[b,a,o] \in \Dist(Q)$ denote the updated belief, defined by for $q \in Q$: $\lambda[b,a,o](q) := 0$ if $O(q) \neq o$; $\lambda[b,a,o](q) := \frac{\sum_{q' \in Q} b(q') \cdot \delta(q',a)(q)}{p(b,a,o)}$ otherwise.
	\end{definition}
	
	When playing in an MEMDP, if we ever know for sure the current environment (i.e. the environment belief is Dirac), then this will never change. POMDPs induced by MEMDPs are said to be \emph{Dirac-preserving}, i.e. for all $(q,a) \in Q \times A$ (with $q$ seen as a Dirac belief), we have: $\forall o \in \msf{Comp}(q,a):\; |\Supp(\lambda[q,a,o])| = 1$.
	
	%Arbitrary POMDPs are not Dirac-preserving. 
	As pointed out in \cite{DBLP:conf/aips/ChatterjeeCK0R20}, there is an alternative way to capture the behavior of POMDPs induced by MEMDPs via the (information theory) notion of entropy \cite{DBLP:journals/sigmobile/Shannon01} of a belief. Formally, the entropy $H(b)$ of a belief $b \in \Dist(Q)$ is defined by $H(b) := - \sum_{q \in Q} b(q) \log(b(q)) \geq 0$. The entropy of a belief is a measure of the amount of information carried by that belief: the higher the entropy, the less carried information. A belief is Dirac if and only if its entropy is null; the maximal value of the entropy is $\log(|Q|)$, it is achieved by uniform probability distributions. 
	
	In MEMDPs, the amount of information we have about the current environment never shrinks: the more distinguishing state-action pairs we visit, the better we know in which environment we are playing. As established in \cite{DBLP:conf/aips/ChatterjeeCK0R20}, this corresponds to the fact that POMDPs induced by MEMDPs have \emph{non-increasing entropy} i.e. they are such that, for all $b \in \Dist(Q)$ and $a \in A$: $H(b) \geq \sum_{o \in \msf{Comp}(b,a)} p(b,a,o) \cdot H(\lambda[q,a,o])$.
	
	Clearly, POMDPs with non-increasing entropy are Dirac-preserving% (it suffices to apply the above inequality to Dirac beliefs)
	. In fact, the converse is also true: all Dirac-preserving POMDPs have non-increasing entropy. (This is harder to show.)
	\begin{proposition}[Proof~\ref{proof:prop_entropy}]
		\label{prop:entropy}
		A POMDP has non-increasing entropy iff it is Dirac-preserving. 
	\end{proposition}
	
	The benefit of the above proposition is twofold. First, the notion of POMDPs with non-increasing entropy is natural, as this corresponds to POMDPs where the knowledge about the current state is monotonous. However, it is not clear from the definition of POMDPs with non-increasing entropy how to effectively decide if a POMDP satisfies this property. It is now apparent that this can be done in polynomial time, since checking that a POMDP is Dirac-preserving can be done by enumerating all triplets of states, actions and observations. 
	Second, given any Dirac-preserving POMDP with an observation-compatible parity objective\footnote{Observation-compatible parity objectives, or visible objectives, are a common assumption in POMDPs. They define objectives that can be observed by the player: after playing and observing the sequence of observations traversed during the play, the player can determine whether it has won or not; see for instance~\cite{DBLP:conf/mfcs/ChatterjeeDH10}.}, that is, a parity objective induced by a labelling function $f : Q \rightarrow \mathbb{N}$ such that $f(q_1)=f(q_2)$ whenever $O(q_1)=O(q_2)$, and given an initial belief, we can construct an exponentially larger MEMDP, together with a parity objective and an prior environment belief, such that the value in the POMDP coincides with the $\msf{pr}$-value in the MEMDP. This is formally stated below. %This shows that, up to an exponential blow-up, MEMDPs are exactly POMDPs with non-increasing entropy.
	\begin{theorem}[Proof~\ref{proof:thm-Dirac-preserving-pomdp-are-memdps}]
		\label{thm:Dirac-preserving-pomdp-are-memdps}
		Consider a Dirac-preserving POMDP $\POMDP = (Q,A,\delta,\Omega,O)$, an observation-compatible parity objective $W \in \msf{Borel}(Q)$, and an initial belief $b \in \Dist(Q)$. Then, we can compute in exponential time an MEMDP $\MEMDP = (Q',A,E,(\delta_e)_{e \in E})$, a parity objective $W' \in \msf{Borel}(Q')$, and an environment belief $b' \in \Dist(E)$ such that there is a distinguished state $q' \in Q'$ for which $\msf{val}(\POMDP,b,W) = \msf{val}^{\msf{pr}}_{q'}(\MEMDP,b',W')$.
	\end{theorem}
	%We could translate Dirac-preserving POMDPs with arbitrary parity objectives into MEMDPs, but this would come at the cost of defining a parity objective per environment in the MEMDP, which is a different (albeit interesting) setting than the one we have studied in this paper. 
	We obtain that, up to an exponential blow-up, MEMDPs are exactly POMDPs with non-increasing entropy. Furthermore, note that the limit-sure \cite{DBLP:conf/icalp/GimbertO10} and gap \cite{DBLP:journals/ai/MadaniHC03} decision problems on arbitrary POMDPs with observation-compatible parity objectives are undecidable. Thus, Dirac-preserving POMDPs constitute a particularly well-behaved subclass of POMDPs. 
	
	\newpage
	\bibliography{ref}
	\bibliographystyle{plain}

	\newpage
	\appendix
	\section{Complements on Section~\ref{sec:solve_gap_problem}}
	\label{appen:complement_section_3}
	
	\subsection{Proof of Proposition~\ref{prop:value_one}}
	\label{proof:prop_value_one}
	\begin{proof}
		If $\msf{val}^{\msf{uni}}_q(\MEMDP,W) = 1$, then for all $\varepsilon > 0$, there is $\sigma \in \msf{Strat}(Q,A)$ such that, for all $e \in E$, $\Prb_q^\sigma[\MEMDP[e],W] \geq 1 - \varepsilon$. We obtain:
		\begin{equation*}
			\msf{val}^{\msf{pr}}_q(\MEMDP,b,W) \geq \sum_{e \in E} b(e) \cdot \Prb_q^\sigma[\MEMDP[e],W] \geq \sum_{e \in E} b(e) \cdot (1 - \varepsilon) = 1 - \varepsilon
		\end{equation*}
		Thus, $\msf{val}^{\msf{pr}}_q(\MEMDP,b,W) = 1$. Assume now that $\msf{val}^{\msf{pr}}_q(\MEMDP,b,W) = 1$. Let $\varepsilon > 0$. Let $d := \min_{e \in E} b(e) > 0$. There is some $\sigma \in \msf{Strat}(Q,A)$ such that $\sum_{e \in E} b(e) \cdot \Prb_q^\sigma[\MEMDP[e],W] \geq 1 - d \cdot \varepsilon$. Assume towards a contradiction that there is some $e_m \in E$ such that $\Prb_q^\sigma[\MEMDP[e_m],W] < 1 - \varepsilon$. In that case:
		\begin{equation*}
			\sum_{e \in E} b(e) \cdot \Prb_q^\sigma[\MEMDP[e],W] < \sum_{e \in E \setminus \{e_m\}} b(e)+ b(e_m) \cdot (1 - \varepsilon) = 1 - b(e_m) \cdot \varepsilon \leq 1 - d \cdot \varepsilon 
		\end{equation*}
		Hence the contradiction. In fact, for all $e \in E$, we have $\Prb_q^\sigma[\MEMDP[e_m],W] \geq 1 - \varepsilon$. The existence of such a strategy $\sigma \in \msf{Strat}(Q,A)$ is ensured for all $\varepsilon > 0$, therefore we have $\msf{val}^{\msf{uni}}_q(\MEMDP,W) = 1$.
		
		Consider now a strategy $\sigma \in \msf{Strat}(Q,A)$. We have:
		\begin{align*}
			\msf{val}^{\msf{uni}}_q(\MEMDP,\sigma,W) = 1 & \Leftrightarrow \forall e \in E,\; \Prb_q^\sigma[\MEMDP[e],W] = 1 \Leftrightarrow \forall e \in E,\; b(e) \cdot \Prb_q^\sigma[\MEMDP[e],W] = b(e) \\
			& \Leftrightarrow \sum_{e \in E} b(e) \cdot \Prb_q^\sigma[\MEMDP[e],W] = 1 \Leftrightarrow \msf{val}^{\msf{pr}}_q(\MEMDP,b,\sigma,W) = 1
		\end{align*}
	\end{proof}
	
	\subsection{Proof of Lemma~\ref{lem:small_belief_change_ok}}
	\label{proof:lem_small_belief_change_ok}
	\begin{proof}
		First of all, note that:
		\begin{equation*}
			\sum_{\substack{e \in E \\ b(e) \geq b'(e)}} (b(e) - b'(e)) = \sum_{\substack{e \in E}} (b(e) - b'(e)) - \sum_{\substack{e \in E \\ b'(e) \geq b(e)}} (b(e) - b'(e)) = \sum_{\substack{e \in E \\ b'(e) \geq b(e)}} (b'(e) - b(e))
		\end{equation*}
		Hence:
		\begin{align*}
			\msf{Diff}(b,b') & = \frac{1}{2} \sum_{e \in E} |b(e) - b'(e)| = \frac{1}{2} \left(\sum_{\substack{e \in E \\ b(e) \geq b'(e)}} (b(e) - b'(e)) + \sum_{\substack{e \in E \\ b'(e) \geq b(e)}} (b'(e) - b(e))\right) \\
			& = \sum_{\substack{e \in E \\ b(e) \geq b'(e)}} (b(e) - b'(e)) = \sum_{\substack{e \in E \\ b'(e) \geq b(e)}} (b'(e) - b(e))
		\end{align*}
		
		Now, consider a strategy $\sigma \in \msf{Strat}(Q,A)$. We have:
		\begin{align*}
			\sum_{e \in E} b(e) \cdot \Prb_{q}^{\sigma}[\MEMDP[e],W] & = \sum_{e \in E} b'(e) \cdot \Prb_{q}^{\sigma}[\MEMDP[e],W] + \sum_{e \in E} (b(e) - b'(e)) \cdot \Prb_{q}^{\sigma}[\MEMDP[e],W] \\
			& \leq \msf{val}^{\msf{pr}}_q(\MEMDP,b',W) + \msf{Diff}(b,b') %= \msf{SymDiff}(b',b)
		\end{align*}
		As this holds for all strategies $\sigma \in \msf{Strat}(Q,A)$, we obtain: $ \msf{val}^{\msf{pr}}_q(\MEMDP,b,W) \leq \msf{val}^{\msf{pr}}_q(\MEMDP,b',W) + \msf{Diff}(b,b')$. By symmetry, we have: $\msf{val}^{\msf{pr}}_q(\MEMDP,b',W) \leq \msf{val}^{\msf{pr}}_q(\MEMDP,b,W) + \msf{Diff}(b,b')$. 
		
		Let $b \in \Dist(E)$ such that $|\Supp(b)| \geq 2$. Let $e \in E$ such that $\min_{e' \in E} b(e') = b(e)$. There is some $e' \in \Supp(b) \setminus \{ e \}$. We define $\msf{Truncate}(b) \in \Dist(E)$ as follows: $\msf{Truncate}(b)(e) := 0$; $\msf{Truncate}(b)(e') := b(e) + b(e')$, and for all $e'' \in E \setminus \{ e,e' \}$, $\msf{Truncate}(b)(e'') := b(e'')$. Then:
		\begin{align*}
			\msf{Diff}(b,\msf{Truncate}(b)) & = \frac{1}{2} \cdot \sum_{e'' \in E} |b(e'') - \msf{Truncate}(b)(e'')| \\
			& = \frac{1}{2} \cdot \left(|b(e) - \msf{Truncate}(b)(e)| + |b(e') - \msf{Truncate}(b)(e')|\right) \\
			& = b(e)
		\end{align*}
	\end{proof}
	
	\subsection{Proof of Theorem~\ref{thm:main_enough_actions_small_belief}}
	\label{proof:thm_main_enough_actions_small_belief}
	%For the proof of this theorem, we assume w.l.o.g. (up to merging identical environments) that for all $e \neq e' \in E$, we have $\delta_e \neq \delta_{e'}$. 
	
	\subsubsection{Definition of $m(\MEMDP,\varepsilon)$}
	\label{appen:def_bound}
	Let us introduce some notations to express the quantity $m(\MEMDP,\varepsilon)$. 
	\begin{definition}
		Consider an MEMDP $\MEMDP = (Q,A,E,(\delta_e)_{e \in E})$. For $m \in \{\max,\min\}$, we let:
		\begin{equation*}
			\msf{ratio}_{m}(\MEMDP) := m \left\{ \frac{\delta_{e}(q,a)(q')}{\delta_{e'}(q,a)(q')} \mid e,e' \in E,\; (q,a,q') \in Q \times A \times Q,\; \delta_{e'}(q,a)(q') > 0 \right\}
		\end{equation*}
		If the above set is empty, then we set $\msf{ratio}_{\min}(\MEMDP) = 1 = \msf{ratio}_{\max}(\MEMDP)$. If that is not case, we have $\msf{ratio}_{\min}(\MEMDP) < 1 < \msf{ratio}_{\max}(\MEMDP)$. We also let:
		\begin{equation*}
			\msf{ratio}_{\min}^{> 1}(\MEMDP) := \min \left\{ \frac{\delta_{e}(q,a)(q')}{\delta_{e'}(q,a)(q')} \mid e,e' \in E,\; (q,a,q') \in Q \times A \times Q,\; \delta_{e}(q,a)(q') > \delta_{e'}(q,a)(q') > 0 \right\} 
		\end{equation*}
		If the above set is empty, then we let $\msf{ratio}_{\min}^{> 1}(\MEMDP)  := 1$. Otherwise, we have $\msf{ratio}_{\min}^{> 1}(\MEMDP) > 1$.
		
		We also let 
		\begin{equation*}
			p_{\msf{min}}(\MEMDP) := \min \{ \delta_{e}(q,a)(q') \mid (q,a,q') \in Q \times A \times Q,\; \delta_{e}(t,a)(t') > 0 \} > 0
		\end{equation*}
		
		We let:
		\begin{equation*}
			\iota(\MEMDP) := \min\left(\left(\sqrt{\msf{ratio}_{\min}^{> 1}(\MEMDP)} - 1\right)^2,1\right) \text{ and }\eta(\MEMDP) := \log(1 - p_{\msf{min}(\MEMDP)} \cdot \iota(\MEMDP))
		\end{equation*}
		Note that, if $\msf{ratio}_{\min}^{> 1}(\MEMDP) = 1$, then $\iota(\MEMDP) = 0 = \eta(\MEMDP)$; otherwise $\iota(\MEMDP) > 0$ and $\eta(\MEMDP) <0$.
		
		Finally, we let:
		\begin{itemize}
			\item $n_1(\MEMDP,\varepsilon) := \left\lceil \frac{\log\left(\varepsilon/(2 \cdot |E|)\right)}{\log(1 - p_{\min}(\MEMDP))}\right\rceil$
			\item $n_2(\MEMDP,\varepsilon) := \left\lceil 2 |\log(\varepsilon)| + n_1(\MEMDP,\varepsilon) \cdot \log(\msf{ratio}_{\max}(\MEMDP))\right\rceil$
			\item If $\eta(\MEMDP) = 0$, then we let $n_3(\MEMDP,\varepsilon) := 0$, otherwise:
			\begin{equation*}
				n_3(\MEMDP,\varepsilon) := \left\lceil \frac{2 \cdot |\log(\varepsilon/(2 \cdot |E|))| \cdot  (\log(\frac{\msf{ratio}_{\max}(\MEMDP)}{\msf{ratio}_{\min}(\MEMDP)}))^2}{\eta(\MEMDP)^2} + \frac{2 \cdot n_2(\MEMDP,\varepsilon)}{|\eta(\MEMDP)|}\right\rceil
			\end{equation*}
			\item $m(\MEMDP,\varepsilon) := |E| \cdot ( n_1(\MEMDP,\varepsilon) + n_3(\MEMDP,\varepsilon))$
		\end{itemize}
	\end{definition}
	
	Let us show that $m(\MEMDP,\varepsilon)$ satisfies the equality stated in the theorem.
	\begin{lemma}
		Consider an MEMDP $\MEMDP = (Q,A,E,(\delta_e)_{e \in E})$. Let $k := |E|$, $l$ denote the maximum number of bits to write any probability occurring in $\MEMDP$, and $x := \lceil \log(\frac{1}{\varepsilon})\rceil$. Then, $m(\MEMDP,\varepsilon)$ defined above is such that $m(\MEMDP,\varepsilon) = O(poly(k,2^l,x))$.
	\end{lemma}
	%To simplify, if $l$ is the number of bits to write a probability, then it may use $l$ bits to encode the numerator, and $l$ bits to encode the denominator. 
	\begin{proof}
		First, we have $%1 - p_{\msf{min}}(\MEMDP),
		p_{\msf{min}}(\MEMDP) \geq \frac{1}{2^l}$. Thus, $|\log(1 - p_{\min}(\MEMDP))| \geq p_{\min}(\MEMDP) \geq \frac{1}{2^l}$. Hence, $n_1(\MEMDP,\varepsilon) = O(2^l \cdot (x + k))$. Furthermore, $\log(\msf{ratio}_{\max}(\MEMDP)) \leq l$, thus $n_2(\MEMDP,\varepsilon) = O(2^l \cdot l \cdot (x + k))$.
		
		In addition, assume that $\eta(\MEMDP) > 0$, and thus $\msf{ratio}_{\min}^{> 1}(\MEMDP) > 1$. If $\msf{ratio}_{\min}^{> 1}(\MEMDP) \geq 4$, then $\iota(\MEMDP) = 1$; otherwise we have $\sqrt{\msf{ratio}_{\min}^{> 1}(\MEMDP)} - 1 = \frac{\msf{ratio}_{\min}^{> 1}(\MEMDP) - 1}{\sqrt{\msf{ratio}_{\min}^{> 1}(\MEMDP)} + 1} \geq \frac{1}{3}(\msf{ratio}_{\min}^{> 1}(\MEMDP) - 1)$, thus $\left(\sqrt{\msf{ratio}_{\min}^{> 1}(\MEMDP)} - 1\right)^2 \geq = %\msf{ratio}_{\min}^{> 1}(\MEMDP) - 2\sqrt{\msf{ratio}_{\min}^{> 1}(\MEMDP)} + 1 \geq \msf{ratio}_{\min}^{> 1}(\MEMDP) - 1
		\frac{1}{9}(\msf{ratio}_{\min}^{> 1}(\MEMDP) - 1)^2$. Since $\msf{ratio}_{\min}^{> 1}(\MEMDP)$ is equal to a ratio of two probabilities that can be encoded with $l$ bits, we have $\msf{ratio}_{\min}^{> 1}(\MEMDP) - 1 \geq \frac{1}{2^{2l}}$. Therefore, $|\eta(\MEMDP)| = |\log(1 - p_{\msf{min}(\MEMDP)} \cdot \iota(\MEMDP))| \geq p_{\msf{min}(\MEMDP)} \cdot \iota(\MEMDP) \geq \frac{1}{9} \cdot \frac{1}{2^{5l}}$. Hence, we have $n_3(\MEMDP,\varepsilon) = O(2^{10l} \cdot (4l)^2 \cdot (x + k) + 2^{6l} \cdot l \cdot (x + k))$. The lemma follows.
	\end{proof}
	
	\subsubsection{Proof of the theorem}
	\label{appen:actual_proof}
	We define the notion of run valuations (i.e. functions mapping runs to real values) in MDPs. Specifically, we proceed in two steps. First, we show that 
	if a history valuation is of a specific kind, then the probability that the sum of visited history values is far from the expected value of this sum is small, which is done by adapting to our setting Hoeffding's inequality. Second, we show that the result of Theorem~\ref{thm:main_enough_actions_small_belief} can be understood as a statement about this kind of history valuations. Let us start by formally introducing the notion of history valuations.
	
	\begin{definition}
		Consider an MDP $\MDP = (Q,A,\delta)$. Let $a,b \in \R$ such that $a < 0 < b$. An $[a,b]$-valuation is a function $v: Q \cdot (A \cdot Q)^* \to [a,b]$. Given $a \leq \eta \leq 0$, an $[a,b]$-valuation $v: Q \cdot (A \cdot Q)^* \to [a,b]$ is $\eta$-bounded if for all $\rho \in q \cdot (A \times Q)^*$ and $\alpha \in A$:
		\begin{itemize}
			\item either for all $q' \in Q$, we have $v(\rho \cdot (\alpha,q')) = 0$;
			\item or $\sum_{q' \in Q} \delta(q,\alpha)(q') \cdot v(\rho \cdot (\alpha,q')) \leq \eta$; we let $\msf{NZ}_v \subseteq Q \times A$ (\textquotedblleft{}Not-Zero\textquotedblright{}) denote the set of such state-action pairs $(q,\alpha)$.
		\end{itemize}
		We also introduce several other useful notations. First, we let:
		\begin{equation*}
			\msf{NZ}^{+}_v := \{ \pi = (q_0,a_0) \cdots (q_n,a_n) \in (Q \cdot A)^+ \mid (q_n,a_n) \in \msf{NZ}_v \text{ and }\forall i \leq n-1,\; (q_i,a_i) \notin \msf{NZ}_v \} 
		\end{equation*} 
		
		Furthermore, for all $\rho = (q_0,a_0) \cdots \in (Q \cdot A)^\omega$, we let:
		\begin{equation*}
			\msf{IndNZ}_v(\rho) := \{ i \in \N \mid (q_i,a_i) \in \msf{NZ}_v \}
		\end{equation*}
		For all $n \in \N$, we let $\msf{IndNZ}_v^n(\rho) \subseteq \msf{IndNZ}_v(\rho)$ denote the set of the $\min(n,|\msf{IndNZ}_v(\rho)|)$ smallest indices in $\msf{IndNZ}_v(\rho)$. %(Thus if $|\msf{IndNZ}_v(\rho)| \geq n$, then $|\msf{IndNZ}^n_v(\rho)| = |\msf{IndNZ}_v(\rho)|$.) 
		We also let $\msf{AZ}_v^\omega := \{ \rho \mid \msf{IndNZ}_v(\rho) = \emptyset \}$ (\textquotedblleft{}Always Zero\textquotedblright{}).
		
		Finally, for all $n \in \N$, we define the function $\msf{s}^n_v: (Q \cdot A)^\omega \to \R$ and $\msf{c}^n_v: (Q \cdot A)^\omega \to \N$ respectively summing the $n$ first values occurring after $\msf{NZ}_v$-actions and counting the number of visited $\msf{NZ}_v$-actions. Specifically, for all $\rho \in (Q \cdot A)^\omega$, we let:
		\begin{align*}
			\msf{s}_v^n(\rho) & := \sum_{i \in \msf{IndNZ}^n_v(\rho)} v(\rho_{\leq i}) \in [n \cdot a,n \cdot b] \\
			\msf{c}_v^n(\rho) & := |\msf{IndNZ}^n_v(\rho)| \leq n
		\end{align*}
		\iffalse
		\begin{align*}
			f_v^n(\rho) := \begin{cases}
				\sum_{i \in \msf{IndNZ}^n_v(\rho)} v(\rho_{\leq i}) & \text{ if }|\msf{Pos}^n_v(\rho)| = n \\
				- \infty & \text{ otherwise }  
			\end{cases}
		\end{align*}
		\fi
	\end{definition}
	
	\subsection{Adapting Hoeffding's lemma}
	\label{subproof:hoeffding}
	Now that we have introduced run valuations, our first goal is to establish the lemma below.
	\begin{lemma}
		\label{lem:small_proba_far_from_expected_value}
		Consider an MDP $\MDP = (Q,A,\delta)$ and a state $q \in Q$. Let $a < 0 < b \in \R$ and let $a \leq \eta \leq 0$. Consider an $[a,b]$-valuation $v: Q \cdot (A \cdot Q)^* \to [a,b]$ that is $\eta$-bounded. For all $n \in \N$, for all $t > 0$, for all strategies $\sigma \in \msf{Strat}(Q,A)$, we have:
		\begin{equation*}
			\Prb_q^\sigma[\MDP,\{ \msf{s}_v^n \geq n \cdot \eta + t \} \cap \{ \msf{c}_v^n = n \}] \leq \exp\left(-\frac{2t^2}{n (b-a)^2}\right)
		\end{equation*}
	\end{lemma}
	Since $\eta$ is an upper bound on the expected value of the valuation $v$ after visiting any $\msf{NZ}_v$-action, the event $\{ \msf{s}_v^n \geq n \cdot \eta + t \}$ corresponds to those infinite paths for which the sum of the $n$ first $v$-values occurring after $\msf{NZ}_v$-actions (computed by $\msf{s}_v^n$) is separated by at least $t$ from the expected value. 
	
	This lemma can be seen as an adaptation of Hoeffding's inequality to our setting of $\eta$-bounded run valuations. However, a priori, Hoeffding's inequality cannot be applied right away because the random variables describing the stochastic evolution of $\msf{s}_v^n$ are not independent. Nonetheless the proof of this lemma follows steps very similar to those taken to establish Hoeffding's inequality; in particular, it relies on Hoeffding's lemma, recalled below.
	\begin{lemma}[Hoeffding's lemma]
		\label{lem:hoeffding_s_lemma}
		Consider a finite set $Q$, some $a \leq b \in \R$, a valuation $v\colon Q \to [a,b]$ and a distribution $d \in \Dist(Q)$. Let $\mu := \sum_{q \in Q} d(q) \cdot v(q)$. For all $s > 0$, we have:
		\begin{equation*}
			\sum_{q \in Q} d(q) \cdot \exp(s(v(q) - \mu)) \leq \exp\left(\frac{s^2 (b-a)^2}{8}\right)
		\end{equation*}
	\end{lemma}
	
	We now establish the lemma below, by induction by invoking the above lemma at each induction step.
	\begin{lemma}
		\label{lem:inequality_expected_value_hoeffding}
		Consider an MDP $\MDP = (Q,A,\delta)$ and a state $q \in Q$. Let $a < 0 < b \in \R$ and let $a \leq \eta \leq 0$. Consider an $[a,b]$-valuation $v: Q \cdot (A \times Q)^* \to [a,b]$ that is $\eta$-bounded. For all $n \in \N$ and $M \geq 0$, we define the function $f_v^{n,M}: (Q \cdot A)^\omega \to \R$ as follows, for all $\rho \in (Q \cdot A)^\omega$:
		\begin{align*}
			f_v^{n,M}(\rho) := \begin{cases}
				\msf{s}_v^n(\rho) & \text{ if }%|\msf{Pos}^n_v(\rho)| = n \\
					\msf{c}_v^n(\rho) = n \\
					%|\msf{IndNZ}_v(\rho)| = n \\
				- M & \text{ otherwise }  
			\end{cases}
		\end{align*}
		\iffalse
		\begin{align*}
			f_v^{n,M}(\rho) := \sum_{i \in \msf{IndNZ}^n_v(\rho)} v(\rho_{\leq i}) - M \cdot (n - |\msf{IndNZ}^n_v(\rho)|)
		\end{align*}
		\fi
		Let $s > 0$. For all $n \in \N$, for all $M \geq 0$, for all strategies $\sigma \in \msf{Strat}(Q,A)$, we have:
		\begin{equation*}
			\Exp_q^\sigma[\MDP,\exp(s(f_v^{n,M} - n \cdot \eta))] \leq \exp\left(n\frac{s^2 (b-a)^2}{8}\right) + n \cdot \exp\left(s\left(-M + n \cdot \max(|a|,|b|)\right) + n\frac{s^2 (b-a)^2}{8}\right)
		\end{equation*}
	\end{lemma}
	\begin{proof}
		%We prove that, for all $\rho \in Q \cdot (A \cdot Q)^*$, $\Exp_\rho^\sigma[\MDP,\exp(f_v^{n,M} - n \cdot \eta)] \leq \exp\left(\frac{s^2 n (b-a)^2}{8}\right) + n \cdot \exp(-M - n \cdot a)$ by induction on $n$.
		We prove the result by induction on $n$. The case $n = 0$ is straightforward since $f_v^{0,M} = 0$, for all $M \geq 0$. Assume now that the property holds for some $n \in \N$. Let $M \geq 0$. Note that, for all $\pi \in \msf{NZ}_v^{+}$, for all $q' \in Q$, and for all $\rho \in q' \cdot A \cdot (Q \cdot A)^\omega$, we have:
		\begin{equation*}
			f_v^{n+1,M}(\pi \cdot \rho) = v(\pi \cdot q') + f_v^{n,M+v(\pi \cdot q')}(\rho)
		\end{equation*}
		
		For all $\rho = (q_0,a_0) \cdots (q_n,a_n) \in \msf{NZ}_v^{+}$, we let $\delta(\rho) := \delta(q_n,a_n)$. Then, by Hoeffding's lemma% (recalled in Lemma~\ref{lem:hoeffding_s_lemma})
		, for all $\rho \in \msf{NZ}_v^+$, since $\mu := \sum_{q' \in Q}  \delta(\rho)(q') \cdot v(\rho \cdot q') \leq \eta \leq 0$, we have:
		\begin{equation*}
			\sum_{q' \in Q}  \delta(\rho)(q') \cdot \exp(s(v(\rho \cdot q')) - \eta) \leq \sum_{q' \in Q}  \delta(\rho)(q') \cdot \exp(s(v(\rho \cdot q')) - \mu) \leq  \exp\left(\frac{s^2 (b-a)^2}{8}\right)
		\end{equation*}
		
		Let $c := \max(|a|,|b|)$. Note that since $a \leq \eta \leq b$, we have $-\eta \leq c$. For all $\rho \in \msf{NZ}_v^{\msf{run}}$, we let $\head{\rho} \in Q$ denote the last state visited by $\rho$; we also let $\sigma_{\rho} \in \msf{Strat}(Q,A)$ be defined by, for all $\pi \in Q \cdot (A \cdot Q)^*$, $\sigma_{\rho}(\pi) := \sigma(\rho \cdot \pi)$. First, we have:
		\begin{align*}
			\Exp_q^{\sigma}[\MDP,\exp(s(f_v^{n+1,M} - (n+1) \cdot \eta))] & = \sum_{\rho \in \msf{NZ}_v^+} \Prb_q^{\sigma}[\MDP,\rho] \cdot \Exp_{\head{\rho}}^{\sigma_{\rho}}[\MDP,\exp(s(f_v^{n+1,M} - (n+1) \cdot \eta))] \\
			& + \Prb_q^{\sigma}[\MDP,\msf{AZ}_v^\omega] \cdot \exp(s (-M - (n+1) \cdot \eta)) 
		\end{align*}
		with 
		\begin{equation*}
			\Prb_q^{\sigma}[\MDP,\msf{AZ}_v^\omega] \cdot \exp(s (-M - (n+1) \cdot \eta)) \leq \exp\left(s(-M + (n+1) \cdot c) + (n+1) \cdot \frac{s^2 (b-a)^2}{8}\right)
		\end{equation*}
		Let $\rho \in \msf{NZ}_v^+$. We have:
		\begin{align*}
			\Exp_{\head{\rho}}^{\sigma_{\rho}}[\MDP,\exp(s&(f_v^{n+1,M} - (n+1) \cdot \eta))] = \sum_{q' \in Q}  \delta(\rho)(q') \cdot \Exp_{q'}^{\sigma_{\rho}}[\MDP,\exp(s(f_v^{n,M+v(\rho \cdot q')} - n \cdot \eta + v(\rho \cdot q') - \eta))]
		\end{align*}
		Now, let $q' \in Q$ and $x := \Exp_{q'}^{\sigma_{\rho}}[\MDP,\exp(s(f_v^{n,M+v(\rho \cdot q')} - n \cdot \eta + v(\rho \cdot q') - \eta))]$. We have:
		\begin{align*}
			x & = \exp(s(v(\rho \cdot q') - \eta)) \cdot \Exp_{q'}^{\sigma}[\MDP,\exp(s(f_v^{n,M+v(\rho \cdot q')} - n \cdot \eta))] \\
			& \leq \exp(s(v(\rho \cdot q') - \eta)) \cdot
			%\exp\left(\frac{s^2 (b-a)^2}{8}\right)%\\
			%& \phantom{\leq \sum_{\rho \in \msf{NZ}_v(q)} \Prb_q^{\sigma}[\MDP,\rho] } 
			%\cdot 
			\left[\exp\left(n\frac{s^2 (b-a)^2}{8}\right) + n\exp\left(s\left(- M - v(\rho \cdot q') + n \cdot c + n\frac{s^2 (b-a)^2}{8}\right)\right) \right] \\
			& \leq \exp\left(\frac{s^2 (b-a)^2}{8}\right) \cdot \left[\exp\left(n\frac{s^2 (b-a)^2}{8}\right) + n\exp(s(- M - v(\rho \cdot q') + n \cdot c + n\frac{s^2 (b-a)^2}{8})) \right] \\
			& \leq \exp\left((n+1) \frac{s^2 (b-a)^2}{8}\right) + n \exp\left(s(- M + (n+1) \cdot c) + (n+1) \frac{s^2 (b-a)^2}{8}\right) %+ d \\
			%& = \exp\left((n+1) \frac{s^2 (b-a)^2}{8}\right) + (n+1) \cdot \exp(s(-M + (n+1) \cdot c) + (n+1) \frac{s^2 (b-a)^2}{8})
		\end{align*}
		This holds for all $q' \in Q$, therefore, we have:
		\begin{align*}
			\Exp_{\head{\rho}}^{\sigma_{\rho}}[\MDP,\exp(s(f_v^{n+1,M} - (n+1) \cdot \eta))] & \leq \exp\left((n+1) \frac{s^2 (b-a)^2}{8}\right) \\
			& + n \exp\left(s(- M + (n+1) \cdot c) + (n+1) \frac{s^2 (b-a)^2}{8}\right)
		\end{align*}
		This holds for all $\rho \in \msf{NZ}_v^+$, thus:
		\begin{align*}
			\Exp_q^{\sigma}[\MDP,\exp(s(f_v^{n+1,M} - (n+1) \cdot \eta))] & \leq \exp\left((n+1) \frac{s^2 (b-a)^2}{8}\right) \\
			& + n \exp\left(s(- M + (n+1) \cdot c) + (n+1) \frac{s^2 (b-a)^2}{8}\right) \\
			& + \exp\left(s(-M + (n+1) \cdot c) + (n+1)\frac{s^2 (b-a)^2}{8}\right) \\
			& \leq \exp\left((n+1) \frac{s^2 (b-a)^2}{8}\right) \\
			& + (n+1) \exp\left(s(- M + (n+1) \cdot c) + (n+1) \frac{s^2 (b-a)^2}{8}\right)
		\end{align*}
		This concludes the induction.
	\end{proof}
	
	We can proceed to the proof of Lemma~\ref{lem:small_proba_far_from_expected_value}.
	\begin{proof}
		Let $t > 0$. For all $M > -n  \cdot \eta - t$, we have $\{ \msf{s}_v^n \geq n \cdot \eta + t \} \cap \{ \msf{c}_v^n = n \} = \{ f_v^{n,M} \geq n \cdot \eta + t \}$. Furthermore, for all $s > 0$, by Markov's inequality, we have:
		\begin{align*}
			\Prb_q^\sigma[\MDP,\{f_v^{n,M} \geq n \cdot \eta + t\}] & = \Prb_q^{\sigma}[\MDP,\{s(f_v^{n,M} - n \cdot \eta) \geq s \cdot t\}] \\
			& = \Prb_q^{\sigma}[\MDP,\{\exp(s(f_v^{n,M} - n \cdot \eta)) \geq \exp(s \cdot t)\}] \\
			& \leq \exp(-s \cdot t) \cdot \Exp_q^{\sigma}[\MDP,\exp(s(f_v^{n,M} - n \cdot \eta))]
		\end{align*}
		Since $n \cdot \exp(-s(M + n \cdot \max(|a|,|b|)) + n\frac{s^2 (b-a)^2}{8}) \to_{M \to \infty} 0$, it follows by Lemma~\ref{lem:inequality_expected_value_hoeffding}, that:
		\begin{equation*}
			\Prb_q^\sigma[\MDP,\{ \msf{s}_v^n \geq n \cdot \eta + t \} \cap \{ \msf{c}_v^n \geq n \}] \leq \exp(-s \cdot t) \cdot \exp\left(n\frac{s^2 (b-a)^2}{8}\right) = \exp\left(-st + n\frac{s^2 (b-a)^2}{8}\right)
		\end{equation*}
		This holds for all $s > 0$, in particular for $s = \frac{4t}{n(b-a)^2}$, for which we obtain:
		\begin{equation*}
			\Prb_q^\sigma[\MDP,\{ \msf{s}_v^n \geq n \cdot \eta + t \} \cap \{ \msf{c}_v^n \geq n \}] \leq \exp\left(-\frac{4t^2}{n(b-a)^2} + \frac{16t^2}{8n(b-a)^2}\right) = \exp\left(-\frac{2t^2}{n(b-a)^2}\right)
		\end{equation*} 
	\end{proof}	
	
	\subsection{Valuation from beliefs update}
	\label{subproof:valuation}
	As a preliminary result, let us show that when playing in an environment, the belief in that environment cannot drop to 0.
	\begin{lemma}
		\label{lem:never_z_belief}
		Let $\MEMDP = (Q,A,E,(\delta_e)_{e \in E})$ be an MEMDP. Consider a state $q \in Q$ and a strategy $\sigma \in \msf{Strat}(Q,A)$. For all environments $e \in E$ and beliefs $b \in \Dist(E)$ such that $b(e) > 0$, for all $\rho \in Q \cdot (A \cdot Q)^*$ such that $\Prb_q^\sigma(\Gamma[e],\rho) > 0$, we have: $\msf{lk}_b(\rho)(e) > 0$.
	\end{lemma}
	\begin{proof}
		Let us prove the result by induction on $\rho \in Q \cdot (A \cdot Q)^*$. This clearly holds if $\rho \in Q$. Now, assume that this holds for some $\rho \in Q \cdot (A \cdot Q)^*$, and let $(a,t) \in A_{\head{\rho}} \times Q$. If $\Prb_q^\sigma(\Gamma[e],\rho \cdot (a,t)) > 0$, then $\Prb_q^\sigma(\Gamma[e],\rho) \cdot \sigma(\rho)(a) \cdot \delta_e(\head{\rho},a)(t) > 0$, and thus $\msf{lk}_b(\rho)(e) > 0$. Hence: $\msf{lk}_b(\rho \cdot (a,t))(e) = \frac{1}{p[\msf{lk}_b(\rho),\head{\rho},a](t)} \cdot \msf{lk}_b(\rho)(e) \cdot \delta_e(\head{\rho},a)(t) > 0$.
	\end{proof}

	Our goal is now to apply Lemma~\ref{lem:small_proba_far_from_expected_value} to establish Theorem~\ref{thm:main_enough_actions_small_belief}. To do so, we describe bounded valuations in MDPs that can be obtained from belief update in MEMDPs. Before we define this valuation, let us first introduce the notion of revealing state-action pair.
	\begin{definition}
		Consider an MEMDP $\MEMDP = (Q,A,E,(\delta_e)_{e \in E})$. Let $e \neq e' \in E$. We say that a state-action pair $(q,a) \in Q \times A$ is $e'/e$-revealing if there is some state $q' \in Q$ such that $\delta_{e'}(q,a)(q') = 0$ and $\delta_{e}(q,a)(q') > 0$.
		
		For all $\rho \in Q \cdot (A \cdot Q)^*$, we let $\msf{nb}_{\msf{rev}}(e',e,\rho) \in \N$ denote the number of occurrences of $e'/e$-revealing state-action pairs in $\rho$.
	\end{definition}
	
	We can now define the history valuations that we consider.
	\begin{definition}
		Consider an MEMDP $\MEMDP = (Q,A,E,(\delta_e)_{e \in E})$. Let $q \in Q$ and $e \neq e' \in E$. We define the history valuation $v[e',e]: Q \cdot (A \cdot Q)^* \to \R$ as follows: $v[e',e](q) := 0$ and for all $\rho = q_0 \cdot (a_1,q_1) \cdots (a_n,q_n) \in Q \cdot (A \cdot Q)^*$:
		\begin{align*}
			v[e',e](\rho) := \begin{cases}
				0 & \text{ if }a_n \text{ is }e'/e\text{-revealing} \\
				0 & \text{ if }\delta_{e}(q_{n-1},a_n)(q_n) = 0 \\
				\log(\frac{\delta_{e'}(q_{n-1},a_n)(q_n)}{\delta_{e}(q_{n-1},a_n)(q_n)}) & \text{ otherwise}
			\end{cases}
		\end{align*}
	\end{definition}
	
	By definition, the exponential of the sum of the $v[e',e]$-values visited on a path $\rho$ occurring with positive probability in the MDP $\MEMDP[e]$ (given any strategy) is closely related to the ratio of the beliefs of the environments $e'$ and $e$ at $\rho$. This is formally stated below.
	\begin{lemma}
		\label{lem:sum_v_small_likelyhood}
		Let $\MEMDP = (Q,A,E,(\delta_e)_{e \in E})$ be an MEMDP. Consider two environments $e \neq e' \in E$, and a prior belief $b \in \Dist(E)$ with $b(e) > 0$. Consider any strategy $\sigma \in \msf{Strat}(Q,A)$ and some $\rho \in Q \cdot (A \cdot Q)^*$ such that $\Prb_q^\sigma[\MEMDP[e],\rho] > 0$. Then:%, for all $k \geq |\rho|$ and $\theta \in (Q \cdot A)^\omega$ such that $\rho \sqsubset \theta$, we have:
		\begin{equation*}
			\frac{b(e')}{b(e)} \cdot 2^{\left(\sum_{\pi \sqsubseteq \rho} v[e',e](\pi) \right)} \cdot \msf{ratio}_{\max}(\MEMDP)^{\msf{nb}_{\msf{rev}}(e',e,\rho)} \geq \frac{\msf{lk}_b(\rho)(e')}{\msf{lk}_b(\rho)(e)}
		\end{equation*}
		where we write $\pi \sqsubseteq \rho$ for runs $\pi \in Q \cdot (A \cdot Q)^*$ that are prefixes of $\rho$.
	\end{lemma}
	\begin{proof}
		First of all, for all $\rho \in Q \cdot (A \cdot Q)^*$ such that $\Prb_q^\sigma[\MEMDP[e],\rho] > 0$, we have $\msf{lk}_b(\rho)(e) > 0$, by Lemma~\ref{lem:never_z_belief}. 
		
		Then, we show the result by induction on $\rho \in Q \cdot (A \cdot Q)^*$ such that $\Prb_q^\sigma[\MEMDP[e],\rho] > 0$. This straightforwardly holds for $\rho = q$, since $v[e',e](q) = 0$, $\msf{nb}_{\msf{rev}}(e',e,q) = 0$, and $\msf{lk}_b(q)(f) = b(f)$ for all $f \in \{e,e'\}$. 
		
		Assume now that this property holds for some $\rho \in Q \cdot (A \cdot Q)^*$, and let $(a,q) \in A \times Q$ such that $\rho \cdot (a,q) \in Q \cdot (A \cdot Q)^*$ and $\Prb_q^\sigma[\MEMDP[e],\rho \cdot (a,q)] > 0$, thus $\delta_e(\head{\rho},a)(q) > 0$. We have:
		\begin{equation*}
			\frac{\msf{lk}_b(\rho \cdot (a,q))(e')}{\msf{lk}_b(\rho \cdot (a,q))(e)} = \frac{\frac{1}{\lambda[\msf{lk}_b(\rho),\head{\rho},a,q]} \cdot \msf{lk}_b(\rho)(e') \cdot \delta_{e'}(\head{\rho},a)(q)}{\frac{1}{\lambda[\msf{lk}_b(\rho),\head{\rho},a,q]} \cdot \msf{lk}_b(\rho)(e) \cdot \delta_e(\head{\rho},a)(q)} = \frac{\msf{lk}_b(\rho)(e')}{\msf{lk}_b(\rho)(e)} \cdot \frac{\delta_{e'}(\head{\rho},a)(q)}{\delta_e(\head{\rho},a)(q)}
		\end{equation*}
		
		Then, there are two cases:
		\begin{itemize}
			\item If $(\head{\rho},a)$ is $e'/e$-revealing, then $\msf{nb}_{\msf{rev}}(e',e,\rho \cdot (a,q)) = \msf{nb}_{\msf{rev}}(e',e,\rho) + 1$ and $v[e',e](\rho \cdot (a,q)) = 0$. Furthermore, $\frac{\delta_{e'}(\head{\rho},a)(q)}{\delta_e(\head{\rho},a)(q)} \leq \msf{ratio}_{\max}(\MEMDP)$. Thus: 
			\begin{align*}
				\frac{\msf{lk}_b(\rho \cdot (a,q))(e')}{\msf{lk}_b(\rho \cdot (a,q))(e)} & \leq  \frac{\msf{lk}_b(\rho)(e')}{\msf{lk}_b(\rho)(e)} \cdot \msf{ratio}_{\max}(\MEMDP) \leq \frac{b(e')}{b(e)} \cdot 2^{\left(\sum_{\pi \sqsubseteq \rho} v[e',e](\pi) \right)} \cdot \msf{ratio}_{\max}(\MEMDP)^{\msf{nb}_{\msf{rev}}(e',e,\rho)+1} \\
				& = \frac{b(e')}{b(e)} \cdot 2^{\left(\sum_{\pi \sqsubseteq \rho \cdot (a,q)} v[e',e](\pi) \right)} \cdot \msf{ratio}_{\max}(\MEMDP)^{\msf{nb}_{\msf{rev}}(e',e,\rho \cdot (a,q))}
			\end{align*}
			\item If $a$ is not $e'/e$-revealing, then $\msf{nb}_{\msf{rev}}(e',e,\rho \cdot (a,q)) = \msf{nb}_{\msf{rev}}(e',e,\rho)$ and $v[e',e](\rho \cdot (a,q)) = \log(\frac{\delta_{e'}(\head{\rho},a)(q)}{\delta_{e}(\head{\rho},a)(q)})$, thus $2^{v[e',e](\rho \cdot (a,q))} = \frac{\delta_{e'}(\head{\rho},a)(q)}{\delta_{e}(\head{\rho},a)(q)}$. Therefore:
			\begin{align*}
				\frac{\msf{lk}_b(\rho \cdot (a,q))(e')}{\msf{lk}_b(\rho \cdot (a,q))(e)} & = \frac{\msf{lk}_b(\rho)(e')}{\msf{lk}_b(\rho)(e)} \cdot 2^{v[e',e](\rho \cdot (a,q))} \leq \frac{b(e')}{b(e)} \cdot 2^{\left(\sum_{\pi \sqsubseteq \rho \cdot (a,q)} v[e',e](\pi) \right)} \cdot \msf{ratio}_{\max}(\MEMDP)^{\msf{nb}_{\msf{rev}}(e',e,\rho \cdot (a,q))}
			\end{align*}
		\end{itemize}
		This concludes the induction.
	\end{proof}
	
	In addition, this valuation $v[e',e]$ is $\eta(\MEMDP)$-bounded.
	\begin{lemma}
		\label{lem:valuation_eta_bounded}
		Consider an MEMDP $\MEMDP = (Q,A,E,(\delta_e)_{e \in E})$, and $e \neq e' \in E$. In the MDP $\MEMDP[e]$, $v[e',e]: Q \cdot (A \cdot Q)^* \to \R$ is a $[\log(\msf{ratio}_{\min}(\MEMDP)),\log(\msf{ratio}_{\max}(\MEMDP))]$-valuation that is $\eta(\MEMDP)$-bounded.
	\end{lemma}
	
	The proof of this lemma relies on the lower bound, recalled below, on the difference between the arithmetic and geometric mean.
	\begin{lemma}[Theorem 1 in \cite{aldaz2008selfimprovemvent}]
		\label{lem:lower_bound_airthmetic_geometric}
		Let $1 \leq k \in \N$. For all $1 \leq i \leq k$, consider some $0 \leq x_i \in \R$ and $0 < \alpha_i \in \R$, and assume that $\sum_{i = 1}^k \alpha_i = 1$. Then, we have:
		\begin{equation*}
			\prod_{i = 1}^k x_i^{\alpha_i} \leq \sum_{i = 1}^k \alpha_i \cdot x_i - \sum_{i = 1}^k \alpha_i \left(\sqrt{x_i} - \sum_{i = 1}^k \alpha_k \sqrt{x_k} \right)^2
		\end{equation*}
	\end{lemma}
	
	Let us now proceed to the proof of Lemma~\ref{lem:valuation_eta_bounded}.
	\begin{proof}
		For all $(q,a,q') \in Q \times A \times Q$, if $\delta_{e}(q,a)(q') > 0$, we have:
		\begin{equation*}
			\msf{ratio}_{\min}(\MEMDP) \leq \frac{\delta_{e'}(q,a)(q')}{\delta_{e}(q,a)(q')} \leq \msf{ratio}_{\max}(\MEMDP)
		\end{equation*}
		It follows that $v[e',e]$ is a $[\log(\msf{ratio}_{\min}(\MEMDP)),\log(\msf{ratio}_{\max}(\MEMDP))]$-valuation. 
		
		Consider now some $\rho \in Q \cdot (A \cdot Q)^*$ and $a \in A$ such that $(\head{\rho},a) \in \msf{NZ}_{v[e',e]}$. By definition of $v[e',e]$, the pair $(\head{\rho},a)$ is not $e'/e$-revealing and we have $\delta_e(q,a) \neq \delta_{e'}(q,a)$. Let $S := \Supp(\delta_e(q,a)) \subseteq Q$. For all $s \in S$, we let $p_{s} := \delta_e(\head{\rho},a)(s) > 0$ and $r_{s} := \delta_{e'}(\head{\rho},a)(s) > 0$ (since $a$ is not $e'/e$-revealing). By definition, we have $\sum_{q' \in Q} \delta_e(\head{\rho},a)(q') \cdot v(\rho \cdot (a,q')) = \sum_{s \in S} p_{s} \cdot v(\rho \cdot (a,s))$. Furthermore:
		\begin{equation*}
			\sum_{s \in S} p_{s} \cdot v(\rho \cdot (a,s)) =  \sum_{s \in S} p_{s} \cdot \log\left(\frac{r_{s}}{p_{s}}\right) = \log\left( \prod_{s \in S} \left(\frac{r_{s}}{p_{s}}\right)^{p_{s}} \right)
		\end{equation*}
		Furthermore, by Lemma~\ref{lem:lower_bound_airthmetic_geometric}, we have:
		\begin{equation*}
			\prod_{s \in S} \left(\frac{r_{s}}{p_{s}}\right)^{p_{s}} \leq \sum_{s \in S} p_{s} \cdot \frac{r_{s}}{p_{s}} - \sum_{s \in S} p_{s} \left(\sqrt{\frac{r_{s}}{p_{s}}} - \sum_{s' \in S} p_{s'} \sqrt{\frac{r_{s'}}{p_{s'}}} \right)^2 %\leq 1 - \sum_{s \in S} p_{s} \left(\sqrt{\frac{r_{s}}{p_{s}}} - \sum_{s' \in S} p_{s'} \sqrt{\frac{r_{s'}}{p_{s'}}} \right)^2
		\end{equation*}
		If $\sum_{s \in S} r_{s} < 1$, then $\sum_{s \in S} r_{s} \leq 1 - p_{\min}(\MEMDP) \leq 1 - p_{\min}(\MEMDP) \cdot \iota(\MEMDP)$, thus $\prod_{s \in S} \left(\frac{r_{s}}{p_{s}}\right)^{p_{s}} \leq 1 - p_{\min}(\MEMDP) \cdot \iota(\MEMDP)$% and $\sum_{s \in S} p_{s} \cdot v(\rho \cdot (a,s)) \leq \eta(\MEMDP)$
		.
		
		Assume now that $\sum_{s \in S} r_{s} = 1$. Then, since $\delta_e(q,a) \neq \delta_{e'}(q,a)$, it follows that there some $t \in S$ such that $\delta_{e'}(q,a)(t) = r_t > p_t = \delta_e(q,a)(t)$, and thus $\frac{r_t}{p_t} \geq \msf{ratio}_{\min}^{> 1}(\MEMDP) > 1$, by definition of $\msf{ratio}_{\min}^{> 1}(\MEMDP)$. Furthermore, by concavity of the square root function, and by Jensen's inequality, we have:
		\begin{equation*}
			\sum_{s' \in S} p_{s'} \sqrt{\frac{r_{s'}}{p_{s'}}} \leq \sqrt{\sum_{s' \in S} p_{s'} \frac{r_{s'}}{p_{s'}}} = \sqrt{\sum_{s' \in S} r_{s'}} = 1
		\end{equation*} 
		Therefore:
		\begin{align*}
			\sum_{s \in S} p_{s} \left(\sqrt{\frac{r_{s}}{p_{s}}} - \sum_{s' \in S} p_{s'} \sqrt{\frac{r_{s'}}{p_{s'}}} \right)^2 \geq p_{t} \left(\sqrt{\frac{r_{t}}{p_{t}}} - \sum_{s' \in S} p_{s'} \sqrt{\frac{r_{s'}}{p_{s'}}} \right)^2 \geq p_{\min}(\MEMDP) \left(\sqrt{\msf{ratio}_{\min}^{> 1}(\MEMDP)} - 1 \right)^2
		\end{align*}
		Hence:
		\begin{equation*}
			\prod_{s \in S} \left(\frac{r_{s}}{p_{s}}\right)^{p_{s}} \leq 1 - p_{\min}(\MEMDP) \left(\sqrt{\msf{ratio}_{\min}^{> 1}(\MEMDP)} - 1 \right)^2 \leq 1 - p_{\min}(\MEMDP) \cdot \iota(\MEMDP)
		\end{equation*}
		
		Overall, we obtain:
		\begin{equation*}
			\sum_{q' \in Q} \delta_e(\head{\rho},a)(q') \cdot v(\rho \cdot (a,q')) = \sum_{s \in S} p_{s} \cdot v(\rho \cdot (a,s)) =   \log\left(\prod_{s \in S}\left(\frac{r_{s}}{p_{s}}\right)^{p_{s}}\right)  \leq \log(1 - p_{\min}(\MEMDP) \cdot \iota(\MEMDP)) = \eta(\MEMDP)
		\end{equation*}
		Thus, $v[e',e]$ is $\eta(\MEMDP)$-bounded.
	\end{proof}
	
	Thanks to the above lemma, we will be able to apply the results of Lemma~\ref{lem:small_proba_far_from_expected_value} to the valuation $v[e',e]$; and in so doing, we will show that if enough distinguishing, and not $e'/e$-revealing, actions are visited, then with high probability the ratio of the beliefs in two environments is small enough (this is handled next page). There remains to handle $e'/e$-revealing actions. This is done in the lemma below.
	\begin{lemma}
		\label{lem:too_many_revealing}
		Consider an MEMDP $\MEMDP = (Q,A,E,(\delta_e)_{e \in E})$, and $e \neq e' \in E$. For all $k \in \N$, and prior beliefs $b \in \Dist(E)$, we let:
		\begin{align*}
			\msf{Reveal}^k(\MEMDP,b,e',e) := & \{ \rho \in Q \cdot (A \cdot Q)^* \mid \msf{nb}_{\msf{rev}}(e',e,\rho) \geq k %\exists 1 \leq i_1 < \ldots < i_k \leq n:\; \forall 1 \leq j \leq k,\; a_{i_j} \text{ is }e'/e\text{-revealing} 
			\text{ and } \msf{lk}_b(\rho)(e') \neq 0 \}
		\end{align*}
		
		%Let $p^e_{\msf{min}} := \min \{ \delta_{e}(t,a)(t') \mid t,t' \in Q,\; a \in A_t,\; \delta_{e}(t,a)(t') > 0 \}$. 
		Then, for all $n \in \N$, for all $q \in Q$, for all strategies $\sigma \in \msf{Strat}(Q,A)$, we have:
		\begin{equation*}
			\Prb_q^\sigma(\MEMDP[e],\msf{Reveal}^n(\MEMDP,b,e',e)) \leq (1-p_{\min}(\MEMDP))^n
		\end{equation*}
	\end{lemma}
	\begin{proof}
		We show the result by induction on $n \in \N$. This straightforwardly holds for $n = 0$. Assume now that this property holds for some $n \in \N$. Let $q \in Q$. We let:
		\begin{equation*}
			\hspace*{-0.8cm}
			\msf{FirstRev} := \{ \pi = (q_0,a_0) \cdots (q_n,a_n) \in (Q \cdot A)^* \mid (q_n,a_n) \text{ is revealing and }\forall i \leq n-1,\; (q_i,a_i) \text{ is not revealing} \}
		\end{equation*}
		%$\msf{Paths} := \{ \rho \in \msf{Reveal}^{1}_q(\MEMDP,b,e',e) \mid \forall \pi \sqsubset \rho: \msf{nb}_{\msf{rev}}(e',e,\pi) = 0 \}$. 
		For all $\rho \in \msf{FirstRev}$, we let $\sigma_{\rho} \in \msf{Strat}(Q,A)$ be defined by, for all $\pi \in Q \cdot (A \cdot Q)^*$, $\sigma_{\rho}(\pi) := \sigma(\rho \cdot \pi)$ (similarly to the proof of Lemma~\ref{lem:inequality_expected_value_hoeffding}). For all $\rho \in \msf{FirstRev}$, we consider some $q_\rho \in Q$ such that $\delta_{e'}(\rho)(q_\rho) = 0$ and $\delta_{e}(\rho)(q_\rho) > 0$. Note that $\sum_{q' \in Q \setminus \{q_\rho\}} \delta_{e}(\rho)(q) = 1 - \delta_{e}(\rho)(q_\rho) \leq 1 - p_{\min}(\MEMDP)$. By definition of $\msf{lk}_b$, we have that for all prefixes $\pi$ of some run $\rho \in Q \cdot (A \cdot Q)^*$, if $\msf{lk}_b(\pi)(e') = 0$, then $\msf{lk}_b(\rho)(e') = 0$. Therefore, letting $x := \Prb_q^\sigma(\MEMDP[e],\msf{Reveal}^{n+1}(\MEMDP,b,e',e))$, we have:
		\begin{align*}
			x & = \sum_{\rho \in \msf{FirstRev}} \Prb_q^\sigma(\MEMDP[e],\rho) \cdot \sum_{q' \in Q} \delta_e(\rho)(q') \cdot  \Prb_{q'}^{\sigma_\rho}(\MEMDP[e],\msf{Reveal}^{n}(\MEMDP,\msf{lk}_b(\rho \cdot q')),e',e) \\
			& = \sum_{\rho \in \msf{FirstRev}} \Prb_q^\sigma(\MEMDP[e],\rho) \cdot \sum_{q' \in Q \setminus \{q_\rho\}} \delta_e(\rho)(q') \cdot  \Prb_{q'}^{\sigma_\rho}(\MEMDP[e],\msf{Reveal}^{n}(\MEMDP,\msf{lk}_b(\rho \cdot q')),e',e) \\
			& \leq \sum_{\rho \in \msf{FirstRev}} \Prb_q^\sigma(\MEMDP[e],\rho) \cdot \sum_{q' \in Q \setminus \{q_\rho\}} \delta_e(\rho)(q') \cdot  (1 - p_{\min}(\MEMDP))^n \\
			& \leq \sum_{\rho \in \msf{FirstRev}} \Prb_q^\sigma(\MEMDP[e],\rho) \cdot (1 - p_{\min}(\MEMDP))^{n+1} \leq (1 - p_{\min}(\MEMDP))^{n+1}
		\end{align*}
		This concludes the induction.
	\end{proof}
	
	\paragraph{Proof of Theorem~\ref{thm:main_enough_actions_small_belief}.}
	\label{subproof:proof_thm}
	We now have all the ingredients to establish Theorem~\ref{thm:main_enough_actions_small_belief}.
	\begin{proof}
		Consider and environment $e \in E$, a prior belief $b \in \Dist(E)$, and a strategy $\sigma \in \msf{Strat}(Q,A)$. If there is some environment $e \in E$ such that $b(e) \leq \varepsilon$, then $\msf{NeverSmallBelief}(b,\varepsilon) = \emptyset$. Hence, let us assume that this is not the case. Consider some $e' \neq e \in E$.
		
		We start by realizing that $\{ \msf{nb}_{\msf{rev}}(e',e) \geq n_1(\MEMDP,\varepsilon)\} \cap \msf{NeverSmallBelief}(b,\varepsilon) \subseteq \msf{Reveal}^{n_1(\MEMDP,\varepsilon)}_q(\MEMDP,b,e',e)$, thus, by Lemma~\ref{lem:too_many_revealing}, we have:
		\begin{equation*}
			%\label{eqn:n1}
			\Prb_q^\sigma(\MEMDP[e],\{ \msf{nb}_{\msf{rev}}(e',e) \geq n_1(\MEMDP,\varepsilon)\} \cap \msf{NeverSmallBelief}(b,\varepsilon)) \leq (1 - p_{\min}(\MEMDP))^{n_1(\MEMDP,\varepsilon)} \leq \frac{\varepsilon }{3|E|}
		\end{equation*}
		
		Note now that there two cases. Assume first that $\eta(\MEMDP) = 0$. This is equivalent to $\iota(\MEMDP) = 0$ and $n_3(\MEMDP,\varepsilon) = 0$. This also implies that there are no distinguishing state-action pairs that are not revealing. Therefore:
		\begin{equation*}
			\left\{ \msf{nb}_{\msf{dstg}}(e',e) \geq \frac{m(\MEMDP,\varepsilon)}{|E|} \right\} \subseteq \{ \msf{nb}_{\msf{rev}}(e',e) \geq n_1(\MEMDP,\varepsilon) \}
		\end{equation*}
		
		Hence, we have:
		\begin{align*}
			\Prb_q^\sigma&\left(\MEMDP[e],\left\{ \msf{nb}_{\msf{dstg}}(e',e) \geq \frac{m(\MEMDP,\varepsilon)}{|E|} \right\} \cap \msf{NeverSmallBelief}(b,\varepsilon)\right) \leq \frac{\varepsilon }{|E|}
		\end{align*}
		
		Assume now that $\eta(\MEMDP) < 0$, $\iota(\MEMDP) > 0$ and $n_3(\MEMDP,\varepsilon) \geq 1$. Then, by definition, we have:
		\begin{align*}
			\left\{ \msf{nb}_{\msf{dstg}}(e',e) \geq \frac{m(\MEMDP,\varepsilon)}{|E|} \right\} & \subseteq (\{ \msf{nb}_{\msf{dstg}}(e',e) - \msf{nb}_{\msf{rev}}(e',e) \geq n_3(\MEMDP,\varepsilon) \} \cap \{ \msf{nb}_{\msf{rev}}(e',e) < n_1(\MEMDP,\varepsilon) \}) \\
			& \bigcup \; \{ \msf{nb}_{\msf{rev}}(e',e) \geq n_1(\MEMDP,\varepsilon) \}
		\end{align*}
		
		Consider any $\rho \in Q \cdot (A \cdot Q)^*$ such that $\msf{nb}_{\msf{rev}}(e',e,\rho) \leq n_1(\MEMDP,\varepsilon)$, $\Prb_q^\sigma(\MEMDP[e],\rho) > 0$ and $\sum_{\pi \sqsubseteq \rho} v[e',e](\pi) \leq -n_2(\MEMDP,\varepsilon)$. By Lemma~\ref{lem:sum_v_small_likelyhood}, we have:
		\begin{align*}
			\frac{\msf{lk}_b(\rho)(e')}{\msf{lk}_b(\rho)(e)} & \leq \frac{b(e')}{b(e)} \cdot 2^{\left(\sum_{\pi \sqsubseteq \rho} v[e',e](\pi) \right)} \cdot \msf{ratio}_{\max}(\MEMDP)^{\msf{nb}_{\msf{rev}}(e',e,\rho)} \\
			& \leq \frac{1}{\varepsilon} \cdot 2^{-n_2(\MEMDP,\varepsilon)} \cdot \msf{ratio}_{\max}(\MEMDP)^{\msf{nb}_{\msf{rev}}(e',e,\rho)} \leq \frac{1}{\varepsilon} \cdot \varepsilon^2 = \varepsilon
		\end{align*}
		It follows that, for all $k \in \N$:
		\begin{equation*}
			\Prb_q^\sigma(\MEMDP[e],\{ \msf{s}_{v[e',e]}^k \leq -n_2(\MEMDP,\varepsilon) \} \cap \{ \msf{nb}_{\msf{rev}}(e',e) < n_1(\MEMDP,\varepsilon)\} \cap  \msf{NeverSmallBelief}(b,\varepsilon)) = 0
		\end{equation*}
		%, for all $\rho \in (Q \cdot A)^\omega$ such that $\msf{s}_{v[e',e]}^k(\rho) \leq -n_2$ for some $k \in \N$, then $\rho \notin \msf{NeverSmallBelief}(b,\varepsilon)$.
		%\begin{equation*}
		%	\Prb_q^\sigma(\MEMDP[e],\{ \msf{s}_{v[e',e]}^n \leq -n_2 \} \cap \msf{NeverSmallBelief}(b,\varepsilon)) = 0
		%\end{equation*}
		
		Furthermore, by Lemma~\ref{lem:valuation_eta_bounded}, in the MDP $\MEMDP[e]$, $v[e',e]$ is a $[\log(\msf{ratio}_{\min}(\MEMDP)),\log(\msf{ratio}_{\max}(\MEMDP))]$-valuation that is $\eta(\MEMDP)$-bounded. Hence, by Lemma~\ref{lem:small_proba_far_from_expected_value}, for all $n \in \N$ and $t > 0$, we have:
		\begin{equation*}
			\Prb_q^\sigma[\MEMDP[E],\{ \msf{s}_{v[e',e]}^n \geq n \cdot \eta(\MEMDP) + t \} \cap \{ \msf{c}_v^n \geq n \}] \leq \exp\left(-\frac{2t^2}{n \cdot  (\log(\frac{\msf{ratio}_{\max}(\MEMDP))}{\msf{ratio}_{\min}(\MEMDP))})^2}\right)
		\end{equation*}
		
		Let $n := n_3(\MEMDP,\varepsilon) \in \N$ and $t := -\frac{n \cdot \eta(\MEMDP)}{2} > 0$. Since $n \geq \frac{2 \cdot n_2(\MEMDP,\varepsilon)}{|\eta(\MEMDP)|}$, we have:
		\begin{equation*}
			n \cdot \eta(\MEMDP) + t = \frac{n \cdot \eta(\MEMDP)}{2} \leq -n_2(\MEMDP,\varepsilon)
		\end{equation*}
		Thus, $\{ \msf{s}_{v[e',e]}^n \geq -n_2(\MEMDP,\varepsilon) \} \subseteq \{ \msf{s}_{v[e',e]}^n \geq n \cdot \eta(\MEMDP) + t \}$. Furthermore, since $n \geq 2 \cdot |\log(\varepsilon)| \cdot  (\log(\frac{\msf{ratio}_{\max}(\MEMDP)}{\msf{ratio}_{\min}(\MEMDP)}))^2 \cdot \frac{1}{\eta(\MEMDP)^2}$, we have:
		\begin{equation*}
			-\frac{2t^2}{n \cdot  (\log(\frac{\msf{ratio}_{\max}(\MEMDP))}{\msf{ratio}_{\min}(\MEMDP))})^2} = - \frac{n \cdot \eta(\MEMDP)^2}{2 \cdot  (\log(\frac{\msf{ratio}_{\max}(\MEMDP))}{\msf{ratio}_{\min}(\MEMDP))})^2} \leq \log(\varepsilon) \leq \ln(\varepsilon)
		\end{equation*}
		
		Overall, we obtain:
		\begin{equation*}
			\Prb_q^\sigma[\MEMDP[E],\{ \msf{s}_{v[e',e]}^{n_3(\MEMDP,\varepsilon)} \geq -n_2(\MEMDP,\varepsilon) \} \cap \{ \msf{c}_v^{n_3(\MEMDP,\varepsilon)} \geq n_3(\MEMDP,\varepsilon) \}] \leq \frac{\varepsilon}{2 \cdot |E|}
		\end{equation*}
		
		Hence, letting $X_1 := \{ \msf{nb}_{\msf{rev}}(e',e) < n_1(\MEMDP,\varepsilon) \}$:
		\begin{align*}
			\Prb_q^\sigma&[\MEMDP[E],\{ \msf{c}_v^{n_3(\MEMDP,\varepsilon)} \geq n_3(\MEMDP,\varepsilon) \} \cap X_1 \cap \msf{NeverSmallBelief}(b,\varepsilon)] \\
			& \leq \Prb_q^\sigma[\MEMDP[E],\{ \msf{s}_{v[e',e]}^{n_3(\MEMDP,\varepsilon)} \geq -n_2(\MEMDP,\varepsilon) \} \cap \{ \msf{c}_v^{n_3(\MEMDP,\varepsilon)} \geq n_3(\MEMDP,\varepsilon) \} \cap X_1 \cap \msf{NeverSmallBelief}(b,\varepsilon)] \\
			& + \Prb_q^\sigma[\MEMDP[E],\{ \msf{s}_{v[e',e]}^{n_3(\MEMDP,\varepsilon)} \leq -n_2(\MEMDP,\varepsilon) \} \cap \{ \msf{c}_v^{n_3(\MEMDP,\varepsilon)} \geq n_3(\MEMDP,\varepsilon) \} \cap X_1 \cap \msf{NeverSmallBelief}(b,\varepsilon)] \\
			& \leq \Prb_q^\sigma[\MEMDP[E],\{ \msf{s}_{v[e',e]}^{n_3(\MEMDP,\varepsilon)} \geq -n_2(\MEMDP,\varepsilon) \} \cap \{ \msf{c}_v^{n_3(\MEMDP,\varepsilon)} \geq n_3(\MEMDP,\varepsilon) \}] \\
			& + \Prb_q^\sigma[\MEMDP[E],\{ \msf{s}_{v[e',e]}^{n_3(\MEMDP,\varepsilon)} \leq -n_2(\MEMDP,\varepsilon) \} \cap X_1 \cap  \msf{NeverSmallBelief}(b,\varepsilon)] \\
			& \leq \frac{\varepsilon}{2 \cdot |E|}
		\end{align*}
		
		Finally, by definition of $v[e',e]$, a state-action pair $(q,a) \in Q \times A$ is in $\msf{NZ}_{v[e',e]}$ if and only if it is $(e',e)$-distinguishing and not $e'/e$-revealing. Therefore, for all $\rho \in Q \cdot (A \cdot Q)^*$, we have:
		\begin{equation*}
			\{ \msf{nb}_{\msf{dstg}}(e',e,\rho) - \msf{nb}_{\msf{rev}}(e',e,\rho) \geq n_3(\MEMDP,\varepsilon) \} = \{ \msf{c}_v^{n_3(\MEMDP,\varepsilon)}(\rho) \geq n_3(\MEMDP,\varepsilon) \} %= \{ \msf{c}_v^{n_3(\MEMDP,\varepsilon)} = n_3(\MEMDP,\varepsilon)\}
		\end{equation*}
		
		Therefore, we have:
		\begin{equation*}
			\Prb_q^\sigma(\MEMDP[e],\{ \msf{nb}_{\msf{dstg}}(e',e) - \msf{nb}_{\msf{rev}}(e',e) \geq n_3(\MEMDP,\varepsilon) \} \cap X_1 \cap  \msf{NeverSmallBelief}(b,\varepsilon)) \leq \frac{\varepsilon}{2 \cdot |E|}
		\end{equation*}
		
		We obtain:
		\begin{align*}
			\Prb_q^\sigma&\left(\MEMDP[e],\left\{ \msf{nb}_{\msf{dstg}}(e',e) \geq \frac{m(\MEMDP,\varepsilon)}{|E|} \right\} \cap X_1 \cap \msf{NeverSmallBelief}(b,\varepsilon)\right) \\ 
			& \leq \Prb_q^\sigma(\MEMDP[e],\{ \msf{nb}_{\msf{rev}}(e',e) \geq n_1(\MEMDP,\varepsilon) \} \cap \msf{NeverSmallBelief}(b,\varepsilon)) \\
			& + \Prb_q^\sigma(\MEMDP[e],\{ \msf{nb}_{\msf{dstg}}(e',e) - \msf{nb}_{\msf{rev}}(e',e) \geq n_3(\MEMDP,\varepsilon) \} \cap X_1 \cap  \msf{NeverSmallBelief}(b,\varepsilon)) \\
			& \leq \frac{\varepsilon}{2 \cdot |E|} + \frac{\varepsilon}{2 \cdot |E|} = \frac{\varepsilon }{|E|}
		\end{align*}
		
		This holds for all $e' \neq e \in E$. Furthermore, we have:
		\begin{equation*}
			\{ \msf{nb}_{\msf{dstg}} \geq m(\MEMDP,\varepsilon) \} \subseteq \bigcup_{e' \in E \setminus \{e\}} \left\{ \msf{nb}_{\msf{dstg}}(e,e') \geq \frac{m(\MEMDP,\varepsilon)}{|E|} \right\}
		\end{equation*}
		Therefore:
		\begin{align*}
			\Prb_q^\sigma&(\MEMDP[e],\{ \msf{nb}_{\msf{dstg}} \geq m(\MEMDP,\varepsilon) \} \cap \msf{NeverSmallBelief}(b,\varepsilon)) \\ 
			& \leq \sum_{e' \in E \setminus \{e\}} 	\Prb_q^\sigma\left(\MEMDP[e],\left\{ \msf{nb}_{\msf{dstg}}(e',e) \geq \frac{m(\MEMDP,\varepsilon)}{|E|} \right\} \cap \msf{NeverSmallBelief}(b,\varepsilon)\right) \\
			& \leq \sum_{e' \in E \setminus \{e\}} \frac{\varepsilon }{|E|} \leq \varepsilon 
		\end{align*}
	\end{proof}
	
	\subsection{Complements on Algorithm~\ref{algo:MEMDP_parity}}
	\label{appen:approximation_algorithm}

	\subsubsection{Definition of an MDP from an MEMDP}
	\label{appen:def-MDP-from-MEMDP} 
	We formally define below the MDP constructed in Line 9 of Algorithm~\ref{algo:MEMDP_parity}.
	\begin{definition}
		\label{def:mdp-from-memdp}
		Consider an MEMDP $\MEMDP$, a parity objective $W$, and an environment belief $b \in \Dist(E)$. We let $S(\MEMDP,b) := \{ (q,a,q') \in Q \times A \times Q \mid (q,a) \in \msf{Dstg}(\MEMDP),\; q' \in \Supp(p[b,q,a]) \}$. Consider a valuation $v\colon S(\MEMDP,b) \to [0,1]$. We let $w_{v,b}\colon \msf{Dstg}(\MEMDP) \to [0,1]$ be defined by, for all $(q,a) \in Q \times A$:
		\begin{equation*}
			w_{v,b}(q,a) := \sum_{\substack{q' \in \Supp(p[b,q,a])}} \sum_{e \in E} b(e) \cdot \delta_e(q,a)(q') \cdot v(q,a,q')
		\end{equation*}
		%with $v(q,a,q')$ set to 0 if $v$ is not defined on $(q,a,q')$ (i.e. $(q,a,q') \notin S(\MEMDP,b)$).
		
		Then, we define the
		MDP $\MDP(\MEMDP,W,b,v)$ and parity objective $W'$ (which, together, are denoted $\msf{MDP}\text{-}\msf{from}\text{-}\msf{MEMDP}(\MEMDP,W,b,v)$ in Algorithm~\ref{algo:MEMDP_parity}) as follows:
		\begin{itemize}
			\item $\MDP(\MEMDP,W,b,x) := (Q',A,\delta)$ with $Q' := Q \cup \{ q_{\msf{win}},q_{\msf{lose}} \}$, with $q_{\msf{win}}$ and $q_{\msf{lose}}$ two fresh states not in $Q$; and $\delta: Q' \times A \to \Dist(A)$ such that:
			\begin{itemize}
				\item for all $a \in A$, $\delta(q_{\msf{win}},a)(q_{\msf{win}}) := 1$ and $\delta(q_{\msf{lose}},a)(q_{\msf{lose}}) := 1$;
				\item for all $(q,a) \in Q \times A \setminus \msf{Dstg}(\MEMDP)$, we let $\delta(q,a) := \delta_e(q,a) \in \Dist(Q)$, for any $e \in E$ (they all coincide since $(q,a)$ is a non-distinguishing state-action pair).
				\item for all $(q,a) \in Q \times \msf{Dstg}(\MEMDP)$, we let:
				\begin{align*}
					\delta(q,a)(q_{\msf{win}}) & := w_{v,b}(q,a) \\
					\delta(q,a)(q_{\msf{lose}}) & := 1 - w_{v,b}(q,a)
				\end{align*}
			\end{itemize}
			\item Consider the function $f\colon Q \to \N$ defining the parity objective $W$. The parity objective $W'$ is defined by the function $g\colon Q' \to \N$ that coincides with $f$ on $Q$ and such that $g(q_{\msf{win}}) := 0$ and $g(q_{\msf{lose}}) := 1$; that way, we have $(q_{\msf{win}})^\omega \in W'$ and $(q_{\msf{lose}})^\omega \notin W'$.
		\end{itemize}
	\end{definition}
	
	Before we establish below the property that this construction ensures, we first introduce several notations.
	\begin{definition}
		Consider an MEMDP $\MEMDP$, a parity objective $W$, an environment belief $b \in \Dist(E)$, and a valuation $v\colon S(\MEMDP,b) \to [0,1]$. Let $\msf{FirstD} := \{ \pi = (q_0,a_0) \cdots (q_n,a_n) \in (Q \times A)^+ \mid (q_n,a_n) \in \msf{Dstg}(\MEMDP),\; \forall i \leq n,\; (q_i,a_i) \notin \msf{Dstg}(\MEMDP)\} \subseteq (Q \times A)^+$.	For all $\pi \in \msf{FirstD}$, letting $(q,a) \in Q \times A$ denote the %we let $\hd{Q}{\pi} \in Q$ and $\hd{A}{\pi} \in A$ denote the 
		last state-action pair visited by $\pi$, we let for all $e \in E$: $\delta_e(\pi) := \delta_e(q,a) \in \Dist(Q)$; $\delta(\pi) := \delta(q,a) \in \Dist(Q')$, $w_{v,b}(\pi) := w_{v,b}(q,a) \in [0,1]$, $p[b,\pi] := p[b,q,a] \in \Dist(e)$, and for all $q' \in Q$: $v(\pi,q') := v(q,a,q')$ and $\lambda(b,\pi,q') := \lambda(b,q,a,q')$.
		
		For all strategies $\sigma \in \msf{Strat}(Q,A)$ and $q \in Q$, we let $x(\MEMDP,q,\sigma,b,v) \in [0,1]$ be defined by: 
		\begin{equation*}
			x(\MEMDP,q,\sigma,b,v) := \sum_{\pi \in \msf{FirstD}} \sum_{q' \in \Supp(p[b,\pi])} \sum_{e \in E} \Prb_q^\sigma[\MEMDP[e],\pi] \cdot b(e) \cdot  \delta_{e}(\pi)(q') \cdot |v(\pi,q') - \msf{val}^{\msf{pr}}_{q'}(\MEMDP,\lambda(b,\pi,q'),W)|
		\end{equation*}
	\end{definition}
	
	We establish below the property that this construction ensures.
	\begin{lemma}
		\label{lem:prop-MDP-from-MEMDP}
		Consider an MEMDP $\MEMDP$, a parity objective $W$, an environment belief $b \in \Dist(E)$, and a valuation $v\colon S(\MEMDP,b) \to [0,1]$. For $\MDP := \MDP(\MEMDP,W,b,x)$ and $W'$ the MDP and the parity objective from Definition~\ref{def:mdp-from-memdp}, for all $q \in Q$, we have:
		\begin{equation*}
			|\msf{val}^{\msf{pr}}_{q}(\MEMDP,b,W) - \msf{val}_{q}(\MDP,W')| \leq \sup_{\sigma \in \msf{Strat}(Q,A)} x(\MEMDP,q,\sigma,b)
		\end{equation*}
		with $\msf{val}_{q}(\MDP,W') := \sup_{\sigma \in \msf{Strat}(Q',A)} \Prb_q^\sigma[\MDP,W']$.
	\end{lemma}
	\begin{proof}
		We let $\msf{NoD} := \{ \rho = (q_0,a_0) \cdots (Q \times A)^\omega \mid \forall i \in \N,\; (q_i,a_i) \notin \msf{Dstg}(\MEMDP)\} \subseteq (Q \times A)^\omega$. 
				
		Let $q \in Q$. Consider a strategy $\sigma_Q \in \msf{Strat}(Q,A)$. For all $\pi \in \msf{FirstD}$, we let $\sigma^{\pi}_Q \in \msf{Strat}(Q,A)$ be such that, for all $\pi' \in Q \cdot (A \cdot Q)^*$, $\sigma_Q^{\pi}(\pi') := \sigma_Q(\pi \cdot \pi')$. We have:
		\begin{equation*}
			\forall e \in E,\; \Prb_q^{\sigma_Q}[\MEMDP[e],W] = \Prb_q^{\sigma_Q}[\MEMDP[e],W \cap \msf{NoD}] + \sum_{\pi \in \msf{FirstD}} \sum_{q' \in \Supp(p[b,\pi])} \Prb_q^{\sigma_Q}[\MEMDP[e],\pi] \cdot \delta_e(\pi)(q') \cdot \Prb_{q'}^{\sigma^{\pi}_Q}[\MEMDP[e],W]
		\end{equation*}
		
		Similarly, consider a strategy $\sigma_{Q'} \in \msf{Strat}(Q',A)$. We have:
		\begin{align*}
			\Prb_q^{\sigma_{Q'}}[\MDP,W'] & = \Prb_q^{\sigma_{Q'}}[\MDP,W' \cap \msf{NoD}] + \sum_{\pi \in \msf{FirstD}} \Prb_q^{\sigma_Q}[\MDP,\pi] \cdot w_{v,b}(\pi) \\
			& = \Prb_q^{\sigma_{Q'}}[\MDP,W' \cap \msf{NoD}] + \sum_{\pi \in \msf{FirstD}} \Prb_q^{\sigma_Q}[\MDP,\pi] \cdot \sum_{q' \in \Supp(p[b,\pi])} \sum_{e \in E} b(e) \cdot \delta_e(\pi)(q') \cdot v(\pi,q')
		\end{align*}
		
		Assume now that the strategies $\sigma_Q$ and $\sigma_{Q'}$ are such that, for all prefixes $\pi \in Q \cdot (A \cdot Q)^*$ of runs in $\msf{FirstD}$, we have $\sigma_Q(\pi) = \sigma_{Q'}(\pi)$. In that case, we have:
		\begin{equation*}
			\forall e \in E,\; \Prb_q^{\sigma_Q}[\MEMDP[e],W \cap \msf{NoD}] = \Prb_q^{\sigma_{Q'}}[\MDP,W' \cap \msf{NoD}]
		\end{equation*}
		and, for all $\pi \in \msf{FirstD}$:
		\begin{equation*}
			\forall e \in E,\; \Prb_q^{\sigma_Q}[\MEMDP[e],\pi] = \Prb_q^{\sigma_{Q'}}[\MDP,\pi]
		\end{equation*}
		
		Therefore, we have:
		\begin{equation*}
			\hspace*{-0.8cm}
			\msf{val}^{\msf{pr}}_{q}(\MEMDP,b,\sigma_Q,W) - \Prb_q^{\sigma_{Q'}}[\MDP,W'] %& = \sum_{\pi \in \msf{FirstD}} \Prb_q^{\sigma_Q}[\MDP,\pi] \cdot \left( \sum_{q' \in Q}\sum_{e \in E} b(e) \cdot \delta_e(\pi)(q') \cdot \Prb_{q'}^{\sigma^{\pi}_Q}[\MEMDP[e],W] - w_{v,b}(\pi) \right) \\
			= \sum_{\pi \in \msf{FirstD}} \sum_{q' \in \Supp(p[b,\pi])}\sum_{e \in E} \Prb_q^{\sigma_Q}[\MDP,\pi] \cdot \left( b(e) \cdot \delta_e(\pi)(q') \cdot (\Prb_{q'}^{\sigma^{\pi}_Q}[\MEMDP[e],W] - v(\pi,q')) \right) \\
		\end{equation*}
		
		Now, consider any strategy $\sigma_Q \in \msf{Strat}(Q,A)$ and consider a strategy $\sigma_{Q'} \in \msf{Strat}(Q',A)$ such that for all prefixes $\pi \in Q \cdot (A \cdot Q)^*$ of runs in $\msf{FirstD}$, we have $\sigma_Q(\pi) = \sigma_{Q'}(\pi)$. Note that for all $\pi \in \msf{FirstD}$ and $q' \in \Supp(p[b,\pi])$, we have:
		\begin{equation*}
			\sum_{e \in E} \lambda(b,\pi,q')(e) \cdot \Prb_{q'}^{\sigma_Q^{\pi}}[\MEMDP[e],W] \leq \msf{val}^{\msf{pr}}_{q'}(\MEMDP,\lambda(b,\pi,q'),W)
		\end{equation*}
		with, for all $e \in E$:
		\begin{equation*}
			\lambda(b,t,a,q')(e) = \frac{b(e) \cdot \delta_e(\pi)(q')}{\sum_{e' \in E} b(e') \cdot \delta_{e'}(\pi)(q')}
		\end{equation*}
		Thus:
		\begin{equation*}
			\sum_{e \in E} b(e) \cdot \delta_e(\pi)(q') \cdot \Prb_{q'}^{\sigma_Q^{\pi}}[\MEMDP[e],W] \leq \sum_{e \in E} b(e) \cdot \delta_e(\pi)(q') \cdot \msf{val}^{\msf{pr}}_{q'}(\MEMDP,\lambda(b,\pi,q'),W)
		\end{equation*}
		
		Therefore, we have:
		\begin{align*}
			\msf{val}^{\msf{pr}}_{q}(\MEMDP,b,\sigma_Q,W) & = \sum_{\pi \in \msf{FirstD}} \sum_{q' \in \Supp(p[b,\pi])}\sum_{e \in E} \Prb_q^{\sigma_Q}[\MDP,\pi] \cdot \left( b(e) \cdot \delta_e(\pi)(q') \cdot (\Prb_{q'}^{\sigma^{\pi}_Q}[\MEMDP[e],W] - v(\pi,q')) \right) \\
			& \phantom{ \leq \sum_{\pi \in \msf{FirstD}} \sum_{q' \in \Supp(p[b,\pi])}\sum_{e \in E} \Prb_q^{\sigma_Q}[\MDP,\pi] \cdot } + \Prb_q^{\sigma_{Q'}}[\MDP,W'] \\
			& \leq \sum_{\pi \in \msf{FirstD}} \sum_{q' \in \Supp(p[b,\pi])}\sum_{e \in E} \Prb_q^{\sigma_Q}[\MDP,\pi] \cdot \left(b(e) \cdot \delta_e(\pi)(q') \cdot (\msf{val}^{\msf{pr}}_{q'}(\MEMDP,\lambda(b,\pi,q'),W) - v(\pi,q')) \right) \\
			& \phantom{ \leq \sum_{\pi \in \msf{FirstD}} \sum_{q' \in \Supp(p[b,\pi])}\sum_{e \in E} \Prb_q^{\sigma_Q}[\MDP,\pi] \cdot } + \Prb_q^{\sigma_{Q'}}[\MDP,W'] \\
			& \leq x(\MEMDP,q,\sigma_Q,b) + \Prb_q^{\sigma_{Q'}}[\MDP,W'] \leq \sup_{\sigma \in \msf{Strat}(Q,A)} x(\MEMDP,q,\sigma,b) + \msf{val}_{q}(\MDP,W')
		\end{align*}
		As this holds for all $\sigma_Q \in \msf{Strat}(Q,A)$, we have: $(\msf{val}^{\msf{pr}}_{q}(\MEMDP,b,W) - \msf{val}_{q}(\MDP,W')) \leq \sup_{\sigma \in \msf{Strat}(Q,A)} x(\MEMDP,q,\sigma,b,v)$. 
		
		Consider now any strategy $\sigma_{Q'} \in \msf{Strat}(Q',A)$. Let $\varepsilon > 0$. Let $\sigma_Q \in \msf{Strat}(Q,A)$ be a strategy in $\MEMDP$ that: coincides with the strategy $\sigma_{Q'}$ all prefixes in $Q \cdot (A \cdot Q)^*$ of runs in $\msf{FirstD}$; furthermore, for all $\pi \in \msf{FirstD}$, the strategy $\sigma_Q$ is such that the strategy $\sigma_Q^{\pi} \in \msf{Strat}(Q,A)$ ensures, for all $q' \in \Supp(p[b,\pi])$: 
		\begin{equation*}
			\sum_{e \in E} \lambda(b,\pi,q')(e) \cdot \Prb_{q'}^{\sigma_Q^{\pi}}[\MEMDP[e],W] \geq \msf{val}^{\msf{pr}}_{q'}(\MEMDP,\lambda(b,\pi,q'),W) - \varepsilon
		\end{equation*}
		and thus
		\begin{equation*}
			\sum_{e \in E} b(e) \cdot \delta_e(\pi)(q') \cdot \Prb_{q'}^{\sigma_Q^{\pi}}[\MEMDP[e],W] \geq \sum_{e \in E} b(e) \cdot \delta_e(\pi)(q') \cdot \msf{val}^{\msf{pr}}_{q'}(\MEMDP,\lambda(b,\pi,q'),W) - \varepsilon
		\end{equation*}
		In that case, we have:
		\begin{align*}
			\Prb_q^{\sigma_{Q'}}[\MDP,W'] & = \sum_{\pi \in \msf{FirstD}} \sum_{q' \in \Supp(p[b,\pi])}\sum_{e \in E} \Prb_q^{\sigma_Q}[\MDP,\pi] \cdot \left( b(e) \cdot \delta_e(\pi)(q') \cdot (v(\pi,q') - \Prb_{q'}^{\sigma^{\pi}_Q}[\MEMDP[e],W]) \right) \\
			& \phantom{ \leq \sum_{\pi \in \msf{FirstD}} \sum_{q' \in \Supp(p[b,\pi])}\sum_{e \in E} \Prb_q^{\sigma_Q}[\MDP,\pi] \cdot } + \msf{val}^{\msf{pr}}_{q}(\MEMDP,b,\sigma_Q,W) \\
			& \leq \sum_{\pi \in \msf{FirstD}} \sum_{q' \in \Supp(p[b,\pi])}\sum_{e \in E} \Prb_q^{\sigma_Q}[\MDP,\pi] \cdot \left(b(e) \cdot \delta_e(\pi)(q') \cdot (v(\pi,q') - \msf{val}^{\msf{pr}}_{q'}(\MEMDP,\lambda(b,\pi,q'),W)) \right) \\
			& \phantom{ \leq \sum_{\pi \in \msf{FirstD}} \sum_{q' \in \Supp(p[b,\pi])}\sum_{e \in E} \Prb_q^{\sigma_Q}[\MDP,\pi] \cdot } + \msf{val}^{\msf{pr}}_{q}(\MEMDP,b,\sigma_Q,W) + \varepsilon \\
			& \leq x(\MEMDP,q,\sigma_Q,b) + \msf{val}^{\msf{pr}}_{q}(\MEMDP,b,\sigma_Q,W) + \varepsilon \leq \sup_{\sigma \in \msf{Strat}(Q,A)} x(\MEMDP,q,\sigma,b,v) + \msf{val}^{\msf{pr}}_{q}(\MEMDP,b,W) + \varepsilon
		\end{align*}
		As this holds for all $\sigma_{Q'} \in \msf{Strat}(Q',A)$ and $\varepsilon > 0$, we have: $(\msf{val}_{q}(\MDP,W') - \msf{val}^{\msf{pr}}_{q}(\MEMDP,b,W)) \leq \sup_{\sigma \in \msf{Strat}(Q,A)} x(\MEMDP,q,\sigma,b)$. That concludes the proof.
	\end{proof}
	
	Building on the above lemma, we now establish a result that will be instrumental in the correction proof of Algorithm~\ref{algo:MEMDP_parity} (an induction step). Let us first introduce the central notion used to express the property of interest. 
	\begin{definition}
		\label{def:approximation}
		Consider an MEMDP $\MEMDP$, a parity objective $W$, and an environment belief $b \in \Dist(E)$. For all $n \geq 1$ and $\gamma \in [0,1]$, we let $\msf{N}(\MEMDP,b,n,\gamma) := \{ \msf{nb}_{\msf{dstg}} \geq n \} \cap \msf{NeverSmallBelief}(b,\gamma)$. Then, for all states $q \in Q$ and $\eta \geq 0$, we say that a real value $x_q \in [0,1]$ is a $(n,\gamma,\eta,b,q)$-approximation if:
		\begin{equation*}
			|\msf{val}^{\msf{pr}}_{q}(\MEMDP,b,W) - x_q| \leq \eta + \msf{val}^{\msf{pr}}_{q}(\MEMDP,b,\msf{N}(\MEMDP,b,n,\gamma))
		\end{equation*}
		\iffalse
		\begin{itemize}
			\item $|\msf{val}^{\msf{pr}}_{q}(\MEMDP,b,W) - x_q| \leq \gamma + \msf{val}^{\msf{pr}}_{q}(\MEMDP,b,\{ \msf{nb}_{\msf{dstg}} \geq n \} \cap \msf{NeverSmallBelief}(b,\gamma))$; and
			\item For all strategies $\sigma \in \msf{Strat}(Q,A)$, we have:
			\begin{equation*}
				\msf{val}^{\msf{pr}}_{q}(\MEMDP,b,\sigma,\{ \msf{nb}_{\msf{dstg}} \geq n \} \cap \msf{NeverSmallBelief}(b,\gamma)) \geq |\msf{val}^{\msf{pr}}_{q}(\MEMDP,b,\sigma,W) - x_q| - \gamma
			\end{equation*} 		
		\end{itemize}
		\fi
	\end{definition}
	
	We can state and prove the desired lemma.
	\begin{lemma}
		\label{lem:induction-step}
		Consider an MEMDP $\MEMDP$, a parity objective $W$, and an environment belief $b \in \Dist(E)$. Let $n \in \N$ and $\gamma \in [0,1]$ such that $\min_{e \in \Supp(b)} b(e) > \gamma$. Let $\eta \geq 0$. Consider a valuation $v\colon S(\MEMDP,b) \to [0,1]$ and assume that, for all $(q,a,q') \in S(\MEMDP,b)$, the value $v(q,a,q')$ is a $(n,\gamma,\eta,\lambda(b,q,a,q'),q')$-approximation. Then, for $\MDP := \MDP(\MEMDP,W,b,x)$ and $W'$ the MDP and the parity objective from Definition~\ref{def:mdp-from-memdp}, for all $q \in Q$, the value $\msf{val}_{q}(\MDP,W')$ is a $(n+1,\gamma,\eta,b,q)$-approximation.
	\end{lemma}
	\begin{proof}
		By Lemma~\ref{lem:prop-MDP-from-MEMDP}, we have:
		\begin{equation*}
			|\msf{val}^{\msf{pr}}_{q}(\MEMDP,b,W) - \msf{val}_{q}(\MDP,W')| \leq \sup_{\sigma \in \msf{Strat}(Q,A)} x(\MEMDP,q,\sigma,b,v)
		\end{equation*}
		Let $\sigma \in \msf{Strat}(Q,A)$. %Our goal is to show that $x(\MEMDP,q,\sigma,b,v) \leq \gamma + \msf{val}^{\msf{pr}}_{q}(\MEMDP,b,\msf{N}(\MEMDP,b,n+1,\gamma))$. 
		For all $\pi \in \msf{FirstD}$ and $q' \in \Supp(p[b,\pi])$, $v(\pi,q')$ is a $(n,\gamma,\eta,\lambda(b,\pi,q'),q')$-approximation, thus: 
		\begin{align*}
			\hspace*{-1cm}
			x(\MEMDP,q,\sigma,b,v) & = \sum_{\pi \in \msf{FirstD}} \sum_{q' \in \Supp(p[b,\pi])} \sum_{e \in E} \Prb_q^\sigma[\MEMDP[e],\pi] \cdot b(e) \cdot  \delta_{e}(\pi)(q') \cdot \underbrace{|v(\pi,q') - \msf{val}^{\msf{pr}}_{q'}(\MEMDP,\lambda(b,\pi,q'),W)|}_{\leq \eta +  \msf{val}^{\msf{pr}}_{q}(\MEMDP,\lambda(b,\pi,q'),\msf{N}(\MEMDP,\lambda(b,\pi,q'),n,\gamma))} \\
			& \leq \eta + \sum_{\pi \in \msf{FirstD}} \sum_{q' \in \Supp(p[b,\pi])} \sum_{e \in E} \Prb_q^\sigma[\MEMDP[e],\pi] \cdot b(e) \cdot  \delta_{e}(\pi)(q') \cdot \msf{val}^{\msf{pr}}_{q}(\MEMDP,\lambda(b,\pi,q'),\msf{N}(\MEMDP,\lambda(b,\pi,q'),n,\gamma))
		\end{align*}
		
		Let $\varepsilon > 0$. We consider a strategy $\sigma_a \in \msf{Strat}(Q,A)$ such that $\sigma_a$ and $\sigma$ coincide on prefixes of all runs in $\msf{FirstD}$; and for all $\pi \in \msf{FirstD}$, the strategy $\sigma_a$ is such that, for all $q' \in \Supp(p[b,\pi])$, $\sigma_{s}^{\pi}$ ensures that:
		\begin{equation*}
			%\sum_{e \in E} \lambda(b,\pi,q')(e) \cdot \Prb_{q'}^{\sigma_a^{\pi}}[\MEMDP[e],\msf{N}(\MEMDP,\lambda(b,\pi,q'),n,\gamma)] = 
			\msf{val}^{\msf{pr}}_{q}(\MEMDP,\lambda(b,\pi,q'),\sigma^{\pi}_a,\msf{N}(\MEMDP,\lambda(b,\pi,q'),n,\gamma)) \geq \msf{val}^{\msf{pr}}_{q}(\MEMDP,\lambda(b,\pi,q'),\msf{N}(\MEMDP,\lambda(b,\pi,q'),n,\gamma)) - \varepsilon
		\end{equation*}
		That is:
		\begin{equation*}
			\sum_{e \in E} \lambda(b,\pi,q')(e) \cdot \Prb_{q'}^{\sigma_a^{\pi}}[\MEMDP[e],\msf{N}(\MEMDP,\lambda(b,\pi,q'),n,\gamma)]\geq \msf{val}^{\msf{pr}}_{q}(\MEMDP,\lambda(b,\pi,q'),\msf{N}(\MEMDP,\lambda(b,\pi,q'),n,\gamma)) - \varepsilon
		\end{equation*}
		By definition of $\lambda(b,\pi,q')$, this implies:
		\begin{align*}
			\sum_{e \in E} b(e) \cdot \delta_e(\pi,q') \cdot & \Prb_{q'}^{\sigma_a^{\pi}}[\MEMDP[e],\msf{N}(\MEMDP,\lambda(b,\pi,q'),n,\gamma)]+ \varepsilon \geq \\
			& \sum_{e \in E} b(e) \cdot \delta_e(\pi,q') \cdot \msf{val}^{\msf{pr}}_{q}(\MEMDP,\lambda(b,\pi,q'),\msf{N}(\MEMDP,\lambda(b,\pi,q'),n,\gamma))
		\end{align*}
		Furthermore, by definition of $\msf{FirstD}$, for all $\pi \in \msf{FirstD}$, the probability $\Prb_q^\sigma[\MEMDP[e],\pi] \in [0,1]$ does not depend on the environment $e \in E$. Therefore, we have:
		\begin{align*}
			\hspace*{-1cm}
			x(\MEMDP,q,\sigma,b,v) & \leq \eta + \sum_{\pi \in \msf{FirstD}} \sum_{q' \in \Supp(p[b,\pi])} \sum_{e \in E} \Prb_q^\sigma[\MEMDP[e],\pi] \cdot b(e) \cdot  \delta_{e}(\pi)(q') \cdot \msf{val}^{\msf{pr}}_{q}(\MEMDP,\lambda(b,\pi,q'),\msf{N}(\MEMDP,\lambda(b,\pi,q'),n,\gamma)) \\
			& \leq \eta + \sum_{\pi \in \msf{FirstD}} \sum_{q' \in \Supp(p[b,\pi])} \sum_{e \in E} \Prb_q^\sigma[\MEMDP[e],\pi] \cdot b(e) \cdot  \delta_{e}(\pi)(q') \cdot \Prb_{q'}^{\sigma_a^{\pi}}[\MEMDP[e],\msf{N}(\MEMDP,\lambda(b,\pi,q'),n,\gamma)]+ \varepsilon \\
			& = \eta + \sum_{\pi \in \msf{FirstD}} \sum_{q' \in \Supp(p[b,\pi])} \sum_{e \in E} \Prb_q^{\sigma_a}[\MEMDP[e],\pi] \cdot b(e) \cdot  \delta_{e}(\pi)(q') \cdot \Prb_{q'}^{\sigma_a^{\pi}}[\MEMDP[e],\msf{N}(\MEMDP,\lambda(b,\pi,q'),n,\gamma)]+ \varepsilon \\
			& = \eta + \sum_{e \in E} b(e) \cdot \Prb_{q}^{\sigma_a}[\MEMDP[e],\msf{N}(\MEMDP,b,n+1,\gamma)] + \varepsilon \\
			& \leq \eta + \msf{val}^{\msf{pr}}_{q}(\MEMDP,b,\msf{N}(\MEMDP,b,n+1,\gamma)) + \varepsilon
		\end{align*}
		The second equality comes from the fact that $\min_{e \in \Supp(b)} b(e) > \gamma$, and, by definition of $\msf{FirstD}$, for all $\pi \in \msf{FirstD}$ and $q' \in \Supp(p[b,\pi])$, we have $\msf{lk}_b(\pi \cdot q') = \lambda(b,\pi,q')$. 
		
		As the inequality holds for all strategies $\sigma \in \msf{Strat}(Q,A)$ and $\varepsilon > 0$, we deduce that:
		\begin{equation*}
			\sup_{\sigma \in \msf{Strat}(Q,A)} x(\MEMDP,q,\sigma,b,v) \leq \eta + \msf{val}^{\msf{pr}}_{q}(\MEMDP,b,\msf{N}(\MEMDP,b,n+1,\gamma))
		\end{equation*}
		The lemma then follows from Lemma~\ref{lem:prop-MDP-from-MEMDP}.
	\end{proof}
	
	\subsubsection{Correction of Algorithm~\ref{algo:MEMDP_parity}}
	\label{appen:correction-algorithm} 
	We consider a slightly modified version of Algorithm~\ref{algo:MEMDP_parity} where, in Line 8, the value $v_{q,a,q'}$ that we consider is truncated, with a $(\frac{\varepsilon}{3|E|} \cdot \frac{1}{m(\frac{\varepsilon}{3|E|})})$-approximation, so that it can be written on $\lceil |\log(\frac{\varepsilon}{3|E|} \cdot \frac{1}{m(\frac{\varepsilon}{3|E|})})| \rceil$ bits. 
	
	Let us establish the correction of this slightly modified version of  Algorithm~\ref{algo:MEMDP_parity}, stated below. 
	\begin{theorem}
		\label{thm:correection-algo}
		Consider an MEMDP $\MEMDP$, a parity objective $W$, and an environment belief $b \in \Dist(E)$. For all $\varepsilon > 0$, Algorithm~\ref{algo:MEMDP_parity} (with a truncation on Line 8) executed on $n := m(\MEMDP,\frac{\varepsilon}{3|E|})$ and $\gamma := \varepsilon$ outputs $f\colon Q \to [0,1]$ such that, for all $q \in Q$, we have $|\msf{val}^{\msf{pr}}_{q}(\MEMDP,b,W) - f(q)| \leq \varepsilon$.
	\end{theorem}
	\begin{proof}
		We prove by induction on $1 \leq |\Supp(b)| \leq E$ that the output $f\colon Q \to [0,1]$ of Algorithm~\ref{algo:MEMDP_parity} is such that, for all $q \in Q$, we have $|\msf{val}^{\msf{pr}}_{q}(\MEMDP,b,W) - f(q)| \leq |\Supp(b)| \cdot \frac{\varepsilon}{|E|}$.
		
		This clearly holds when $|\Supp(b)| = 1$, since Algorithm~\ref{algo:MEMDP_parity} does not make any approximation in that case. Assume now that the result holds for all $1 \leq i \leq j$ for some $j \leq |E|-1$. Let us show by induction on $t \in \N$ that for all beliefs $b \in \Dist(E)$ such that $|\Supp(b)| = j+1$, the output $f\colon Q \to [0,1]$ of Algorithm~\ref{algo:MEMDP_parity} executed on $b$, with $n := t$ and $\gamma := \varepsilon$ is such that, for all $q \in Q$: $|\msf{val}^{\msf{pr}}_{q}(\MEMDP,b,W) - f(q)| \leq |\Supp(b)| \cdot \frac{\varepsilon}{|E|} - \frac{2\varepsilon}{3|E|} + t \cdot \frac{\varepsilon}{3|E|} \cdot \frac{1}{m(\frac{\varepsilon}{3|E|})} + \msf{val}^{\msf{pr}}_{q}(\MEMDP,b,\msf{N}(\MEMDP,b,t,\frac{\varepsilon}{3|E|}))$.
		
		Let us first consider the case where $t = 0$ and let $b \in \Dist(E)$ such that $|\Supp(b)| = j+1$. Let $q \in Q$. If $\min_{e \in \Supp(b)} b(e) > \frac{\varepsilon}{3|E|}$, then %$\msf{N}(\MEMDP,b,t,\frac{\varepsilon}{3|E|}) = (Q \cdot A)^\omega$ and 
		$\msf{val}^{\msf{pr}}_{q}(\MEMDP,b,\msf{N}(\MEMDP,b,t,\frac{\varepsilon}{3|E|})) = 1$, thus the inequality holds. Otherwise, for $b' := \msf{Truncate}(b)$, we have $\msf{Diff}(b,b') \leq \frac{\varepsilon}{3|E|}$, thus by Lemma~\ref{lem:small_belief_change_ok}: $|\msf{val}^{\msf{pr}}_{q}(\MEMDP,b,W) - \msf{val}^{\msf{pr}}_{q}(\MEMDP,b',W)| \leq \frac{\varepsilon}{3|E|}$. Furthermore, $|\Supp(b')| < |\Supp(b)|$, thus by our induction hypothesis, we have the value $x_q \in [0,1]$ computed by the algorithm called recursively on $b'$ with $n := m(\MEMDP,\frac{\varepsilon}{3|E|})$ and $\gamma := \varepsilon$ for the state $q$ (Line 4) such that $|\msf{val}^{\msf{pr}}_{q}(\MEMDP,b',W) - x_q| \leq |\Supp(b')| \cdot \frac{\varepsilon}{|E|}$. Therefore, we have: $|\msf{val}^{\msf{pr}}_{q}(\MEMDP,b,W) - x_q| \leq |\Supp(b')| \cdot \frac{\varepsilon}{|E|} + \frac{\varepsilon}{3|E|} \leq |\Supp(b)| \cdot \frac{\varepsilon}{|E|} - \frac{2\varepsilon}{3|E|}$.
		
		Assume now that the induction hypothesis holds for some $t \in \N$ and let $b \in \Dist(E)$ such that $|\Supp(b)| = j+1$. Let $q \in Q$.  We consider the case where the algorithm is executed on $b$ with $n := t+1$ and $\gamma = \varepsilon$. If $\min_{e \in \Supp(b)} b(e) \leq \frac{\varepsilon}{3|E|}$, then this is as above: the value $x_q \in [0,1]$ computed by the algorithm called recursively on $b' := \msf{Truncate}(b)$ with $n := m(\MEMDP,\frac{\varepsilon}{3|E|})$ and $\gamma := \varepsilon$ for the state $q$ (Line 4) is such that $|\msf{val}^{\msf{pr}}_{q}(\MEMDP,b,W) - x_q| \leq |\Supp(b)| \cdot \frac{\varepsilon}{|E|} - \frac{2\varepsilon}{3|E|}$.
		
		Assume now that $\min_{e \in \Supp(b)} b(e) > \frac{\varepsilon}{3|E|}$, and thus the for loop on Line 5 is entered. For all $(q,a,q') \in S(\MEMDP,b)$, a value $v_{q,a,q'} \in [0,1]$ is computed by calling Algorithm~\ref{algo:MEMDP_parity} recursively (on Line 8). Note that the value 
		$v_{q,a,q'} \in [0,1]$ results from a $(\frac{\varepsilon}{3|E|} \cdot \frac{1}{m(\frac{\varepsilon}{3|E|})})$-approximation after calling Algorithm~\ref{algo:MEMDP_parity}. Considering $b' := \lambda[b,q,a,q'] \in \Dist(E)$, there are two cases:
		\begin{itemize}
			\item If $\Supp(b') \subsetneq \Supp(b)$, then $v_{q,a,q'}$ results from the call of the algorithm executed on $b'$ with $n := m(\MEMDP,\frac{\varepsilon}{3|E|})$ and $\gamma := \varepsilon$. By our induction hypothesis (on the size of the support), we have $|\msf{val}^{\msf{pr}}_{q'}(\MEMDP,b',W) - v_{q,a,q'}| \leq |\Supp(b')| \cdot \frac{\varepsilon}{|E|} + \frac{\varepsilon}{3|E|} \cdot \frac{1}{m(\frac{\varepsilon}{3|E|})} \leq |\Supp(b)| \cdot \frac{\varepsilon}{|E|} - \frac{2\varepsilon}{3|E|} + \frac{\varepsilon}{3|E|} \cdot \frac{1}{m(\frac{\varepsilon}{3|E|})}$.
			\item Otherwise, $v_{q,a,q'}$ results from the call of the algorithm executed on $b'$ with $n := t$ and $\gamma := \varepsilon$. By our induction hypothesis (on $n$), we have $|\msf{val}^{\msf{pr}}_{q'}(\MEMDP,b',W) - v_{q,a,q'}| \leq |\Supp(b)| \cdot \frac{\varepsilon}{|E|} - \frac{2\varepsilon}{3|E|} + t \cdot \frac{\varepsilon}{3|E|} \cdot \frac{1}{m(\frac{\varepsilon}{3|E|})} + \msf{val}^{\msf{pr}}_{q'}(\MEMDP,b',\msf{N}(\MEMDP,b',t,\frac{\varepsilon}{3|E|})) + \frac{\varepsilon}{3|E|} \cdot \frac{1}{m(\frac{\varepsilon}{3|E|})} = |\Supp(b)| \cdot \frac{\varepsilon}{|E|} - \frac{2\varepsilon}{3|E|} + (t+1) \cdot \frac{\varepsilon}{3|E|} \cdot \frac{1}{m(\frac{\varepsilon}{3|E|})} + \msf{val}^{\msf{pr}}_{q'}(\MEMDP,b',\msf{N}(\MEMDP,b',t,\frac{\varepsilon}{3|E|}))$.
		\end{itemize}   
		This holds for all $(q,a,q') \in S(\MEMDP,b)$. Therefore, by Lemma~\ref{lem:induction-step} (for $\gamma = \frac{\varepsilon}{3|E|}$ and $\eta = |\Supp(b)| \cdot \frac{\varepsilon}{|E|} - \frac{2\varepsilon}{3|E|} + (t+1) \cdot \frac{\varepsilon}{3|E|} \cdot \frac{1}{m(\frac{\varepsilon}{3|E|})}$), the MDP $\MDP$ and parity objective $W'$ computed in Line 9 by the algorithm executed on $b$ with $n := t+1$ and $\gamma = \varepsilon$ is such that $|\msf{val}^{\msf{pr}}_{q}(\MEMDP,b,W) - \msf{val}_{q}(\MDP,W')| \leq |\Supp(b)| \cdot \frac{\varepsilon}{|E|} -  \frac{2\varepsilon}{3|E|} + (t+1) \cdot \frac{\varepsilon}{3|E|} \cdot \frac{1}{m(\frac{\varepsilon}{3|E|})} + \msf{val}^{\msf{pr}}_{q}(\MEMDP,b,\msf{N}(\MEMDP,b,t+1,\frac{\varepsilon}{3|E|}))$. The induction hypothesis for $n := t+1$ follows, since the algorithm $\msf{MDP}$-$\msf{Parity}$ computes exactly the value $\msf{val}_{q}(\MDP,W')$.%does not make any approximation.
		
		In fact, the induction hypothesis (on $n$) holds for all $t \in \N$. In particular, the output $f\colon Q \to [0,1]$ of Algorithm~\ref{algo:MEMDP_parity} executed on $b$, with $n := m(\MEMDP,\frac{\varepsilon}{3|E|})$ and $\gamma := \varepsilon$ is such that, for all $q \in Q$: $|\msf{val}^{\msf{pr}}_{q}(\MEMDP,b,W) - f(q)| \leq |\Supp(b)| \cdot \frac{\varepsilon}{|E|} - \frac{\varepsilon}{3|E|} + \msf{val}^{\msf{pr}}_{q}(\MEMDP,b,\msf{N}(\MEMDP,b,m(\MEMDP,\frac{\varepsilon}{3|E|}),\frac{\varepsilon}{3|E|}))$. In addition, Theorem~\ref{thm:main_enough_actions_small_belief} gives that, for all $q \in Q$: $\msf{val}^{\msf{pr}}_{q}(\MEMDP,b,\msf{N}(\MEMDP,b,m(\MEMDP,\frac{\varepsilon}{3|E|}),\frac{\varepsilon}{3|E|})) \leq \frac{\varepsilon}{3|E|}$. Therefore, for all $q \in Q$, we have that $|\msf{val}^{\msf{pr}}_{q}(\MEMDP,b,W) - f(q)| \leq |\Supp(b)| \cdot \frac{\varepsilon}{|E|}$. 
		
		Thus, the induction hypothesis (on the support of beliefs) holds for belief support size equal to $j+1$. In fact, it holds for all $1 \leq j \leq |E|$. Therefore, for all beliefs $b$, for all $\varepsilon > 0$, Algorithm~\ref{algo:MEMDP_parity} executed on $n := m(\MEMDP,\frac{\varepsilon}{3|E|})$ and $\gamma := \varepsilon$ outputs a function $f\colon Q \to [0,1]$ such that, for all $q \in Q$, we have $|\msf{val}^{\msf{pr}}_{q}(\MEMDP,b,W) - f(q)| \leq |\Supp(b)| \cdot \frac{\varepsilon}{|E|} \leq \varepsilon$.
	\end{proof}                                                                                                                      
	\subsubsection{Complexity of Algorithm~\ref{algo:MEMDP_parity}}
	\label{appen:complexity-algorithm} 
	Let us first state a lemma about the number of bits to represent an updated belief.
	\begin{lemma}
		\label{lem:updated-belief-cenominator}
		Consider an MEMDP $\MEMDP$ and a belief $b \in \Dist(E)$. Assume that there is $N \in \N$ such that, for all $e \in E$, $b(e) = \frac{x_e}{N}$ for some $0 \leq x_e \leq N$. Assume also that there is $T \in \N$ such that, for all $e \in E$, and $(q,a,q') \in Q \times A \times Q$, we have $\delta_e(q,a)(q') = \frac{y_{q,a,q',e}}{T}$ for some $0 \leq y_{q,a,q',e} \leq T$. Then, for all $(q,a,q') \in Q \times A \times Q$ such that $p[b,q,a,q'] > 0$, there is $N' \in \N$ such that $\log(N') \leq \log(N) + \log(T) + \log(|E|)$ and, for all $e \in E$, there is $0 \leq z \leq N'$ such that $\lambda[b,q,a,q'](e) = \frac{z}{N'}$.
	\end{lemma}
	\begin{proof}
		Let $(q,a,q') \in Q \times A \times Q$ such that $p[b,q,a,q'] > 0$ and let $e \in E$. We have:
		\begin{align*}
			\lambda[b,q,a,q'](e) = \frac{b(e) \cdot \delta_e(q,a)(q')}{\sum_{e' \in E} b(e') \cdot \delta_{e'}(q,a)(q')} = \frac{\frac{x_e}{N} \cdot \frac{y_{q,a,q',e}}{T}}{\sum_{e' \in E} \frac{x_{e'}}{N} \cdot \frac{y_{q,a,q',e'}}{T}} = \frac{x_e \cdot y_{q,a,q',e}}{\sum_{e' \in E} x_{e'} \cdot y_{q,a,q',e'}}
		\end{align*}
		The denominator $\sum_{e' \in E} x_{e'} \cdot y_{q,a,q',e'}$ does not depend on the environment $e$ considered. Furthermore:
		\begin{equation*}
			\log\left(\sum_{e' \in E} x_{e'} \cdot y_{q,a,q',e'}\right) \leq \log\left(\sum_{e' \in E} N \cdot T\right) = \log(|E| \cdot N \cdot T) = \log(|E|) + \log(N) + \log(T)
		\end{equation*}
	\end{proof}
	
	We can now state and prove the complexity of Algorithm~\ref{algo:MEMDP_parity}. Note that we have established the correction of this algorithm (in Theorem~\ref{thm:correection-algo}) on a slightly modified version, where, in Line 8, the value $v_{q,a,q'}$ that we consider is truncated, with a $(\frac{\varepsilon}{3|E|} \cdot \frac{1}{m(\frac{\varepsilon}{3|E|})})$-approximation, so that it can be written on $\lceil |\log(\frac{\varepsilon}{3|E|} \cdot \frac{1}{m(\frac{\varepsilon}{3|E|})})| \rceil$ bits. We establish the complexity of Algorithm~\ref{algo:MEMDP_parity} with the same modification of Line 8.
	\begin{theorem}
		\label{thm:complexity-algo}
		Consider an MEMDP $\MEMDP$, a parity objective $W$, and an environment belief $b \in \Dist(E)$, and $\varepsilon > 0$. Let $n := |Q|$, $k := |E|$, $r := |A|$, $x := \lceil \frac{1}{\varepsilon} \rceil$, and $l \in \N$ (resp. $h \in \N$) denote the maximum number of bits to write probabilities in $\MEMDP$ (resp. the probabilities in $b$). Then, Algorithm~\ref{algo:MEMDP_parity} (with a truncation on Line 8) with $n := m(\MEMDP,\frac{\varepsilon}{3|E|})$ and $\gamma := \varepsilon$ executes in space $O(poly(n,r,k,h,x,2^l))$.
	\end{theorem}
	\begin{proof}
		As argued in the main part of the paper, each recursive call of the algorithm is done either on a belief with smaller support or with a smaller bound on the number of distinguishing state-actions pairs to be visited before truncation. Thus, the recursion depth of the algorithm is bounded by $k \cdot m(\MEMDP,\frac{\varepsilon}{3k})$, with $m(\MEMDP,\frac{\varepsilon}{3k}) = O(poly(k,2^l,x))$, by Theorem~\ref{thm:main_enough_actions_small_belief}. 
		
		Furthermore, the prior belief $b$ can be represented with all probabilities on the same denominator (as in Lemma~\ref{lem:updated-belief-cenominator}), which can be represented with at most $k \cdot h$ bits. The same can be done with all the probabilities in the MEMDP, with a denominator which can be represented with at most $l' := k \cdot n^2 \cdot r \cdot l$ bits (since we need to handle all quadruplets $(q,a,q',e) \in Q \times A \times Q \times E$). Then, the beliefs handle in each recursive calls are represented with the same denominator for all environments. When the belief is truncated, the denominator stays the same (cf. proof of Lemma~\ref{lem:small_belief_change_ok}). Moreover, the belief is updated at most $k \cdot m(\MEMDP,\frac{\varepsilon}{3k})$ times which, by Lemma~\ref{lem:updated-belief-cenominator}, only increases the number of bits to represent the denominator by $(l' + \log(k)) \cdot k \cdot m(\MEMDP,\frac{\varepsilon}{3k})$. Thus, the number of bits the represent belief used in all recursive calls is in $O(poly(n,r,k,h,x,2^l))$.
		
		Furthermore, the values in Line 8 are truncated, and thus can be represented on $
		O(poly(k,x,2^l))$ bits. Therefore, the MDPs built in Line 9, in all (recursive) calls to Algorithm~\ref{algo:MEMDP_parity}, have $n+2$ states, and probabilities that can be represented on $O(poly(n,r,k,h,x,2^l))$ bits (recall Definition~\ref{def:mdp-from-memdp}). Thus, the calls to $\msf{MDP}$-$\msf{Parity}$ are all executed in time polynomial in $O(poly(n,r,k,h,x,2^l))$ (\cite{DBLP:books/daglib/0020348}). Overall, we can conclude that Algorithm~\ref{algo:MEMDP_parity} executes in space $O(poly(n,r,k,h,x,2^l))$.
	\end{proof}
	
	\section{Complements on Section~\ref{sec:relating-the-values}}
	\label{appen:complements-section-values}
	
	\subsection{Proof of Lemma~\ref{lem:vonNeuman}}
	\label{proof:lem_vonBeuman}
	Before we proceed to the proof of the lemma, let us formally define the $\msf{uni}$- and $\msf{pr}$-values of mixed strategies, similarly to how they are defined for strategies $\sigma \in \msf{Strat}(Q,A)$.
	\begin{definition}
		Consider now an MEMDP $\MEMDP = (Q,A,E,(\delta_e)_{e \in E})$. For all states $q \in Q$, beliefs $b \in \Dist(E)$, and mixed strategies $\tau \in \Dist(\msf{Strat}(Q,A))$, we let:
		\begin{align*}
			\msf{val}^{\msf{pr}}_q(\MEMDP,b,\tau,W) & := \sum_{e \in E} b(e) \cdot \Prb_q^\tau[\MEMDP[e],W] = \sum_{e \in E} \sum_{\tau \in \msf{Strat}(Q,A)} b(e) \cdot \tau(\sigma) \cdot \Prb_q^\sigma[\MEMDP[e],W]\\
			\msf{val}^{\msf{uni}}_q(\MEMDP,\tau,W) & := \min_{e \in E} \Prb_q^\tau[\MEMDP[e],W] = \min_{e \in E} \sum_{\sigma \in \msf{Strat}(Q,A)} \tau(\sigma) \cdot \Prb_q^\sigma[\MEMDP[e],W]
		\end{align*}
	\end{definition}
	
	Let us now proceed to the proof of Lemma~\ref{lem:vonNeuman}.
	\begin{proof}
		Consider any two-non empty sets $A,B$ and $f\colon A \times B \to [0,1]$. Let $g\colon \Dist(A) \times \Dist(B) \to [0,1]$ (with probability distributions with countable support) be such that, for all $(\tau_A,\tau_B) \in \Dist(A) \times \Dist(B)$, we have $g(\tau_A,\tau_B) := \sum_{a \in A} \sum_{b \in B} \tau_A(a) \cdot \tau_B(b) \cdot f(a,b)$. Then, if $A$ or $B$ is finite: 
		\begin{equation*}
			\sup_{\tau_a \in \Dist(A)} \inf_{\tau_b \in \Dist(B)} g(\tau_a,\tau_b) = \inf_{\tau_b \in \Dist(B)} \sup_{\tau_a \in \Dist(A)} g(\tau_a,\tau_b)
		\end{equation*}
		This result corresponds to von Neuman's minimax theorem \cite{von1947theory} when $A$ and $B$ are finite, the case where one of the set is infinite (possibly uncountably) is a standard generalization, see e.g. \cite{sion1958general}. Applied to our case, we obtain:
		\begin{equation*}
			\sup_{\tau \in \Dist(\msf{Strat}(Q,A))} \inf_{b \in \Dist(E)} \msf{val}^{\msf{pr}}_q(\MEMDP,b,\tau,W) = \inf_{b \in \Dist(E)} \sup_{\tau \in \Dist(\msf{Strat}(Q,A))} \msf{val}^{\msf{pr}}_q(\MEMDP,b,\tau,W)
		\end{equation*}
		
		Consider any $\tau \in \Dist(\msf{Strat}(Q,A))$. Let $e \in E$ be such that $\msf{val}^{\msf{pr}}_q(\MEMDP,e,\tau,W) = \min_{e' \in E} \msf{val}^{\msf{pr}}_q(\MEMDP,e',\tau,W)$. Then, for all $b \in \Dist(E)$, we have:
		\begin{equation*}
			\msf{val}^{\msf{pr}}_q(\MEMDP,b,\tau,W) = \sum_{e' \in E} b(e') \cdot \msf{val}^{\msf{pr}}_q(\MEMDP,e',\tau,W) \geq \sum_{e' \in E} b(e') \cdot \msf{val}^{\msf{pr}}_q(\MEMDP,e,\tau,W) = \msf{val}^{\msf{pr}}_q(\MEMDP,e,\tau,W)
		\end{equation*}
		Thus: $\inf_{b \in \Dist(E)} \msf{val}^{\msf{pr}}_q(\MEMDP,b,\tau,W) = \min_{e' \in E} \msf{val}^{\msf{pr}}_q(\MEMDP,e',\tau,W) = \msf{val}^{\msf{uni}}_q(\MEMDP,\tau,W)$.
		
		Similarly, consider any $b \in \Dist(E)$ and $\varepsilon > 0$. Let $\sigma \in \msf{Strat}(Q,A)$ such that $\msf{val}^{\msf{pr}}_q(\MEMDP,b,\sigma,W) \geq \sup_{\sigma' \in \msf{Strat}(Q,A)} \msf{val}^{\msf{pr}}_q(\MEMDP,b,\sigma',W) - \varepsilon$. Then, for all $\tau \in \Dist(\msf{Strat}(Q,A))$, we have:
		\begin{align*}
			\msf{val}^{\msf{pr}}_q(\MEMDP,b,\tau,W) & = \sum_{\sigma' \in \msf{Strat}(Q,A)} \tau(\sigma') \cdot \msf{val}^{\msf{pr}}_q(\MEMDP,b,\sigma',W) \leq \sum_{\sigma' \in \msf{Strat}(Q,A)} \tau(\sigma') \cdot \msf{val}^{\msf{pr}}_q(\MEMDP,b,\sigma,W) + \varepsilon \\
			& = \msf{val}^{\msf{pr}}_q(\MEMDP,b,\sigma,W) + \varepsilon \leq \msf{val}^{\msf{pr}}_q(\MEMDP,b,W) + \varepsilon
		\end{align*}
		Since this holds for all $\varepsilon > 0$, we obtain $\sup_{\tau \in \Dist(\msf{Strat}(Q,A))} \msf{val}^{\msf{pr}}_q(\MEMDP,b,\tau,W) = \msf{val}^{\msf{pr}}_q(\MEMDP,b,W)$.
		
		Overall, we have: 
		\begin{equation*}
			\sup_{\tau \in \Dist(\msf{Strat}(Q,A))} \msf{val}^{\msf{uni}}_q(\MEMDP,\tau,W) = \inf_{b \in \Dist(E)} \msf{val}^{\msf{pr}}_q(\MEMDP,b,W)
		\end{equation*}
	\end{proof}
	
	\subsection{Proof of Lemma~\ref{lem:approx_value}}
	\label{appen:mixed_strategies}
	
	Before we proceed to the proof of Lemma~\ref{lem:approx_value}, let us consider how it implies Theorem~\ref{thm:approx_value}.
	\begin{proof}[Proof of Theorem~\ref{thm:approx_value}]
		By Lemmas~\ref{lem:vonNeuman} and~\ref{lem:approx_value}, we have:
		\begin{equation*}
			\msf{val}_q^{\msf{uni}}(\MEMDP,W) = \sup_{\tau \in \Dist(\msf{Strat}(Q,A))} \msf{val}^{\msf{uni}}_q(\MEMDP,\tau,W) = \inf_{b \in \Dist(E)} \msf{val}^{\msf{pr}}_q(\MEMDP,b,W)
		\end{equation*}
		The theorem follows.
	\end{proof}
	
	Let us now proceed to the proof of Lemma~\ref{lem:approx_value}. Note that it is actually similar to the proof of \cite[Lemma 19]{DBLP:journals/corr/abs-2504-15960}.
	\begin{proof}
		Consider an MEMDP $\MEMDP$, a state $q \in Q$, and a mixed strategy $\tau \in \Dist(\msf{Strat}(Q,A))$. For all $\rho \in Q \cdot (A \cdot Q)^*$, $e \in E$ and $\sigma \in \msf{Strat}(Q,A)$, we let $p[\sigma,e,\rho] := \tau(\sigma) \cdot \Prb_q^\sigma[\MEMDP[e],\rho]$ and $p[\tau,e,\rho] := \sum_{\sigma \in \msf{Strat}(Q,A)} p[\sigma,e,\rho] = \Prb_q^\tau[\MEMDP[e],\rho]$. Consider now any $\rho \in Q \cdot (A \cdot Q)^*$. Let us show that, for all $e,e' \in E$ such that $p[\tau,e,\rho],p[\tau,e',\rho] \neq 0$, we have, for all $\sigma \in \msf{Strat}(Q,A)$:
		\begin{equation*} 
			\frac{p[\sigma,e,\rho]}{p[\tau,e,\rho]} = \frac{p[\sigma,e',\rho]}{p[\tau,e',\rho]}
		\end{equation*}
		Indeed, consider any $e \in E$ and $\sigma \in \msf{Strat}(Q,A)$. Let us write $\rho$ as $\rho = q_0 \cdot (a_0,q_1) \cdots (a_{n-1},q_{n})$, with $q_0 := q$. For all $0 \leq i \leq n$, we let $\rho_{\leq i} := q_0 \cdots (a_{i-1},q_{i})$. (In particular, $\rho_{\leq 0} = q_0$.) We have:
		\begin{equation*}
			\Prb_q^\sigma[\MEMDP[e],\rho] = \Pi_{i = 0}^{n-1} \sigma(\rho_{\leq i})(a_{i}) \cdot \delta_e(q_i,a_i)(q_{i+1})
		\end{equation*}
		Therefore:
		\begin{equation*} 
			\frac{p[\sigma,e,\rho]}{p[\tau,e,\rho]} = \frac{\tau(\sigma) \cdot \prod\limits_{i = 0}^{n-1} \sigma(\rho_{\leq i})(a_{i}) \cdot \delta_e(q_i,a_i)(q_{i+1})}{\sum\limits_{\sigma' \in \msf{Strat}(Q,A)} \tau(\sigma') \cdot \prod\limits_{i = 0}^{n-1} \sigma'(\rho_{\leq i})(a_{i}) \cdot \delta_e(q_i,a_i)(q_{i+1})} = \frac{\tau(\sigma) \cdot \prod\limits_{i = 0}^{n-1} \sigma(\rho_{\leq i})(a_{i})}{\sum\limits_{\sigma' \in \msf{Strat}(Q,A)} \tau(\sigma') \cdot \prod\limits_{i = 0}^{n-1} \sigma'(\rho_{\leq i})(a_{i})}
		\end{equation*}
		This does not depend on the environment $e$, thus $\frac{p[\sigma,e,\rho]}{p[\tau,e,\rho]} = \frac{p[\sigma,e',\rho]}{p[\tau,e',\rho]}$. 
		
		Now, let us define a strategy $\sigma \in \msf{Strat}(Q,A)$. For all $\rho \in Q \cdot (A \cdot Q)^*$, we let $\sigma(\rho) \in \Dist(A)$ be such that, for all $a \in A$:
		\begin{align*}
			\sigma(\rho)(a) := \begin{cases}
				\sum\limits_{\sigma' \in \msf{Strat}(Q,A)} \sigma'(\rho)(a) \cdot \frac{p[\sigma',e,\rho]}{p[\tau,e,\rho]} & \text{ for some }e \in E\text{ such that }p[\tau,e,\rho] \neq 0 \text{ if one exists}\\
				\text{is arbitrary} & \text{ if }p[\tau,e,\rho] = 0 \text{ for all }e \in E
			\end{cases}
		\end{align*}
		
		Consider an environment $e \in E$. Let us show by induction on $\rho \in Q \cdot (A \cdot Q)^*$ that $\Prb_q^\sigma[\MEMDP[e],\rho] = p[\tau,e,\rho]$. This clearly holds for all $\rho \in Q$. Assume now that this holds for some $\rho \in Q \cdot (A \cdot Q)^*$ and let $(a,q') \in Q$. If $p[\tau,e,\rho] = \Prb_q^\sigma[\MEMDP[e],\rho] = 0$, then $p[\tau,e,\rho \cdot (a,q')],\Prb_q^\sigma[\MEMDP[e],\rho \cdot (a,q')] = 0$. Thus, we assume that $p[\tau,e,\rho] = \Prb_q^\sigma[\MEMDP[e],\rho] > 0$. We have:
		\begin{align*}
			\Prb_q^\sigma[\MEMDP[e],\rho \cdot (a,q')] & = \Prb_q^\sigma[\MEMDP[e],\rho] \cdot \sigma(\rho)(a) \cdot \delta_e(\head{\rho},a)(q') \\
			& = p[\tau,e,\rho] \cdot \left(\sum\limits_{\sigma' \in \msf{Strat}(Q,A)} \sigma'(\rho)(a) \cdot \frac{p[\sigma',e,\rho]}{p[\tau,e,\rho]}\right) \cdot \delta_e(\head{\rho},a)(q') \\
			& = \sum\limits_{\sigma' \in \msf{Strat}(Q,A)} \sigma'(\rho)(a) \cdot p[\sigma',e,\rho]  \cdot \delta_e(\head{\rho},a)(q') \\
			& = \sum\limits_{\sigma' \in \msf{Strat}(Q,A)} p[\sigma',e,\rho \cdot (a,q')] \\
			& = p[\tau,e,\rho \cdot (a,q')]
		\end{align*}
		Therefore, the property holds for $\rho \cdot (a,q') \in Q \cdot (A \cdot Q)^*$ as well. In fact, it holds for all $\rho \in Q \cdot (A \cdot Q)^*$. It follows that, for all $\rho \in Q \cdot (A \cdot Q)^*$:
		\begin{equation*}
			\Prb_q^\sigma[\MEMDP[e],\cyl(\rho)] = \sum_{\sigma' \in \msf{Strat}(Q,A)} \tau(\sigma') \cdot \Prb_q^{\sigma'}[\MEMDP[e],\cyl(\rho)] = \Prb_q^\tau[\MEMDP[e],\cyl(\rho)]
		\end{equation*}
		By Carathéodory's extension theorem, a probability measure is entirely defined by the measure of cylinder sets. Since $\Prb_q^\sigma[\MEMDP[e],\cdot]\colon \msf{Borel}(Q)$ and $\sum_{\sigma' \in \msf{Strat}(Q,A)} \tau(\sigma') \cdot \Prb_q^{\sigma'}[\MEMDP[e],\cdot]\colon \msf{Borel}(Q)$ are probability measures, we obtain that, for all Borel objectives $W \in \msf{Borel}(Q)$:
		\begin{equation*}
			\Prb_q^\sigma[\MEMDP[e],W] = \Prb_q^\tau[\MEMDP[e],W]
		\end{equation*}
		This holds for all environment $e \in E$, thus:
		\begin{equation*}
			\msf{val}^{\msf{uni}}_q(\MEMDP,\tau,W) = \msf{val}^{\msf{uni}}_q(\MEMDP,\sigma,W) \leq \msf{val}^{\msf{uni}}_q(\MEMDP,W)
		\end{equation*}
		This holds for all $\sigma \in \Dist(\msf{Strat}(Q,A))$. In addition, we have $\msf{val}^{\msf{uni}}_q(\MEMDP,W) \leq \sup\limits_{\tau \in \Dist(\msf{Strat}(Q,A))} \msf{val}^{\msf{uni}}_q(\MEMDP,\tau,W)$. The lemma follows.
	\end{proof}

	\subsection{Proof of Theorem~\ref{thm:complexity_deciding_gap_problem_uni_val}}
	\label{proof:thm_complexity_deciding_gap_problem_uni_val}
	\begin{proof}
		As in the proof sketch, let $N := \lceil \log(\frac{1}{\varepsilon}) \rceil$, $k := |E|$, and $S := \{ b \in \Dist(E) \mid \forall e \in E,\; b(e) \in \{ \frac{x}{k \cdot 2^{N+1}} \mid 0 \leq x \leq k \cdot 2^{N+1}\} \}$. We indeed have $|S| = O(2^{k^2 \cdot N})$, we also have $\varepsilon \geq \frac{1}{2^N}$.
		
		Consider any $b \in \Dist(E)$. Let us define $b' \in S$ such that $\msf{Diff}(b,b') \leq \frac{\varepsilon}{2}$. Consider any $e \in E$. For all $e' \in E \setminus \{e\}$, we let $b'(e') := \frac{x_{e'}}{k \cdot 2^{N+1}}$ for the largest $0 \leq x \leq k \cdot 2^{N+1}$ such that $b'(e) \leq b(e)$. Therefore, we have $0 \leq b(e') - b'(e') \leq \frac{1}{k \cdot 2^{N+1}}$. Furthermore, we have $\sum_{e' \in E \setminus \{e\}} x_{e'} = k \cdot 2^{N+1} \cdot \left(\sum_{e' \in E \setminus \{e\}} b'(e')\right) \leq  k \cdot 2^{N+1} \left(\sum_{e' \in E \setminus \{e\}} b(e')\right) \leq k \cdot 2^{N+1}$. Thus, we can define $b'(e) := \frac{\sum_{e' \in E \setminus \{e\}} x_{e'}}{k \cdot 2^{N+1}}$ and have $b(e) \leq b'(e) \in [0,1]$. We have defined $b' \in \Dist(E)$. As noted in the proof of Lemma~\ref{lem:small_belief_change_ok}, we have:
		\begin{equation*}
			\msf{Diff}(b,b') = \sum_{\substack{\tilde{e} \in E \\ b(\tilde{e}) \geq b'(\tilde{e})}} (b(\tilde{e}) - b'(\tilde{e}))
		\end{equation*}
		Therefore:
		\begin{equation*}
			\msf{Diff}(b,b') = \sum_{\substack{\tilde{e} \in E \\ b(\tilde{e}) \geq b'(\tilde{e})}} (b(\tilde{e}) - b'(\tilde{e})) = \sum_{e' \in E \setminus \{e\}} (b(e') - b'(e')) \leq \sum_{e' \in E \setminus \{e\}} \frac{1}{k \cdot 2^{N+1}} \leq \frac{1}{2^{N+1}} \leq \frac{\varepsilon}{2}
		\end{equation*}
		
		Therefore, by Theorem~\ref{thm:approx_value} and Lemma~\ref{lem:small_belief_change_ok}:
		\begin{equation*}
			\msf{val}_q^{\msf{uni}}(\MEMDP,W) = \inf_{b \in \Dist(E)} \msf{val}_q^{\msf{pr}}(\MEMDP,b,W) \leq \min_{b \in S} \msf{val}_q^{\msf{pr}}(\MEMDP,b,W) \leq \inf_{b \in \Dist(E)} \msf{val}_q^{\msf{pr}}(\MEMDP,b,W) + \frac{\varepsilon}{2} = \msf{val}_q^{\msf{uni}}(\MEMDP,W) + \frac{\varepsilon}{2}
		\end{equation*}
		Thus:
		\begin{equation*}
			|\min_{b \in S} \msf{val}_q^{\msf{pr}}(\MEMDP,b,W) - \msf{val}_q^{\msf{uni}}(\MEMDP,W)| \leq \frac{\varepsilon}{2}
		\end{equation*}

		Furthermore, the number of bits to encode any $b \in S$ is polynomial in $N$ and $k$. Therefore, Algorithm~\ref{algo:MEMDP_parity} executed on any belief $b \in S$ with $\gamma := \varepsilon/2$ and $n := m(\MEMDP,\frac{\varepsilon}{4k})$ computes in exponential space (resp. polynomial space if the probabilities are given in unary) a value $v_b \in [0,1]$ such that $|\msf{val}_q^{\msf{pr}}(\MEMDP,b,W) - v_b| \leq \frac{\varepsilon}{2}$. Letting $b_m \in S$ be such that $\min_{b \in S} \msf{val}_q^{\msf{pr}}(\MEMDP,b,W) = \msf{val}_q^{\msf{pr}}(\MEMDP,b_m,W)$ and $b_s \in S$ be such that $\min_{b \in S} v_b = v_{b_s}$, we have:
		\begin{equation*}
			-\frac{\varepsilon}{2} \leq \msf{val}_q^{\msf{pr}}(\MEMDP,b_m,W) - v_{b_m} \leq \msf{val}_q^{\msf{pr}}(\MEMDP,b_m,W) - v_{b_s} \leq \msf{val}_q^{\msf{pr}}(\MEMDP,b_s,W) - v_{b_s} \leq \frac{\varepsilon}{2}
		\end{equation*}
		Therefore, we have:
		\begin{equation*}
			|\min_{b \in S} \msf{val}_q^{\msf{pr}}(\MEMDP,b,W) - \min_{b \in S} v_b| \leq \frac{\varepsilon}{2}
		\end{equation*}
		
		Hence:
		\begin{equation*}
			|\min_{b \in S} v_b - \msf{val}_q^{\msf{uni}}(\MEMDP,W)| \leq \varepsilon
		\end{equation*}
		
		The environment beliefs in $S$ can be enumerated in polynomial space, thus the value $\min_{b \in S} v_b$ can be computed in exponential space (resp. polynomial space if the probabilities are given in unary), the theorem follows.
	\end{proof}
	
	\subsection{Proof of Proposition~\ref{prop:binary_search}}
	\label{proof:lem_binary_search}
	Let us first introduce the objects that we will use throughout this subsection. We fix an MEMDP $\MEMDP$, with $E := \{ e_1,e_2\}$, and we consider some state $q \in Q$, and a parity objective $W \in \msf{Borel}(Q)$. 
	
	To simplify the notations, we let $S := \msf{Strat}(Q,A)$. Furthermore, since there are only two environments $e_1$ and $e_2$, a distribution $b \in \Dist(E)$ is entirely defined by the value it gives to environment $e_1$. Thus, for all $x \in [0,1]$, we let $b_x \in \Dist(E)$ be such that $b_x(e_1) := x$ and $b_x(e_2) := 1 - x$. We then let $g\colon S \times [0,1] \to [0,1]$ be such that, for all $\sigma \in S$ and $x \in [0,1]$: $g(\sigma,x) := \msf{val}_q^{\msf{pr}}(\MEMDP,b_x,\sigma,W) = x \cdot \Prb_q^\sigma[\MEMDP[e_1],W] + (1-x) \cdot \Prb_q^\sigma[\MEMDP[e_2],W] \in [0,1]$. We also let $f\colon [0,1] \to [0,1]$ be such that, for all $x \in [0,1]$, $f(x) := \msf{val}_q^{\msf{pr}}(\MEMDP,b_x,W) = \sup_{\sigma \in S} g(\sigma,x)$. 
	
	Now, let us show that the function $f$ satisfies the lemma below.
	\begin{lemma}
		\label{lem:sub-convex}
		For all $x < y < z \in [0,1]$, we have:
		\begin{equation*}
			f(y) \leq \min(f(x),f(z))
		\end{equation*} 
		
		Furthermore, consider any $x < z \in [0,1]$. Letting $y := \frac{z + x}{2}$, $\delta := z - x$, and $\eta := \max(|f(y) - f(x)|,|f(z) - f(y)|)$, then:
		\begin{equation*}
			f(y) - \frac{2\eta}{\delta} \leq \inf_{t \in [0,1]} f(t) \leq f(y)
		\end{equation*}
	\end{lemma} 
	\begin{proof}
		Consider first some $x < y < z \in [0,1]$. Let $\varepsilon > 0$ and $\sigma \in S$ such that $g(\sigma,y) \geq f(y) - \varepsilon$. Let $a_1 := g(\sigma,1)$ and $a_2 := g(\sigma,0)$. There are two cases:
		\begin{itemize}
			%\item If $a := a_1 = a_2 \in [0,1]$, then $f(y) - \varepsilon \leq g(\sigma,y) = a = g(\sigma,z) \leq f(z) \leq \min(f(x),f(z))$;
			\item If $a_1 \geq a_2 \in [0,1]$, then $g(\sigma,z) - g(\sigma,y) = (z a_1 + (1-z) a_2) - (y a_1 + (1-y) a_2) = (z-y) (a_1 - a_2) \geq 0$. Therefore: $f(y) - \varepsilon \leq g(\sigma,y) \leq g(\sigma,z) \leq f(z) \leq \min(f(x),f(z))$;
			\item If $a_1 < a_2 \in [0,1]$, then $g(\sigma,x) - g(\sigma,y) = (x a_1 + (1-x) a_2) - (y a_1 + (1-y) a_2) = (x-y) (a_1 - a_2) > 0$. Therefore: $f(y) - \varepsilon \leq g(\sigma,y) \leq g(\sigma,x) \leq f(x) \leq \min(f(x),f(z))$.
		\end{itemize}
		Since this holds for all $\varepsilon > 0$, we have $f(y) \leq \min(f(x),f(z))$.
		
		%By definition, we have:
		%\begin{equation*}
		%	g(\sigma,z) - g(\sigma,y) = z \cdot a_1 + (1-z) \cdot a_2 
		%\end{equation*}
		
		Consider now some $x < z \in [0,1]$, and let $y := \frac{z + x}{2}$, $\delta := z - x$, and $\eta := \max(|f(y) - f(x)|,|f(x) - f(z)|)$. Let $\varepsilon > 0$ and consider some $\sigma \in S$ such that $g(\sigma,y) \geq f(y) - \varepsilon$. Let $a_1 := g(\sigma,1)$ and $a_2 := g(\sigma,0)$. We have: $g(\sigma,y) = g(\sigma,z) + (z - y) (a_2 - a_1)$. Thus
		\begin{equation*}
			f(y) - \varepsilon \leq g(\sigma,y) = g(\sigma,z) + (z - y) (a_2 - a_1) \leq f(z) + \frac{\delta(a_2 - a_1)}{2}
		\end{equation*}
		Thus:
		\begin{equation*}
			a_1 - a_2 \leq \frac{2}{\delta} (f(z) - f(y) + \varepsilon)  \leq \frac{2}{\delta} (\eta + \varepsilon)
		\end{equation*}

		Symmetrically, we have: $g(\sigma,y) = g(\sigma,x) + (y - x) (a_1 - a_2)$. Thus
		\begin{equation*}
			f(y) - \varepsilon \leq g(\sigma,y) = g(\sigma,x) + (y - x) (a_1 - a_2) \leq f(x) + \frac{\delta(a_1 - a_2)}{2}
		\end{equation*}
		Thus:
		\begin{equation*}
			a_2 - a_1 \leq \frac{2}{\delta} (f(x) - f(y) + \varepsilon)  \leq \frac{2}{\delta} (\eta + \varepsilon)
		\end{equation*}
		
		As this holds for all $\varepsilon > 0$, we obtain:
		\begin{equation*}
			|a_1 - a_2| \leq \frac{2(\eta+\varepsilon)}{\delta}
		\end{equation*} 
		Now, let $a := \min(a_1,a_2)$. We have $f(y) - \varepsilon \leq g(\sigma,y) \leq a + \frac{2(\eta+\varepsilon)}{\delta}$. Furthermore, for all $t \in [0,1]$, we have $f(t) \geq g(\sigma,t) \geq a$. Overall, we obtain:
		\begin{equation*}
			f(y) - \varepsilon - \frac{2(\eta+\varepsilon)}{\delta} \leq a \leq \inf_{t \in [0,1]} f(t) \leq f(y)
		\end{equation*}
		As this holds for $\varepsilon > 0$, the lemma follows.
	\end{proof}
	
	\begin{algorithm}[t]
		\caption{$\textsf{Inf-f}(x,t,\varepsilon)$}
		\label{algo:binary-search}
		\textbf{Input}: $x < t \in [0,1]$, $\varepsilon > 0$
		\begin{algorithmic}[1]
			\If{$t - x \leq \frac{\varepsilon}{2}$}
			\Return $\textsf{Compute-f}(x,\frac{\varepsilon}{2})$
			\EndIf
			\State $y \gets \frac{2x+t}{3}$, $a_y \gets \textsf{Compute-f}(y,\frac{\varepsilon^2}{48})$
			\State $z \gets \frac{x+2t}{3}$, $a_z \gets \textsf{Compute-f}(z,\frac{\varepsilon^2}{48})$
			%\If{$a_y < a_z + \frac{\varepsilon^2}{24}$}
			%\Return $\textsf{Inf-f}(x,z,\varepsilon)$
			%\EndIf
			%\If{$a_y > a_z + \frac{\varepsilon^2}{24}$}
			%\Return $\textsf{Inf-f}(y,t,\varepsilon)$
			%\EndIf
			\State $u \gets \frac{y+z}{2}$, $a_u \gets \textsf{Compute-f}(u,\frac{\varepsilon^2}{48})$
			\If{$a_u > a_z + \frac{\varepsilon^2}{24}$ \text{ or }$a_y > a_u + \frac{\varepsilon^2}{24}$}
			\Return $\textsf{Inf-f}(y,t,\varepsilon)$
			\EndIf
			\If{$a_u > a_y + \frac{\varepsilon^2}{24}$ \text{ or }$a_z > a_u + \frac{\varepsilon^2}{24}$}
			\Return $\textsf{Inf-f}(x,z,\varepsilon)$
			\EndIf
			\State
			\Return $a_u$
		\end{algorithmic}
	\end{algorithm}
	
	We can now proceed to the proof of Proposition~\ref{prop:binary_search}.
	\begin{proof}
		%	Let us prove that, given as input $x < t \in [0,1]$ and $\varepsilon > 0$, 
		Consider Algorithm~\ref{algo:binary-search} called $\textsf{Inf-f}$. The algorithm $\textsf{Compute-f}$, called by Algorithm~\ref{algo:binary-search}, on an input $u \in [0,1]$ and some $\varepsilon > 0$, outputs a value $v \in [0,1]$ such that $|v - f(u)| \leq \varepsilon$. In practice, calling $\textsf{Compute-f}(u,\varepsilon)$ could simply amount to calling Algorithm~\ref{algo:MEMDP_parity} with $b := b_u$, $\gamma := \varepsilon$ and $n := m(\frac{\varepsilon}{3|E|})$ (see Theorem~\ref{thm:correection-algo}). In this proof, we view calls to $\textsf{Compute-f}$ as calls to an oracle approximating $f$ (and thus $\msf{val}_q^{\msf{pr}}$).
		
		Let $\varepsilon > 0$. We show by induction on $k \in \N$ that, for all $x < t \in [0,1]$ such that $t - x \leq \frac{\varepsilon}{2} \cdot (\frac{3}{2})^k$, 
		algorithm $\textsf{Inf-f}$ called on $(x,t,\varepsilon)$ outputs a value $v \in [0,1]$, with at most $k-1$ recursive calls (to $\textsf{Inf-f}$), such that $|v - \sup_{u \in [x,t]} f(u)| \leq \varepsilon$. 
		
		Consider first the case $k = 0$ and $x < t \in [0,1]$ such that $t - x \leq \frac{\varepsilon}{2}$. Then, $\textsf{Inf-f}(x,t,\varepsilon)$ outputs the value $v \in [0,1]$, with no recursive calls (to $\textsf{Inf-f}$), returned by $\textsf{Compute-f}(x,\frac{\varepsilon}{2})$, which satisfies $|v - f(x)| \leq \frac{\varepsilon}{2}$. Furthermore, for all $u \in [x,t]$, we have $\msf{Diff}(b_x,b_u) = u-x \leq \frac{\varepsilon}{2}$. Hence, by Lemma~\ref{lem:small_belief_change_ok}, we have $|f(u) - f(x)| \leq \frac{\varepsilon}{2}$. Thus, $|f(x) - \inf_{u \in [x,t]} f(u)| \leq \frac{\varepsilon}{2}$, thus $|v - \sup_{u \in [x,t]} f(u)| \leq \varepsilon$.
		
		Assume now that the inductive property is satisfied for some $k \in \N$ and consider some $x < t \in [0,1]$ such that $t - x \leq \frac{\varepsilon}{2} \cdot (\frac{3}{2})^{k+1}$. If $t - x \leq \frac{\varepsilon}{2}$, this is as above. Otherwise, we have: $y = \frac{2x+t}{3}$, $z = \frac{x+2t}{3}$, $u = \frac{y+z}{2}$, and $|a_y - f(y)| \leq \frac{\varepsilon^2}{48}$, $|a_z - f(z)| \leq \frac{\varepsilon^2}{48}$, $|a_u - f(u)| \leq \frac{\varepsilon^2}{48}$. There are three cases.
		\begin{itemize}
			\item If $a_u > a_z + \frac{\varepsilon^2}{24}$, then $f(u) > f(z)$. We have $y < u < z$, thus by Lemma~\ref{lem:sub-convex}, we have $f(u) \leq \max(f(y),f(z))$, thus $f(y) \geq f(u) > f(z)$. In addition, if $a_y > a_u + \frac{\varepsilon^2}{24}$, then $f(y) > f(u)$. Thus: $f(y) > \min(f(u),f(z))$.
			
			Then, for any $x \leq \alpha < y$, since $\alpha < y < u < z$, we have $f(y) \leq \max(f(\alpha),f(z))$ and $f(y) \leq \max(f(\alpha),f(u))$, thus $f(\alpha) \geq f(y) > \min(f(u),f(z))$. Therefore: $\inf_{\alpha \in [x,t]} f(\alpha) = \inf_{\alpha \in [y,t]} f(\alpha)$. Furthermore, we have: 
			\begin{equation*}
				t - y = t - \frac{2x+t}{3} = \frac{2}{3} (t-x) \leq \frac{\varepsilon}{2} \cdot (\frac{3}{2})^{k}
			\end{equation*}
			Thus, by our induction hypothesis, $\textsf{Inf-f}(y,t,\varepsilon)$ outputs a value $v \in [0,1]$, with at most $k-1$ recursive calls (to $\textsf{Inf-f}$) which satisfies $|v - \inf_{\alpha \in [x,t]} f(\alpha)| = |v - \inf_{\alpha \in [y,t]} f(\alpha)| \leq \varepsilon$. That value $v$ is outputted by $\textsf{Inf-f}(x,t,\varepsilon)$ with at most $k$ recursive calls.
			\item If $a_u > a_y + \frac{\varepsilon^2}{24}$, then $f(u) > f(y)$. We have $y < u < z$, thus by Lemma~\ref{lem:sub-convex}, we have $f(u) \leq \max(f(y),f(z))$, thus $f(z) \geq f(u) > f(y)$. In addition, if $a_z > a_u + \frac{\varepsilon^2}{24}$, then $f(z) > f(u)$. Thus: $f(z) > \min(f(u),f(y))$.
			
			Then, for any $z < \alpha \leq t$, since $y < u < z < \alpha$, we have $f(z) \leq \max(f(y),f(\alpha))$ and $f(z) \leq \max(f(u),f(\alpha))$, thus $f(\alpha) \geq f(z) > \min(f(y),f(u))$. Therefore: $\inf_{\alpha \in [x,t]} f(\alpha) = \inf_{\alpha \in [x,z]} f(\alpha)$. Furthermore, we have: 
			\begin{equation*}
				z - x = \frac{x+2t}{3} - x = \frac{2}{3} (t-x) \leq \frac{\varepsilon}{2} \cdot (\frac{3}{2})^{k}
			\end{equation*}
			Thus, by our induction hypothesis, $\textsf{Inf-f}(x,z,\varepsilon)$ outputs a value $v \in [0,1]$, with at most $k-1$ recursive calls (to $\textsf{Inf-f}$) which satisfies $|v - \inf_{\alpha \in [x,t]} f(\alpha)| = |v - \inf_{\alpha \in [x,z]} f(\alpha)| \leq \varepsilon$. That value $v$ is outputted by $\textsf{Inf-f}(x,t,\varepsilon)$ with at most $k$ recursive calls.
			\item If none of the above apply, then $|a_u - a_y| \leq \frac{\varepsilon^2}{24}$ and $|a_u - a_z| \leq \frac{\varepsilon^2}{24}$. Therefore: $|f(u) - f(y)| \leq |f(u) - a_u| + |f(y) - a_y| + |a_u - a_y| \leq \frac{\varepsilon^2}{48} + \frac{\varepsilon^2}{48} + \frac{\varepsilon^2}{24} = \frac{\varepsilon^2}{12}$, and similarly $|f(u) - f(z)| \leq \frac{\varepsilon^2}{12}$. Hence, by Lemma~\ref{lem:sub-convex}, we have:
			\begin{equation*}
				|f(u) - \inf_{\alpha \in [x,t]} f(t)| \leq \frac{2 \cdot \frac{\varepsilon^2}{12}}{z - y} = \frac{\varepsilon^2}{2(t-x)} \leq \frac{\varepsilon^2}{\varepsilon} = \varepsilon
			\end{equation*}
			This last two inequalities come from the fact that $z - y = \frac{t - x}{3}$ and $t - x \geq \frac{\varepsilon}{2}$. 
		\end{itemize}
		Overall, this induction hypothesis holds for all $k \in \N$. 
		
		Now, let $k := \lceil \frac{1}{\log(\frac{3}{2})} \cdot \log(\frac{2}{\varepsilon}) \rceil \in \N$. Note that $k = O(N)$ and $\frac{\varepsilon}{2} \cdot (\frac{3}{2})^k \geq 1$. Then, $\textsf{Inf-f}(0,1,\varepsilon)$ outputs a value $v \in [0,1]$, with at most $k-1$ recursive calls, such that $|v - \msf{val}_q^{\msf{uni}}(\MEMDP,W)| = |v - \sup_{u \in [0,1]} f(u)| \leq \varepsilon$ (by Theorem~\ref{thm:approx_value}, since for all $u \in [0,1]$, we have $f(u) = \msf{val}_q^{\msf{pr}}(\MEMDP,b_u,W)$). Furthermore, each one of these recursive calls involves at most three calls to $\textsf{Compute-f}$, with a precision either $\frac{\varepsilon}{2}$ or $\frac{\varepsilon^2}{48}$, which are both polynomial in $\varepsilon$%(and thus to Algorithm~\ref{algo:binary-search})
		. In addition, the values in $[0,1]$ on which $\textsf{Compute-f}$ is called are in the set $X := \{ \frac{x}{3^k} \mid 0 \leq x \leq 3^k \}$, or are the average of two values in that set $X$; thus the corresponding beliefs are described with a number of bits polynomial in $N$. The proposition follows.		
	\end{proof}
	
	\subsection{Proof of Corollary~\ref{coro:complexity_deciding_gap_problem_uni_val_hardness}}
	\label{proof:cor_complexity_deciding_gap_problem_uni_val_hardness}
	\begin{proof}
		It was shown in \cite[Theorem 26]{DBLP:journals/corr/RaskinS14} that the gap problem (P1) with $\msf{uni}$-values and reachability objectives\footnote{They can straightforwardly be encoded into parity objectives.} in two-environment MEMDPs with probabilities written in unary (and precision in binary), is $\msf{NP}$-hard. Now, consider the gap problem (P2) with $\msf{pr}$-values and parity objectives in two-environment MEMDPs with probabilities written in unary, with the belief and precision given in binary. If (P2) can solved in polynomial time, then we can compute $\gamma$-approximation of $\msf{pr}$-values (with $\gamma$ given in binary) with a binary search (it suffices to consider a precision for (P2) smaller than $\gamma$). However, by Proposition~\ref{prop:binary_search}, if we can compute $\gamma$-approximation of $\msf{pr}$-values in polynomial time, then problem (P1) can be solved in polynomial time. Therefore, if (P2) can be solved in polynomial time, so can (P1), and $\msf{P} = \msf{NP}$.
	\end{proof}
	
	\section{Complements on Section~\ref{sec:entropy}}
	\label{appen:entropy_POMDP}	
	
	\subsection{POMDPs induced by MEMDPs are Dirac-preserving}
	\begin{proposition}
		For all MEMDPs $\MEMDP$, the POMDP $\POMDP(\MEMDP)$ is Dirac-preserving. 
	\end{proposition}                                          
	\begin{proof}
		Consider any $(q,e) \in Q \times E$ and action $a \in A$. For all $(q',e') \in Q \times E$, we have that $\delta((q,e),a)((q',e')) > 0$ implies $e' = e$; and $O((q',e')) = q'$. Therefore, for all compatible observations $q' \in \msf{Comp}((q,e),a) \subseteq Q$, we have $\lambda((q,e),a,q')((q',e)) = 1$, and thus $H(\lambda((q,e),a,q')) = 0$.
	\end{proof}
	
	\subsection{Proof of Proposition~\ref{prop:entropy}}
	\label{proof:prop_entropy}
	Let us first quickly introduce a notation on POMDPs.
	\begin{definition}
		\label{def:notation_pomdp}
		Consider a Dirac-preserving POMDP $\POMDP = (Q,A,\delta,\Omega,O)$. For all $q \in Q$, for all $a \in A$, and for all $o \in \msf{Comp}(q,a)$, we let $\msf{sc}(q,a,o) \in Q$ be such that $\lambda(q,a,o)(\msf{sc}(q,a,o)) = 1$. We also let $\msf{Prec}(q,a,o) := \{ q' \in Q \mid q = \msf{sc}(q',a,o) \}$. 
	\end{definition}
	
	Let us now proceed to the proof of Proposition~\ref{prop:entropy}.
	\begin{proof}
		Let us first assume that $\POMDP$ has non-increasing entropy. Consider any state $q \in Q$ and action $a \in A$. We have:
		\begin{equation*}
			0 = H(q) \geq \sum_{o \in \msf{Comp}(q,a)} p(q,a,o) \cdot H(\lambda(q,a,o)) \geq 0
		\end{equation*}
		Then, for all $o \in \msf{Comp}(q,a)$, since $p(q,a,o) > 0$, we necessarily have $H(\lambda(q,a,o)) = 0$. Hence, $\POMDP$ is Dirac-preserving.
		
		Let us now assume that $\POMDP$ is Dirac-preserving. 
		
		For all $(q,a) \in Q \times A$, $o \in \msf{Comp}(q,a)$, and $q' \in O^{-1}(o)$, we have:
		\begin{align*}
			\delta(q,a)(q') = \begin{cases}
				p(q,a,o) > 0 & \text{ if }q' = \msf{sc}(q,a,o) \\
				0 & \text{ otherwise }
			\end{cases}
		\end{align*}
		Therefore, we have:
		\begin{align*}
			\sum_{o \in \msf{Comp}(q,a)} p(q,a,o) & = \sum_{o \in \msf{Comp}(q,a)} \delta(q,a)(\msf{sc}(q,a,o)) = \sum_{o \in \msf{Comp}(q,a)} \sum_{q' \in O^{-1}(o)}  \delta(q,a)(q') = \sum_{q' \in Q} \delta(q,a)(q') = 1
		\end{align*}
		
		In addition, for all beliefs $b \in \Dist(Q)$, actions $a \in A$, and observations $o \in \msf{Comp}(b,a)$:
		\begin{align*}
			p(b,a,o) & = \sum_{q \in O^{-1}(o)} \sum_{q' \in Q} b(q') \cdot \delta(q',a)(q) = \sum_{q \in O^{-1}(o)} \sum_{q' \in \msf{Prec}(q,a,o)} b(q') \cdot \delta(q',a)(q) \\
			& = \sum_{q \in O^{-1}(o)} b(q) \cdot \delta(q,a)(\msf{sc}(q,a,o)) = \sum_{q \in O^{-1}(o)} b(q) \cdot p(q,a,o)
		\end{align*}
		
		Now, we define a belief define $\lambda'(b,a,o) \in \Dist(Q)$ as follows, for all $q \in Q$:
		\begin{align*}
			\lambda'(b,a,o)(q) := \begin{cases}
				0 & \text{ if }O(q) \neq o \\
				%0 & \text{ if }o \notin \msf{Comp}(q,a) \\
				%\frac{1}{p(b,a,o)} \cdot b(q) \cdot \delta(q,a)(\msf{sc}(q,a,o)) = 
				\frac{1}{p(b,a,o)} \cdot b(q) \cdot p(q,a,o) & \text{ otherwise }
			\end{cases}
		\end{align*}
		
		Let us show the two following facts, for all beliefs $b \in \Dist(Q)$ and actions $a \in A$:
		\begin{equation}
			\label{eqn:p'_geq_p}
			\forall o \in \msf{Comp}(b,a):\; H(\lambda'(b,a,o)) \geq H(\lambda(b,a,o))
		\end{equation}
		\begin{equation}
			\label{eqn:new_expected_value}
			H(b) \geq \sum_{o \in \msf{Comp}(b,a)} p(b,a,o) \cdot H(\lambda'(b,a,o))
		\end{equation}
		
		Let us start by showing Equation~(\ref{eqn:p'_geq_p}). For all beliefs $b \in \Dist(Q)$, actions $a \in A$, compatible observations $o \in \msf{Comp}(b,a)$, and $q \in O^{-1}(o)$, we have:
		\begin{align*}
			\lambda(b,a,o)(q) & = \frac{1}{p(b,a,o)} \sum_{q' \in Q} b(q') \cdot \delta(q',a)(q) = \sum_{q' \in \msf{Prec}(q,a,o)} \frac{b(q') \cdot p(q',a,o)}{p(b,a,o)} = \sum_{q' \in \msf{Prec}(q,a,o)} \lambda'(b,a,o)(q')
		\end{align*}
		Therefore, we have:
		\begin{align*}
			H(\lambda(b,a,o)) & = -\sum_{q \in O^{-1}(o)} \lambda(b,a,o)(q) \cdot \log(\lambda(b,a,o)(q)) \\
			& = -\sum_{q \in O^{-1}(o)} \left(\sum_{q' \in \msf{Prec}(q,a,o)} \lambda'(b,a,o)(q')\right) \cdot \log\left(\sum_{q' \in \msf{Prec}(q,a,o)} \lambda'(b,a,o)(q')\right) \\
			& = -\sum_{q \in O^{-1}(o)} \sum_{q' \in \msf{Prec}(q,a,o)} \left(\lambda'(b,a,o)(q') \cdot \log\left(\sum_{q'' \in \msf{Prec}(q,a,o)} \lambda'(b,a,o)(q'')\right)\right) \\
			& \leq -\sum_{q \in O^{-1}(o)} \sum_{q' \in \msf{Prec}(q,a,o)} \lambda'(b,a,o)(q') \cdot \log\left(\lambda'(b,a,o)(q')\right) \\
			& = - \sum_{\substack{q' \in Q \\ o \in \msf{Comp}(q',a)}} \lambda'(b,a,o)(q') \cdot \log\left(\lambda'(b,a,o)(q')\right) \\
			& = - \sum_{q' \in Q} \lambda'(b,a,o)(q') \cdot \log\left(\lambda'(b,a,o)(q')\right) = H(\lambda'(b,a,o))
		\end{align*}
		The last but least equality comes from the fact that $\{ q' \in \msf{Prec}(q,a,o) \mid q \in O^{-1}(o) \} = \{ q' \in Q \mid \msf{sc}(q',a,o) = q,\; q \in O^{-1}(o) \} = \{ q' \in Q \mid \msf{sc}(q',a,o) \in O^{-1}(o) \} = \{ q' \in Q \mid o \in \msf{Comp}(q',a) \}$. The last equality comes from the fact that, for all $q' \in Q$ such that $o \notin \msf{Comp}(q',a)$, we have $p(q',a,o) = 0$ and thus $\lambda'(b,a,o)(q') = 0$. This proves Equation~(\ref{eqn:p'_geq_p}). 
		
		Let us now turn to Equation~(\ref{eqn:new_expected_value}). %For readability, we will omit $b$ and $a$ when writing $\msf{Comp}(b,a)$, $p(b,a,o)$, $\lambda'(b,a,o)$, $\msf{sc}(q,a,o)$, $\delta(q,a)$, and $O(q,a)$, and thus write $\msf{Comp}$, $p(o)$, $\lambda(o)$, $\msf{sc}(q,o)$, $\delta(q)$, and $O(q)$ instead. 
		Letting $x := \sum_{o \in \msf{Comp}(b,a)} p(b,a,o) \cdot H(\lambda'(b,a,o))$, and by Jensen inequality, since the $\log$ is concave, we have:
		\begin{align*}
			x & = \sum_{o \in \msf{Comp}(b,a)} p(b,a,o) \cdot \left(-\sum_{q \in Q} \lambda'(b,a,o)(q) \cdot \log(\lambda'(b,a,o)(q))\right) \\
			& = - \sum_{o \in \msf{Comp}(b,a)} \sum_{q \in Q} \left[b(q) \cdot p(q,a,o) \cdot \log\left(\frac{b(q) \cdot p(q,a,o)}{p(b,a,o)} \right)\right] \\
			& = - \sum_{o \in \msf{Comp}(b,a)} \sum_{q \in Q} b(q) \cdot p(q,a,o) \cdot \log\left(b(q)\right) - \sum_{o \in \msf{Comp}(b,a)} \sum_{q \in Q} \left[b(q) \cdot p(q,a,o) \cdot \log\left(\frac{p(q,a,o)}{p(b,a,o)} \right)\right] \\
			& = - \sum_{q \in Q} b(q) \cdot \log\left(b(q)\right) + \sum_{o \in \msf{Comp}(b,a)} p(b,a,o) \cdot \left[\sum_{q \in Q} \frac{b(q) \cdot p(q,a,o)}{p(b,a,o)} \cdot \log\left(\frac{p(b,a,o)}{p(q,a,o)} \right)\right] \\ \\
			& \leq H(b) + \sum_{o \in \msf{Comp}(b,a)} p(b,a,o) \cdot \log\left(\sum_{q \in Q} \frac{b(q) \cdot p(q,a,o)}{p(b,a,o)} \cdot \frac{p(b,a,o)}{p(q,a,o)}\right) \\
			& = H(b) + \sum_{o \in \msf{Comp}(b,a)} p(b,a,o) \cdot \log\left(\sum_{q \in Q} b(q)\right) = H(b)
		\end{align*}
		Thus, Equation~(\ref{eqn:new_expected_value}) holds. From Equations~(\ref{eqn:p'_geq_p}) and~(\ref{eqn:new_expected_value}), it straightforwardly follows that $H(b) \geq \sum_{o \in \msf{Comp}(b,a)} p(b,a,o) \cdot H(\lambda(b,a,o))$. Hence, the POMDP $\POMDP$ has non-increasing entropy.
	\end{proof}
	
	\subsection{Proof of Theorem~\ref{thm:Dirac-preserving-pomdp-are-memdps}}
	\label{proof:thm-Dirac-preserving-pomdp-are-memdps}
	We introduce below the definitions that we will use to define an MEMDP from a Dirac-preserving POMDP. Note that we will use the notation introduced in Definition~\ref{def:notation_pomdp}.
	\begin{definition}
		\label{def:equivalence relation on sequences}
		Consider a Dirac-preserving POMDP $\POMDP = (Q,A,\delta,\Omega,O)$. For all $S \subseteq Q$, we let $T_S := \{ (t_s)_{s \in S} \in (Q \cup \{\bot\})^S \mid \exists o \in O, \forall s \in S:\; t_s = \bot \text{ or }O(t_s) = o\}$. We let $\bot_S := (\bot_s)_{s \in S} \in T_S$. For all $t \in T_S$ and $s \in S$, we let $t[s] \in Q \cup \{\bot\}$ denote $s$-indexed element of $t$. We also let $O(t) \in \Omega$ denote the observation common to all elements in $t$ (if $t = \bot_S$, $O(t) \in \Omega$ is defined arbitrarily).
		
		For all $t \in T_s$, $a \in A$, and $o \in O$, we let $\msf{sc}(t,a,o) \in T_S$ be such that, for all $s \in S$:
		\begin{align*}
			\msf{sc}(t,a,o)[s] := \begin{cases}
				\msf{sc}(t[s],a,o) & \text{ if }o \in \msf{Comp}(t[s],a) \\
				\bot & \text{ otherwise }
			\end{cases}
		\end{align*}
		%Note that $O(\msf{sc}(t,a,o)) = o$ since, for all $q \in Q$ and $o \in \msf{Comp}(q,a)$, $O(\msf{sc}(q,a,o)) = o$.
		%Finally, for all $q \in Q$ and $a \in A$, we let $d[q,a] \in \Dist(O)$ be such that, for all $o \in O$: $d[q,a](o) := \delta(q,a)(\msf{sc}(q,a,o)) = p[q,a,o]$.
		
		Finally, for all $\rho = \rho_0 \cdots \rho_n \in (T_S)^+$, and $s \in S$, we let $\rho[s] := \rho_0[s] \cdots \rho_n[s] \in (Q \cup \{\bot\})^+$. This is done similarly for $\rho \in (T_S)^\omega,T_S \cdot (A \cdot T_S)^*,(T_S \cdot A)^\omega$ with $\rho[s] \in (Q \cup \{\bot\})^\omega,(Q \cup \{\bot\}) \cdot (A \cdot (Q \cup \{\bot\}))^*,((Q \cup \{\bot\}) \cdot A)^\omega$ respectively. 
		
		Let $s \in S$. For all $W \in \msf{Borel}(Q)$, we let $W_s \in \msf{Borel}(T_S)$ be defined by $\{ \rho \in (T_S)^\omega \mid \rho[s] \in W\}$. We also let $\msf{Incomp}_s := \{ \rho \in (T_S)^\omega \mid \exists i \in \N,\; \rho_i[s] = \bot \}$.
	\end{definition}
	
	Let us now define an MEMDP derived from a Dirac-preserving POMDP.
	\begin{definition}
		\label{def:MEMDP_from_POMDP}
		Consider a Dirac-preserving POMDP $\POMDP = (Q,A,\delta,\Omega,O)$ and an \emph{observation-compatible} $S \subseteq Q$, i.e. such that $O(S) = \{o_S\}$ for some $o_S \in \Omega$. We let $\MEMDP(\POMDP)_S = (Q',A,E,(\delta_e)_{e \in E})$ be an MEMDP defined by:
		\begin{itemize}
			\item $Q' := T_S$;
			\item $E := S$;
			\item for all $t \in T_S$, $a \in A$ and $s \in E = S$, we let $\delta_s(t,a) \in \Dist(Q') = \Dist(T_S)$ be such that if $t[s] = \bot$, then $\delta_s(t,a)(\bot_S) := 1$; otherwise, for all $o \in O$:
			\begin{equation*}
				\delta_s(t,a)(\msf{sc}(t,a,o)) := \delta(t[s],a)(\msf{sc}(t[s],a,o)) = p[t[s],a,o]
			\end{equation*} 
		\end{itemize}
		We let $t_S^{\msf{init}} := (s)_{s \in S} \in T_S$ (since $S$ is observation-compatible).
		
		Consider now an observation-compatible parity objective $W \in \msf{Borel}(Q)$ defined by a function $f\colon Q \to \N$, and $S \subseteq Q$. We let $f_S: T_S \to \N$ be such that, for all $t \in T_S \setminus \{ \bot_S \}$, $f_S(t) := f(q)$, for any $q \in Q$ such that $O(q) = O(t)$ (recall that $f$ is observation-compatible); $f_S(\bot_S) := 1$. We let $W_S \in \msf{Borel}(T_S)$ denote the parity objective induced by $f_S$.
	\end{definition}
	
	This definition satisfies the following properties.
	\begin{lemma}
		\label{lem:memdp-from-pomdp}
		Consider a Dirac-preserving POMDP $\POMDP = (Q,A,\delta,\Omega,O)$ and an observation-compatible subset $S \subseteq Q$.
		\begin{enumerate}
			\item For all observation-compatible parity objectives $W \in \msf{Borel}(Q)$ and $s \in E = S$:%, for all $\sigma \in \msf{Strat}(T_S,A)$:
			\begin{equation*}
				%\Prb_{t_S^{\msf{init}}}^{\sigma}[\MEMDP(\POMDP)_S[s],W_s] = \Prb_{t_S^{\msf{init}}}^{\sigma}[\MEMDP(\POMDP)_S[s],W_S \setminus \msf{Incomp}_s]
				W_s = W_S \setminus \msf{Incomp}_s
			\end{equation*}
			\item For all $\sigma \in \msf{Strat}(T_S,A)$, for all $s \in S = E$, we have:
			\begin{equation*}
				\Prb_{t_S^{\msf{init}}}^{\sigma}[\MEMDP(\POMDP)_S[s],\msf{Incomp}_s] = 0
			\end{equation*}
			\item For all $\kappa \in \msf{Strat}(\Omega,A)$, there is $\sigma(\kappa) \in \msf{Strat}(T_S,A)$ and for all $\sigma \in \msf{Strat}(T_S,A)$, there is $\kappa(\sigma) \in \msf{Strat}(\Omega,A)$ such that, for all Borel sets $W \in \msf{Borel}(Q)$, and $s \in S = E$:
			\begin{equation*}
				\Prb_{s}^\kappa[\POMDP,W] =  \Prb_{t_S^{\msf{init}}}^{\sigma(\kappa)}[\MEMDP(\POMDP)_S[s],W_s] \text{ and } \Prb_{q}^{\kappa(\sigma)}[\POMDP,W] =  \Prb_{t_S^{\msf{init}}}^{\sigma}[\MEMDP(\POMDP)_S[s],W_s]
			\end{equation*}
		\end{enumerate}
	\end{lemma}
	\begin{proof}
		\begin{enumerate}
			\item For all $\rho \in W_s$, we have $\rho[s] \in W$. This implies that for all $i \in \N$, we have $\rho_i[s] \neq \bot$, thus $\rho \notin \msf{Incomp}_s$. In fact, $W_s \cap \msf{Incomp}_s = \emptyset$. 
			                                                                        
			Consider any $\rho \in (T_S)^\omega \setminus \msf{Incomp}_s$. For all $i \in \N$, we have $\rho_i[s] \neq \bot$; thus $O(\rho_i) = O(\rho_i[s])$ and $f_S(\rho_i) = f(\rho_i[s])$. Thus, $\rho \in W_S$ if and only if $\rho[s] \in W$ if and only if $\rho \in W_s$. Hence: $W_s = W_S \setminus \msf{Incomp}_s$. 
			\iffalse
			Let $\rho \in W_s$. By definition, this means that $\rho[s] \in W$. in particular, this implies that for all $i \in \N$, we have $\rho_i[s] \neq \bot$, thus $\rho \notin \msf{Incomp}_s$. Furthermore, for all $i \in \N$, we have $O(\rho_i) = O(\rho_i[s])$, thus $f(\rho_i) = f_S(\rho_i[s])$. Hence, $\rho \in W_S$. Hence: $W_s \subseteq W_S \setminus \msf{Incomp}_s$
			
			Conversely, let $\rho \in W_S \setminus \msf{Incomp}_s$. For all $i \in \N$, we have $\rho_i[s] \neq \bot$; thus $O(\rho_i) = O(\rho_i[s])$ and $f(\rho_i) = f_S(\rho_i[s])$. 
			
			Furthermore, let $\msf{FullIncomp}_s := \{ \rho \in (T_S)^\omega \mid \exists j \in \N,\; \forall i \geq j,\; \rho_i[s] = \bot \}$. By construction, for all $a \in A$ and $t \in T_S$ such that $t[s] = \bot$, we have $\delta_s(t,a)(\bot_S) = 1$. Therefore, for all $\sigma \in \msf{Strat}(T_S,A)$, we have:
			\begin{equation*}
				\Prb_{t_S^{\msf{init}}}^{\sigma}[\MEMDP(\POMDP)_S[s],\msf{Incomp}_s \cap \msf{FullIncomp}_s] = \Prb_{t_S^{\msf{init}}}^{\sigma}[\MEMDP(\POMDP)_S[s],\msf{Incomp}_s]
			\end{equation*}
			In addition, $W_S \cap \msf{FullIncomp}_s = \emptyset$ since $f(\bot_S)$ is odd. 
			\fi
			\item Let $\sigma \in \msf{Strat}(T_S,A)$ and $s \in S = E$. Let us show by induction on $\rho \in T_S \cdot (A \cdot T_S)^*$ that if $\Prb_{t_S^{\msf{init}}}^{\sigma}[\MEMDP(\POMDP)_S[s],\rho] > 0$, then $\head{\rho}[s] \neq \bot$. This clearly holds for $\rho \in T_S$. Assume now that this holds for some $\rho \in T_S \cdot (A \cdot T_S)^*$ and let $(a,t) \in A \times T_S$ such that $\Prb_{t_S^{\msf{init}}}^{\sigma}[\MEMDP(\POMDP)_S[s],\rho \cdot (a,t)] > 0$, with: 
			\begin{equation*}
				\Prb_{t_S^{\msf{init}}}^{\sigma}[\MEMDP(\POMDP)_S[s],\rho \cdot (a,t)] = \Prb_{t_S^{\msf{init}}}^{\sigma}[\MEMDP(\POMDP)_S[s],\rho] \cdot \sigma(\rho)(a) \cdot \delta_s(\head{\rho},a)(t)
			\end{equation*}
			We have $\Prb_{t_S^{\msf{init}}}^{\sigma}[\MEMDP(\POMDP)_S[s],\rho] > 0$, thus $\head{\rho}[s] \neq \bot$. Furthermore, $\delta_s(\head{\rho},a)(t) > 0$ means that there is $o \in O$ such that $t = \msf{sc}(\head{\rho},a,o)$. In that case, we have $\delta_s(\head{\rho},a)(t) = p[\head{\rho}[s],a,o]$, which is equal to 0 if $o \notin \msf{Comp}(\head{\rho}[s],a)$. Thus, $o \in \msf{Comp}(\head{\rho}[s],a)$, and $t[s] = \msf{sc}(\head{\rho},a,o)[s] = \msf{sc}(\head{\rho}[s],a,o) \in Q$. Thus, the inductive property holds for $\rho \cdot (a,t)$ as well. In fact, it holds for all $\rho \in T_S \cdot (A \cdot T_S)^*$. Thus, we have:
			\begin{equation*}
				\Prb_{t_S^{\msf{init}}}^{\sigma}[\MEMDP(\POMDP)_S[s],\msf{Incomp}_s] = \sum_{\substack{\rho \in T_S \cdot (A \cdot T_S)^* \\ \head{\rho} = \bot}} \Prb_{t_S^{\msf{init}}}^{\sigma}[\MEMDP(\POMDP)_S[s],\rho] = 0
			\end{equation*}
			\item Let $o_S \in O$ be such that $O(s) = o_S$ for all $s \in S$. 
			
			For all $s \in S = E$, we define by induction a function $\msf{sc}^*_S\colon \Omega \cdot (A \times \Omega)^* \to T_S \cdot (A \cdot T_S)^*$ as follows: for all $\rho \in \Omega \cdot (A \times \Omega)^* \setminus o_S \cdot (A \times \Omega)^*$, $\msf{sc}^*_S(\rho) := \bot_S$; furthermore, $\msf{sc}^*_S(o_S) := t_S^{\msf{init}}$ and, for all $\rho \in o_S \cdot (A \times \Omega)^*$ and $(a,o) \in A \times \Omega$:
			\begin{equation*}
				\msf{sc}^*_S(\rho \cdot (a,o)) := \msf{sc}^*_S(\rho) \cdot (a,\msf{sc}(\head{\msf{sc}^*_S(\rho)},a,o)) \in T_S \cdot (A \cdot T_S)^*
			\end{equation*}
			For all $s \in S$, we let $\msf{Compatible}_s(A,\Omega) := \{ \rho \in o_s \cdot (A \times \Omega)^* \mid \head{\msf{sc}^*_S(\rho)[s]} \neq \bot \}$. We let $\msf{Compatible}_S(A,\Omega) := \cup_{s \in S} \msf{Compatible}_s(A,\Omega)$. For all $\pi \in Q \cdot (A \cdot Q)^*$, we let $O(\pi) \in \Omega \cdot (A \cdot \Omega)^*$ be the sequence of actions and observations associated to $\pi$. In particular, for all $s \in S$ and $\rho \in \msf{Compatible}_s(A,\Omega)$, we have $O(\msf{sc}^*_S(\rho)[s]) = \rho$.
			
			We show several intermediary results.
			\begin{enumerate}
				\item Let $s \in S$ and $\kappa \in \msf{Strat}(\Omega,A)$. As a first step, we show by induction on $\pi \in s \cdot (A \cdot Q)^*$ that if $\Prb_{s}^\kappa[\POMDP,\pi] > 0$, then $\msf{sc}^*_S(O(\pi))[s] = \pi$. This holds for $\pi = s$. Assume now that it holds for some $\pi \in s \cdot (A \cdot Q)^*$, and let $(a,q) \in A \times Q$. Assume that $\Prb_{s}^\kappa[\POMDP,\pi \cdot (a,q)] > 0$. We have:
				\begin{equation*}
					\Prb_{s}^\kappa[\POMDP,\pi \cdot (a,q)]¨ = \Prb_{s}^\kappa[\POMDP,\pi] \cdot \kappa(O(\pi))(a) \cdot \delta(\head{\pi},a)(q) > 0
				\end{equation*}
				Hence, $\delta(\head{\pi},a)(q) > 0$ and thus, for $o := O(q) \in \Omega$, $o \in \msf{Comp}(\head{\pi},a)$ and $\msf{sc}(\head{\pi},a,o) = q$. Therefore:
				\begin{align*}
					\msf{sc}^*_S(O(\pi \cdot (a,q)))[s] & = \msf{sc}^*_S(O(\pi) \cdot (a,o))[s] = (\msf{sc}^*_S(O(\pi)) \cdot (a,\msf{sc}(\head{\msf{sc}^*_S(O(\pi))},a,o)))[s] \\
					& = \pi \cdot (a,\msf{sc}(\head{\msf{sc}^*_S(O(\pi))},a,o)[s]) = \pi \cdot (a,\msf{sc}(\head{\msf{sc}^*_S(O(\pi))[s]},a,o)) \\
					& = \pi \cdot (a,\msf{sc}(\head{\pi},a,o)) = \pi \cdot (a,q)
				\end{align*} 
				Hence, the equality holds for all $\pi \in s \cdot (A \cdot Q)^*$ such that $\Prb_{s}^\kappa[\POMDP,\pi] > 0$.
				\item Let $s \in S$ and $\kappa \in \msf{Strat}(\Omega,A)$. Let us show that for all $\pi \in s \cdot (A \cdot Q)^*$ such that
				$O(\pi) \notin \msf{Compatible}_s(A,\Omega)$, we have $\Prb_{s}^\kappa[\POMDP,\pi] = 0$. Since $O(s) = o_S \in \msf{Compatible}_s(A,\Omega)$, this holds for $\pi = s$. Consider now some $\pi \in s \cdot (A \cdot Q)^*$ and $(a,q) \in A \times Q$. Assume that $O(\pi \cdot (a,q)) \notin \msf{Compatible}_s(A,\Omega)$. We have:
				\begin{equation*}
					\Prb_{s}^\kappa[\POMDP,\pi \cdot (a,q)] = \Prb_{s}^\kappa[\POMDP,\pi] \cdot \kappa(O(\pi))(a) \cdot \delta(\head{\pi},a)(q) 
				\end{equation*}
				If $\Prb_{s}^\kappa[\POMDP,\pi] = 0$, then $\Prb_{s}^\kappa[\POMDP,\pi \cdot (a,q)] = 0$. Thus, assume that $\Prb_{s}^\kappa[\POMDP,\pi] > 0$. By item (a), this implies that $\head{\msf{sc}^*_S(O(\pi))[s]} = \head{\pi} \in Q$. In addition, since $O(\pi \cdot (a,q)) \notin \msf{Compatible}_s(A,\Omega)$, we have $\msf{sc}(\head{\msf{sc}^*_S(O(\pi))[s]},a,o) = \head{\msf{sc}^*_S(O(\pi \cdot (a,q)))[s]} = \bot$. That is, $o \notin \msf{Comp}(\head{\msf{sc}^*_S(O(\pi))[s]},a) = \msf{Comp}(\head{\pi},a)$. Since $O(q) = o$, this implies $\delta(\head{\pi},a)(q) = 0$, and thus $\Prb_{s}^\kappa[\POMDP,\pi \cdot (a,q)] = 0$.
				\item Consider now any strategy $\sigma \in \msf{Strat}(T_S,A)$. Let $s \in S$ and $\pi \in s \cdot (A \cdot Q)^*$. By definition, we have:
				\begin{equation*}
					\cyl(\pi)_s = \{ \rho \in (T_S \cdot A)^\omega \mid \rho[s] \in \cyl(\pi)  \} = \{ \rho \in (T_S \cdot A)^\omega \mid \rho[s] \in \pi \cdot (A \cdot Q)^\omega \}
				\end{equation*}
				We define $\cyl(\pi)_s^\bot \in \msf{Borel}(T_S \cdot A)$ as follows:
				\begin{equation*}
					\cyl(\pi)_s^\bot = \{ \rho \in (T_S \cdot A)^\omega \mid \rho[s] \in \pi \cdot (A \cdot (Q \cup \{\bot\}))^\omega \}
				\end{equation*}
				We have:
				\begin{equation*}
					\cyl(\pi)_s \subseteq \cyl(\pi)_s^\bot \subseteq \cyl(\pi)_s \cup \msf{Incomp}_s
				\end{equation*}
				Therefore, by item (b), we have:
				\begin{equation*}
					\Prb_{t_S^{\msf{init}}}^{\sigma}[\MEMDP(\POMDP)_S[s],\cyl(\pi)_s] = \Prb_{t_S^{\msf{init}}}^{\sigma}[\MEMDP(\POMDP)_S[s],\cyl(\pi)_s^\bot]
				\end{equation*}
				In addition, we have:
				\begin{equation*}
					\Prb_{t_S^{\msf{init}}}^{\sigma}[\MEMDP(\POMDP)_S[s],\cyl(\pi)_s^\bot] = \sum_{\substack{\rho \in T_S \cdot (A \cdot T_S)^* \\ \rho[s] = \pi}} \Prb_{t_S^{\msf{init}}}^{\sigma}[\MEMDP(\POMDP)_S[s],\rho]
				\end{equation*}
				Let us show by induction on $\pi \in s \cdot (A \cdot Q)^*$ that, for all $\rho \in T_S \cdot (A \cdot T_S)^*$ such that $\rho[s] = \pi$, if $\Prb_{t_S^{\msf{init}}}^{\sigma}[\MEMDP(\POMDP)_S[s],\rho] > 0$, then $\rho = \msf{sc}^*_S(O(\pi))$. This clearly holds if $\pi = s$. Now, assume that it holds for some $\pi \in s \cdot (A \cdot Q)^*$ and let $(a,q) \in A \times Q$. Consider any $\rho \cdot (a,t) \in T_S \cdot (A \cdot T_S)^+$ such that $\rho \cdot (a,t)[s] = \pi \cdot (a,q)$ and assume that $\Prb_{t_S^{\msf{init}}}^{\sigma}[\MEMDP(\POMDP)_S[s],\rho \cdot (a,t)] > 0$. We have:
				\begin{equation*}
					\Prb_{t_S^{\msf{init}}}^{\sigma}[\MEMDP(\POMDP)_S[s],\rho \cdot (a,t)] = \Prb_{t_S^{\msf{init}}}^{\sigma}[\MEMDP(\POMDP)_S[s],\rho] \cdot \sigma(\rho)(a) \cdot \delta_s(\head{\rho},a)(t) > 0
				\end{equation*}
				Since $\Prb_{t_S^{\msf{init}}}^{\sigma}[\MEMDP(\POMDP)_S[s],\rho] > 0$, by our induction hypothesis, we have $\rho = \msf{sc}^*_S(O(\pi))$. Furthermore, since $\delta_s(\head{\rho},a)(t) > 0$, there is $o \in \Omega$ such that $t = \msf{sc}(\head{\rho},a,o)$, with $q = t[s] = \msf{sc}(\head{\rho}[s],a,o)$, and thus $O(q) = O(t[s]) = o$. Therefore, we have:
				\begin{align*}
					\msf{sc}^*_S(O(\pi \cdot (a,q))) & = \msf{sc}^*_S(O(\pi) \cdot (a,o)) = \msf{sc}^*_S(O(\pi)) \cdot (a,\msf{sc}(\head{\msf{sc}^*_S(O(\pi))},a,o)) \\
					& = \rho \cdot (a,\msf{sc}(\head{\rho},a,o)) = \rho \cdot (a,t)
				\end{align*}
				In fact, our induction hypothesis holds for all $\pi \in s \cdot (A \cdot Q)^*$. We deduce that, for all $\pi \in s \cdot (A \cdot Q)^*$:
				\begin{align*}
					\Prb_{t_S^{\msf{init}}}^{\sigma}[\MEMDP(\POMDP)_S[s],\cyl(\pi)_s] = %\Prb_{t_S^{\msf{init}}}^{\sigma}[\MEMDP(\POMDP)_S[s],\cyl(\pi)_s^\bot] = 
					\begin{cases}
						0 & \text{if }\msf{sc}^*_S(O(\pi))[s] \neq \pi \\
						\Prb_{t_S^{\msf{init}}}^{\sigma}[\MEMDP(\POMDP)_S[s],\msf{sc}^*_S(O(\pi))] & \text{otherwise}
					\end{cases}
				\end{align*}
				\item Now, consider two strategies $\kappa \in \msf{Strat}(\Omega,A)$ and $\sigma \in \msf{Strat}(T_S,A)$, and assume that they are \emph{similar}, i.e. for all $\rho \in \msf{Compatible}_S(A,\Omega)$, we have $\kappa(\rho) = \sigma(\msf{sc}^*_S(\rho)) \in \Dist(A)$. We show by induction on $\pi \in s \cdot (A \cdot Q)^*$ that:
				\begin{equation*}
					\Prb_{s}^\kappa[\POMDP,\pi] =  \Prb_{t_S^{\msf{init}}}^{\sigma}[\MEMDP(\POMDP)_S[s],%\msf{sc}^*_S(O(\pi))
					\cyl(\pi)_s]
				\end{equation*}
				This clearly holds for $\pi = s$. Assume now that it holds for some $\pi \in s \cdot (A \cdot Q)^*$ and let $(a,q) \in A \times Q$. Let $o := O(q) \in \Omega$. We have:
				\begin{equation*}
					\msf{sc}^*_S(O(\pi \cdot (a,q)) = \msf{sc}^*_S(O(\pi) \cdot (a,o)) = \msf{sc}^*_S(O(\pi)) \cdot (a,\msf{sc}(\head{\msf{sc}^*_S(O(\pi))},a,o))
				\end{equation*}
				If $\Prb_{s}^\kappa[\POMDP,\pi] = 0$, then $\Prb_{s}^\kappa[\POMDP,\pi \cdot (a,q)] = 0$, $\Prb_{t_S^{\msf{init}}}^{\sigma}[\MEMDP(\POMDP)_S[s],\msf{sc}^*_S(O(\pi))] = 0$ and thus $\Prb_{t_S^{\msf{init}}}^{\sigma}[\MEMDP(\POMDP)_S[s],\msf{sc}^*_S(O(\pi \cdot (a,q)))] = 0$. Assume now that $\Prb_{s}^\kappa[\POMDP,\pi] > 0$. By item (b), this implies $O(\pi) \in \msf{Compatible}_s(A,\Omega)$; by item (a), letting $\rho := \msf{sc}^*_S(O(\pi))$, this implies $\rho[s] = \pi$. Now, if $\delta(\head{\pi},a)(q) = \delta(\head{\rho}[s],a)(q) = 0$, then $\Prb_{s}^\kappa[\POMDP,\pi \cdot (a,q)] = 0$ and either $o = O(q) \notin \msf{Comp}(\head{\rho}[s],a)$ or $o \in \msf{Comp}(\head{\rho}[s],a)$ and $\msf{sc}(\head{\rho}[s],a,o) \neq q$. In both cases, we have $\head{\msf{sc}^*_S(O(\pi \cdot (a,q))} \neq q$, and thus $	\msf{sc}^*_S(O(\pi \cdot (a,q))[s] \neq \pi \cdot (a,q)$, which implies $\Prb_{t_S^{\msf{init}}}^{\sigma}[\MEMDP(\POMDP)_S[s],\cyl(\pi \cdot (a,q))_s] = 0$. Assume now that $\delta(\head{\pi},a)(q) > 0$, and therefore $o \in \msf{Comp}(\head{\pi},a)$ and $q = \msf{sc}(\head{\pi},a,o) = \msf{sc}(\head{\rho}[s],a,o)$, which implies $\msf{sc}^*_S(O(\pi \cdot (a,q))[s] = \pi \cdot (a,q)$. Then, we have:
				\begin{align*}
					\Prb_{s}^\kappa[\POMDP,\pi \cdot (a,q)] & = \Prb_{s}^\kappa[\POMDP,\pi] \cdot \kappa(O(\pi))(a) \cdot \delta(\head{\pi},a)(q) \\ 
					& = \Prb_{t_S^{\msf{init}}}^{\sigma}[\MEMDP(\POMDP)_S[s],\rho] \cdot \sigma(\rho)(a) \cdot \delta(\head{\rho}[s],a)(q) \\ 
					& = \Prb_{t_S^{\msf{init}}}^{\sigma}[\MEMDP(\POMDP)_S[s],\rho] \cdot \sigma(\rho)(a) \cdot \delta_s(\head{\rho},a)(\msf{sc}(\head{\rho},a,o)) \\ 
					& = \Prb_{t_S^{\msf{init}}}^{\sigma}[\MEMDP(\POMDP)_S[s],\rho \cdot (a,\msf{sc}(\head{\rho},a,o))] \\
					& = \Prb_{t_S^{\msf{init}}}^{\sigma}[\MEMDP(\POMDP)_S[s],\msf{sc}^*_S(O(\pi \cdot (a,q)))] \\
					& = \Prb_{t_S^{\msf{init}}}^{\sigma}[\MEMDP(\POMDP)_S[s],\cyl(\pi \cdot (a,q))_s]
				\end{align*}
				Thus the property holds for $\pi \cdot (a,q) \in s \cdot (A \cdot Q)^+$. In fact, it holds for all $\pi \in s \cdot (A \cdot Q)^*$.
				\item Let $s \in S$. Consider two similar strategies $\kappa \in \msf{Strat}(\Omega,A)$ and $\sigma \in \msf{Strat}(T_S,A)$. We let $\Prb_{\kappa}\colon \msf{Borel}(Q) \to [0,1]$ such that, for all $W \in \msf{Borel}(Q)$, we have $\Prb_{\kappa}(W) := \Prb_s^\kappa[\POMDP,W] \in [0,1]$. We also let $\Prb_{\sigma}\colon \msf{Borel}(Q) \to [0,1]$ such that, for all $W \in \msf{Borel}(Q)$, we have $\Prb_{\sigma}(W) := \Prb_{t_S^{\msf{init}}}^{\sigma}[\MEMDP(\POMDP)_S[s],W_s]$. Clearly, $\Prb_{\kappa}$ is a probability measure. This is also the case of $\Prb_{\sigma}$ since $\Prb_{\sigma}[\emptyset] = 0$; $\Prb_{\sigma}[Q^\omega] = \Prb_{t_S^{\msf{init}}}^{\sigma}[\MEMDP(\POMDP)_S[s],(T_S \cdot A)^\omega \setminus W_s] = 1$ (by item (2)); and for all $(W^n)_{n \in \N} \in \msf{Borel}(Q)^{\N}$ pairwise disjoint, we have $(\cup_{n \in \N} W^n)_s = \cup_{n \in \N} W^n_s$, thus $\Prb_{\sigma}[\cup_{n \in \N} W_n] = \sum_{n \in \N} \Prb_{\sigma}[W_n]$. For all $\rho = s \cdot q_1 \cdots q_n \in s \cdot Q^*$, we let $\msf{Proj}(\rho) := \{ s \cdot (a_1,q_1) \cdots (a_{n},q_n) \in Q \cdot (A \cdot Q)^* \mid a_1,\ldots,a_{n} \in A \}$. Then, for all $\rho \in s \cdot Q^*$, we have, by (d):
				\begin{align*}
					\Prb_{\kappa}[\cyl(\rho)] & = \Prb_s^\kappa[\POMDP,\cyl(\rho)] = \sum_{\pi \in \msf{Proj}(\rho)} \Prb_s^\kappa[\POMDP,\cyl(\pi)] = \sum_{\pi \in \msf{Proj}(\rho)} \Prb_s^\kappa[\POMDP,\pi] \\
					& %= \sum_{\pi \in \msf{Proj}(\rho)} \Prb_{t_S^{\msf{init}}}^{\sigma}[\MEMDP(\POMDP)_S[s],\msf{sc}^*_S(O(\pi))] 
					= \sum_{\pi \in \msf{Proj}(\rho)} \Prb_{t_S^{\msf{init}}}^{\sigma}[\MEMDP(\POMDP)_S[s],\cyl(\pi)_s] = \Prb_{t_S^{\msf{init}}}^{\sigma}[\MEMDP(\POMDP)_S[s],\cyl(\rho)_s] = \Prb_{\sigma}[\cyl(\rho)]
				\end{align*}
				Furthermore, for all $\rho \in Q^+ \setminus s \cdot Q^*$, we have $\Prb_{\kappa}[\cyl(\rho)] = 0 = \Prb_{\sigma}[\cyl(\rho)]$. Therefore, the two probability measures $\Prb_{\kappa}$ and $\Prb_{\sigma}$ coincide on all cylinder sets, and thus, by Caratheodory's extension theorem, they are equal: $\Prb_{\kappa} = \Prb_{\sigma}$, that is, for all $W \in \msf{Borel}(Q)$:
				\begin{equation*}
					\Prb_s^\kappa[\POMDP,W] = \Prb_{t_S^{\msf{init}}}^{\sigma}[\MEMDP(\POMDP)_S[s],W_s]
				\end{equation*}
				\item For all $\sigma \in \msf{Strat}(T_S,A)$, we let $\kappa(\sigma) \in \msf{Strat}(\Omega,A)$ be such that, for all $\rho \in \Omega \cdot (A \cdot \Omega)^*$: $\kappa(\sigma)(\rho) := \sigma(\msf{sc}^*_S(\rho))$. Then, the strategies $\kappa(\sigma)$ and $\sigma$ are similar, thus, by the previous item, for all $W \in \msf{Borel}(Q)$, we have:
				\begin{equation*}
					\Prb_s^{\kappa(\sigma)}[\POMDP,W] = \Prb_{t_S^{\msf{init}}}^{\sigma}[\MEMDP(\POMDP)_S[s],W_s]
				\end{equation*}
				
				Conversely, consider any $\kappa \in \msf{Strat}(\Omega,A)$. Our goal is to define $\sigma(\kappa) \in \msf{Strat}(T_S,A)$ such that the strategies $\sigma(\kappa)$ and $\kappa$ are similar. Let us argue that the function $\msf{sc}^*_S\colon \Omega \cdot (A \times \Omega)^* \to T_S \cdot (A \cdot T_S)^*$ is injective on $\msf{Compatible}_S(A,\Omega)$. Indeed, consider any $\rho,\rho' \in \Omega \cdot (A \times \Omega)^* \cap \msf{Compatible}_S(A,\Omega)$ such that $\msf{sc}^*_S(\rho) = \msf{sc}^*_S(\rho')$. There is some $s \in S$ such that $\rho,\rho' \in \Omega \cdot (A \times \Omega)^* \cap \msf{Compatible}_s(A,\Omega)$. Hence, we have $\rho = O(\msf{sc}^*_S(\rho)[s]) = O(\msf{sc}^*_S(\rho')[s]) = \rho'$. Then, for all $\rho \in T_S \cdot (A \cdot T_S)^*$, there is at most one $\pi_\rho \in \msf{Compatible}_S(A,\Omega)$ such that $\rho = \msf{sc}^*_S(\pi_\rho)$, in which case we let $\sigma(\kappa)(\rho) := \kappa(\pi_\rho)$; otherwise $\sigma(\kappa)(\rho)$ is defined arbitrarily. That way, the strategies $\kappa$ and $\sigma(\kappa)$ are similar, thus, by the previous item, for all $W \in \msf{Borel}(Q)$, we have:
				\begin{equation*}
					\Prb_s^{\kappa}[\POMDP,W] = \Prb_{t_S^{\msf{init}}}^{\sigma(\kappa)}[\MEMDP(\POMDP)_S[s],W_s]
				\end{equation*}
				This proves the lemma.
			\end{enumerate}
		\end{enumerate}
	\end{proof}
	
	We can now proceed to the proof of Theorem~\ref{thm:Dirac-preserving-pomdp-are-memdps}.
	\begin{proof}
		From a Dirac-preserving POMDP $\POMDP = (Q,A,\delta,\Omega,O)$, a belief $b \in \Dist(Q)$, and an observation-compatible parity objective $W \in \msf{Borel}(Q)$, it is clear that we can build in exponential time the MEMDP $\MEMDP(\POMDP)_S = (Q',A,E,(\delta_e)_{e \in E})$ with $S := \Supp(b)$\footnote{Without loss of generality, we assume that that $\Supp(b)$ is observation-compatible, since we gather an observation on the initial state when playing for the first time in the POMDP.} and the parity objective $W_S$, with $|Q'|  = |T_S| \leq (n+1)^n$, for $n := |Q|$. We use $b$ itself as the initial environment belief in $\MEMDP(\POMDP)_S$.
		
		Furthermore, by Lemma~\ref{lem:memdp-from-pomdp}, for all $\sigma \in \msf{Strat}(T_S,A)$, for all $s \in S$, we have:
		\begin{align*}
			\Prb_{t_S^{\msf{init}}}^{\sigma}[\MEMDP(\POMDP)_S[s],W_s]
			& = \Prb_{t_S^{\msf{init}}}^{\sigma}[\MEMDP(\POMDP)_S[s],W_S \setminus \msf{Incomp}_s] = \Prb_{t_S^{\msf{init}}}^{\sigma}[\MEMDP(\POMDP)_S[s],W_S]
		\end{align*}
		
		In addition, again by Lemma~\ref{lem:memdp-from-pomdp}, for all $\kappa \in \msf{Strat}(\Omega,A)$, there is $\sigma(\kappa) \in \msf{Strat}(T_S,A)$ such that:
		\begin{align*}
			\sum_{s \in S} b(s) \cdot \Prb_{s}^\kappa[\POMDP,W] = \sum_{s \in S} b(s) \cdot  \Prb_{t_S^{\msf{init}}}^{\sigma(\kappa)}[\MEMDP(\POMDP)_S[s],W_S] \leq \msf{val}^{\msf{pr}}_{t_S^{\msf{init}}}(\MEMDP(\POMDP)_S,b,W_S)
		\end{align*}
		
		Furthermore, for all $\sigma \in \msf{Strat}(T_S,A)$, there is $\kappa(\sigma) \in \msf{Strat}(\Omega,A)$ such that:
		\begin{align*}
			\sum_{s \in S} b(s) \cdot \Prb_{t_S^{\msf{init}}}^{\sigma}[\MEMDP(\POMDP)_S[s],W_S] = 
			\sum_{s \in S} b(s) \cdot \Prb_{s}^{\kappa(\sigma)}[\POMDP,W] \leq \msf{val}(\POMDP,b,W)
		\end{align*}
		
		In fact, we have:
		\begin{equation*}
			\msf{val}(\POMDP,b,W) = \msf{val}^{\msf{pr}}_{t_S^{\msf{init}}}(\MEMDP(\POMDP)_S,b,W_S)
		\end{equation*}
		
		The theorem follows.
	\end{proof}
\end{document}